\documentclass[a4paper,USenglish]{article}

\let\mythebibliography\thebibliography
\let\myendthebibliography\endthebibliography

\usepackage[numbers]{natbib}
\let\thebibliography\mythebibliography
\let\endthebibliography\myendthebibliography

 \usepackage{latexsym,amsmath,amsfonts,amssymb}

\usepackage{amsthm}
\newtheorem{theorem}{Theorem}
\newtheorem{lemma}{Lemma}[theorem]

\usepackage{thm-restate}

\usepackage{tabularx,tabulary}
\usepackage[table,cmyk]{xcolor}
\usepackage{paralist}
\usepackage{xspace}

\usepackage{color}
\usepackage{colortbl}

\usepackage[font=scriptsize,labelfont=scriptsize,margin=0pt]{caption}
\usepackage{subfig}

\usepackage{tikz}

\usepackage{tikz}
\usepackage{fp}
\usetikzlibrary{arrows,positioning,shapes.misc,calc,shapes,shapes.multipart}
\usetikzlibrary{decorations.pathreplacing,decorations.pathmorphing,fit}
\usetikzlibrary{patterns}

\definecolor{matheonblue}{rgb}{0.15,0.35,0.65}
\definecolor{sidinggreen}{RGB}{201,253,0}
\colorlet{segmentyellow}{yellow!18}
\colorlet{redsegmentarea}{red!20!segmentyellow}
\colorlet{bluesegmentarea}{blue!45!segmentyellow}
\colorlet{blacksegmentarea}{black!60!segmentyellow}
\definecolor{waterblue}{rgb}{0.19,0.62,1}

\tikzset{segment/.style={color = segmentyellow}}
\tikzset{siding/.style={color = sidinggreen}}

\pgfdeclarelayer{bg}
\pgfsetlayers{bg,main}

\xdef\procheight{.5}
\xdef\timeunit{.5}

\xdef\procdist{1}

\xdef\xdelt{.06}

\newcommand{\updstripjob}[7]{
    \path[fill = #1, fill opacity = .3]
        (#2 - #7*\procdist + \xdelt,#3) -- (#2 + \xdelt -#7*\procdist ,#3 + #4) -- (#2,#3 + #4 + #5)
        -- (#2,#3 + #5) -- cycle;
    \path[draw = #1!60!black, very thick]
      (#2 -#7*\procdist + \xdelt ,#3) -- 
      (#2 -#7*\procdist + \xdelt ,#3 + #4);
    \path[draw = #1!60!black]
      (#2 + \xdelt -#7*\procdist ,#3 + #4) -- (#2,#3 + #4 + #5);

}

\newcommand{\downdstripjob}[7]{
    \path[fill = #1, fill opacity = .3]
        (#2 + #7*\procdist - \xdelt,#3) -- (#2 - \xdelt + #7*\procdist ,#3 + #4) -- (#2,#3 + #4 + #5)
        -- (#2,#3 + #5) -- cycle;
    \path[draw = #1!60!black, very thick]
      (#2 + #7*\procdist - \xdelt ,#3) -- 
      (#2 + #7*\procdist - \xdelt ,#3 + #4);
    \path[draw = #1!60!black]
      (#2 - \xdelt + #7*\procdist ,#3 + #4) -- (#2,#3 + #4 + #5);
}

\newcommand{\upjob}[6]{
  \updstripjob{#1}{#2}{#3}{#4}{#5}{#6}{1}
}
\newcommand{\downjob}[6]{
  \downdstripjob{#1}{#2}{#3}{#4}{#5}{#6}{1}
}

\newcommand{\gnode}[5]{
  \node[draw = #1!60!black, fill = #1, fill opacity = .3, circle]
          (#2) at (#3,#4){#5};
}

\usepackage{booktabs,multirow,array}

 \usepackage[pdftex,colorlinks,linkcolor=black,citecolor=black,
 urlcolor=black,pdfpagemode=UseOutlines,
 bookmarks,bookmarksopen,bookmarksdepth=2,plainpages=false]{hyperref}
\usepackage{ifthen}
\newboolean{ptas-more}
\setboolean{ptas-more}{false}

\newcommand{\subfigure}[2][]{\subfloat[#1]{#2}}

\newcommand{\myline}[1]{\raisebox{1.5mm}{\underline{\hspace{#1}}}}

\newcommand{\N}{\ensuremath{\mathbb{N}}}
\newcommand{\sat}{\textsc{Sat}\xspace}
\newcommand{\rsat}[2]{\ensuremath{(#1,#2)}-\sat}

\newcommand{\dir}{\ensuremath{\mathrm{d}}}
\newcommand{\rb}{\ensuremath{\mathrm{r}}}
\newcommand{\lb}{\ensuremath{\mathrm{l}}}
\newcommand{\dirb}[2]{\ensuremath{#1_{#2}}}
\newcommand{\rightb}[1]{\ensuremath{\dirb{#1}{\rb}}}
\newcommand{\leftb}[1]{\ensuremath{\dirb{#1}{\lb}}}
\newcommand{\set}[1]{\ensuremath{#1}}
\newcommand{\dirbset}[2]{\ensuremath{\set{#1}^{#2}}}
\newcommand{\rightbset}[1]{\ensuremath{\dirbset{#1}{\rb}}}
\newcommand{\leftbset}[1]{\ensuremath{\dirbset{#1}{\lb}}}

\newcommand{\jobs}{\ensuremath{\set{J}}}
\newcommand{\rightbjs}{\ensuremath{\rightbset{J}}}
\newcommand{\leftbjs}{\ensuremath{\leftbset{J}}}
\newcommand{\edges}[1]{\ensuremath{\set{E}_{#1}}}

\newcommand{\proc}[1]{\ensuremath{p_{#1}}}
\newcommand{\transit}[1]{\ensuremath{\tau_{#1}}}
\newcommand{\rel}[1]{\ensuremath{r_{#1}}}
\newcommand{\start}[1]{\ensuremath{S_{#1}}}
\newcommand{\compl}[1]{\ensuremath{C_{#1}}}
\newcommand{\cupdot}{\ensuremath{\mathaccent\cdot\cup}}

\newcommand{\cmax}{\ensuremath{\compl{\max}}}

\newcommand{\eps}{\ensuremath{\varepsilon}}
\newcommand{\cs}{\ensuremath{\sigma}}
\newcommand{\ctnum}{\ensuremath{\kappa}}

\newcommand{\NP}{\mathsf{NP}}
\newcommand{\APX}{\mathsf{APX}}

\newcommand{\tindex}[2]{\ensuremath{#1^{#2}}}
\newcommand{\tstart}[1]{\ensuremath{A_{#1}}}
\newcommand{\truejs}[1]{\ensuremath{\set{T}^{#1}}}
\newcommand{\falsejs}[1]{\ensuremath{\set{F}^{#1}}}
\newcommand{\indefjs}{\ensuremath{\set{Q}}}
\newcommand{\blockingjs}[1]{\ensuremath{\set{B}_{#1}}}
\newcommand{\dummyjs}[1]{\ensuremath{\set{H}_{#1}}}

\newcommand{\varn}{\ensuremath{|X|}}
\newcommand{\cln}{\ensuremath{|C|}}

\newcommand{\job}[2]{\ensuremath{#1_{#2}}}
\newcommand{\truejob}[2]{\job{t}{#1}^{#2}}
\newcommand{\falsejob}[2]{\job{f}{#1}^{#2}}
\newcommand{\indefjob}[2]{\job{q}{#1}^{#2}}
\newcommand{\blockingjob}[2]{\tindex{\job{b}{#1}}{#2}}
\newcommand{\dummyjob}[2]{\tindex{\job{h}{#1}}{#2}}

\newcommand{\mtrue}{\mathrm{t}}
\newcommand{\mfalse}{\mathrm{f}}

\newcommand{\idle}[2]{\ensuremath{\Xi[#1,#2]}}

\newcommand{\noun}[1]{$\textsc{#1}$}
\newcommand{\I}{\ensuremath{\mathcal{I}}\xspace}
\renewcommand{\O}{\ensuremath{\mathcal{O}}\xspace}

 \title{Scheduling Bidirectional Traffic on a Path}
 \author{Yann Disser \and Max Klimm\thanks{This research was carried out in the framework of {\sc Matheon} supported by Einstein Foundation Berlin.} \and Elisabeth L\"ubbecke$^*$}
\date{\small Department of Mathematics, Technische Universit\"at Berlin\\
\texttt{\{disser,klimm,eluebbecke\}@math.tu-berlin.de}}
\title{Scheduling Bidirectional Traffic on a Path}

\begin{document}

\maketitle
\begin{abstract}
We study the fundamental problem of scheduling bidirectional traffic along a path composed of multiple segments.
The main feature of the problem is that jobs traveling in the same direction can be scheduled in quick succession on a 
segment, while jobs in opposing directions cannot cross a segment at the same time. 
We show that this tradeoff makes the problem significantly harder than the related flow shop problem, by proving that it is $\NP$-hard even for identical jobs.
We complement this result with a PTAS for a single segment and non-identical jobs.
If we allow some pairs of jobs traveling in different directions to cross a segment concurrently, the  problem becomes $\APX$-hard even on a single segment and with identical jobs.
We give polynomial algorithms for the setting with restricted compatibilities between jobs on a single and any constant number of segments, respectively.
\end{abstract}


\section{Introduction}
The scheduling of bidirectional traffic on a path is essential when operating single-track infrastructures such as single-track railway lines, canals, or communication channels. Roughly speaking, the schedule governs when to move jobs from one node of the path to another along the segments of the path. The goal is to schedule all jobs such that the sum of their arrival times at their respective destinations is minimized.
A central feature of real-world single-track infrastructures is that after one job enters a segment of the path, further jobs moving in the \emph{same} direction can do so with relatively little headway, while traffic in the \emph{opposite} direction usually has to wait until the whole segment is empty again (cf.~Fig.~\ref{fig:canal}a for a schematic illustration).

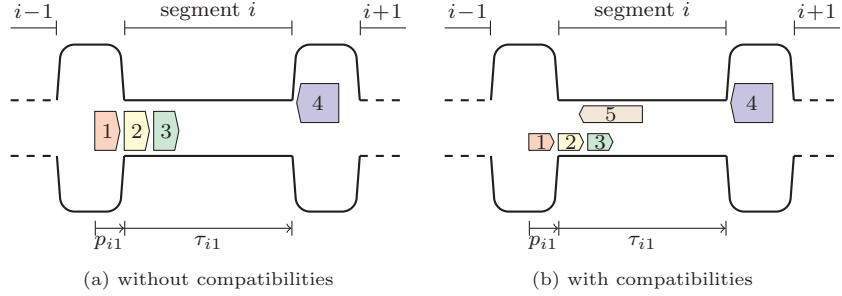
\begin{figure}[tb]
\centering
\subfigure[\scriptsize without compatibilities]{
\centering
\begin{tikzpicture}[x=0.615cm,y=0.82cm,scale=.9]
\footnotesize
\draw[thick,dashed] (-1.2,1) -- (-0.1,1);
\draw[thick,rounded corners=5] (-0.1,1) -- (0,0) -- (1.4,0) -- (1.5,1);
\draw[thick] (1.5,1) -- (5.5,1);
\draw[thick,rounded corners=5] (5.5,1) -- (5.6,0) -- (7.0,0) -- (7.1,1);
\draw[thick,dashed] (7.1,1) -- (8.2,1);
\draw[thick,dashed] (-1.2,2) -- (-0.1,2);
\draw[thick,rounded corners=5] (-0.1,2) -- (0,3) -- (1.4,3) -- (1.5,2);
\draw[thick] (1.5,2) -- (5.5,2);
\draw[thick,rounded corners=5] (5.5,2) -- (5.6,3) -- (7.0,3) -- (7.1,2);
\draw[thick,dashed] (7.1,2) -- (8.2,2);

\filldraw[fill=red!20] (0.8,1.1) -- (1.3,1.1) -- (1.4,1.45) -- (1.3,1.8) -- (0.8,1.8) -- (0.8,1.1);
\node at (1.1,1.45) {$1$};

\filldraw[fill=yellow!20] (1.5,1.1) -- (2.0,1.1) -- (2.1,1.45) -- (2.0,1.8) -- (1.5,1.8) -- (1.5,1.1);
\node at (1.8,1.45) {$2$};

\filldraw[fill=green!20] (2.2,1.1) -- (2.7,1.1) -- (2.8,1.45) -- (2.7,1.8) -- (2.2,1.8) -- (2.2,1.1);
\node at (2.5,1.45) {$3$};


\filldraw[fill=blue!20] (6.6,2.3) -- (5.7,2.3) -- (5.6,1.95) -- (5.7,1.6) -- (6.6,1.6) -- (6.6,2.3);
\node at (6.1, 1.95) {$4$};

\draw[|->|] (1.5,-0.3) -- node[below] {$\tau_{i1}$} (5.5,-0.3);

\draw[|->] (0.8,-0.3) -- node[below] {$p_{i1}$} (1.5,-0.3);

\draw[|-|] (1.5,3.3) -- node[above] {segment~$i$} (5.5,3.3);

\draw[-|] (-1.2,3.3) -- node[above] {$i\!-\!1$} (-0.1,3.3);

\draw[|-] ( 7.1,3.3) -- node[above] {$i\!+\!1$} (8.2,3.3);
\end{tikzpicture}
}
\hspace{0.05cm}
\subfigure[\scriptsize with compatibilities]{
\centering
\begin{tikzpicture}[x=0.615cm,y=0.82cm,scale=.9]
\footnotesize
\draw[thick,dashed] (-1.2,1) -- (-0.1,1);
\draw[thick,rounded corners=5] (-0.1,1) -- (0,0) -- (1.4,0) -- (1.5,1);
\draw[thick] (1.5,1) -- (5.5,1);
\draw[thick,rounded corners=5] (5.5,1) -- (5.6,0) -- (7.0,0) -- (7.1,1);
\draw[thick,dashed] (7.1,1) -- (8.2,1);
\draw[thick,dashed] (-1.2,2) -- (-0.1,2);
\draw[thick,rounded corners=5] (-0.1,2) -- (0,3) -- (1.4,3) -- (1.5,2);
\draw[thick] (1.5,2) -- (5.5,2);
\draw[thick,rounded corners=5] (5.5,2) -- (5.6,3) -- (7.0,3) -- (7.1,2);
\draw[thick,dashed] (7.1,2) -- (8.2,2);

\filldraw[fill=red!20] (0.8,1.1) -- (1.3,1.1) -- (1.4,1.25) -- (1.3,1.4) -- (0.8,1.4) -- (0.8,1.1);
\node at (1.1,1.25) {$1$};

\filldraw[fill=yellow!20] (1.5,1.1) -- (2.0,1.1) -- (2.1,1.25) -- (2.0,1.4) -- (1.5,1.4) -- (1.5,1.1);
\node at (1.8,1.25) {$2$};

\filldraw[fill=green!20] (2.2,1.1) -- (2.7,1.1) -- (2.8,1.25) -- (2.7,1.4) -- (2.2,1.4) -- (2.2,1.1);
\node at (2.5,1.25) {$3$};

\filldraw[fill=brown!20] (3.5,1.9) -- (2.1,1.9) -- (2.0,1.75) -- (2.1,1.6) -- (3.5,1.6) -- (3.5,1.9);
\node at (2.75, 1.75) {$5$};

\filldraw[fill=blue!20] (6.6,2.3) -- (5.7,2.3) -- (5.6,1.95) -- (5.7,1.6) -- (6.6,1.6) -- (6.6,2.3);
\node at (6.1, 1.95) {$4$};

\draw[|->|] (1.5,-0.3) -- node[below] {$\tau_{i1}$}  (5.5,-0.3);

\draw[|->] (0.8,-0.3) --  node[below] {$p_{i1}$} (1.5,-0.3);

\draw[|-|] (1.5,3.3) -- node[above] {segment~$i$} (5.5,3.3);

\draw[-|] (-1.2,3.3) -- node[above] {$i\!-\!1$} (-0.1,3.3);

\draw[|-] ( 7.1,3.3) -- node[above] {$i\!+\!1$} (8.2,3.3);
\end{tikzpicture}
}
\caption{\label{fig:canal}
Bidirectional scheduling of ship traffic through a canal, with and without compatibilities.
The processing time $p_{ij}$ of job~$j$ is the time needed to enter segment~$i$ with sufficient security headway, i.e., the delay before other jobs in the same direction may enter the segment.
The travel time~$\tau_{ij}$ is the time needed to traverse the entire segment once entered. 
In both (a) and (b), jobs~$1,2,3$ can enter the segment in quick succession, while job 4 has to wait until they left the segment.
In~(b), job 5 is compatible with jobs~$1,2,3$ so that they may cross concurrently. The time to cross turnouts is assumed to be negligible.\\[-1cm]
}
\end{figure}


Formally, in the bidirectional scheduling problem we are given a path of consecutive segments connected at nodes, and a set of jobs, each with a release date and a designated start and destination node. The time job~$j$ needs to traverse segment~$i$ is governed by two quantities: its \emph{processing time}~$p_{ij}$ and its \emph{transit time}~$\tau_{ij}$.
While the former prevents the segment from being used by any other job (running in \emph{either} direction), the latter only blocks the segment from being used by jobs running in \emph{opposite} direction. 
For example, this allows us to model settings with bidirectional train traffic on a railway line split into single-track segments that are connected by turnouts (cf.~\citet[Section~2]{LusbyLER2011}).
In this setting, jobs correspond to trains, the processing time of a job is the time needed for the train to fully enter the next segment, and the transit time is the time to traverse the segment (and entirely move into the next turnout).
While a train is entering a single-track segment of the line, no other train may do so.
The next train in the same direction can enter immediately afterwards, whereas trains in opposite direction have to wait until the segment is clear again in order to prevent a collision.

Fig.~\ref{fig:path-time} shows the path-time-diagram of a feasible schedule for two segments and four jobs.
Jobs are represented by parallelograms of the same color.
The processing time of a job on a segment is reflected by the height of the corresponding parallelogram, while the transit time is the remaining time ($y$-distance) to the lowest point of the parallelogram.
In a feasible schedule, jobs may not intersect, and, in particular, a job can only begin being processed at a segment once it has fully exited the previous segment.
Note that in the example it makes sense for the two rightbound jobs to switch order while waiting at the central node.

\begin{figure}[tb]
 \begin{center}
 ~\\[1cm]
   \begin{tikzpicture}[scale=.8]\label{pic:all-part1}
   \pgftransformxscale{.8*\procheight}
   \pgftransformyscale{-.3*\timeunit}
  


  \xdef\length{23}
  \xdef\transI{5}
  \xdef\transII{7}
   \draw  (-1-\transI,\length) -- (-1-\transI,0) -- 
          (-1,0) -- (-1,\length);
   \draw  (1+\transII,\length)  -- (1+\transII,0) --
          (1,0)  -- (1,\length);

  \node[anchor=south] (i1) at (-3.5, 0){$i=1$};
  \node[anchor=south] (i2) at (4.5, 0){$i=2$};

  \draw[->] (9,.75*\length) -- node[anchor=west]{time} +(0,.25*\length);
  

  \updstripjob{red}{-1}{0}{2}{\transI}{}{\transI}
  \updstripjob{blue}{-1}{2}{.5}{\transI}{}{\transI}

  \downdstripjob{green} {-6}{13}{4}{\transI}{}{\transI}
  \path (-1,13) -- node[font=\small,pos=.65]{$j$} (-1-\transI,17+\transI);

  \downdstripjob{black}{1}{0}{2}{\transII}{}{\transII}
  \downdstripjob{green}{1}{2}{4}{\transII}{}{\transII}
  \path (1+\transII,2+4) -- node[font=\small]{$j$} (1,2+\transII);

  \updstripjob{blue}{8}{13}{.5}{\transII}{}{\transII}
  \updstripjob{red}{8}{13.5}{2}{\transII}{}{\transII}

  \begin{scope}
    \draw (1+\transII,2) -- ++(1.6,0) node[font=\small,anchor=west] {$\rel{j}$} ;
    \draw (1+\transII,6) -- ++(.8,0);
    \draw (1+\transII,6+\transII) -- ++(.8,0);
    \draw[<->] (1.4+\transII,2) -- node[font=\small,anchor=west] {$\proc{2j}$} ++(0,4);
    \draw[<->] (1.4+\transII,6) -- node[font=\small,anchor=west] {$\transit{2j}$} ++(0,\transII);
    \draw[dashed] (-1-\transI,6+\transII) -- ++(2+\transI +\transII,0);    
    
    \draw[dashed] (-1-\transI,10+\transII) -- ++(\transI,0);   
    \draw (-1-\transI,6 +\transII) -- ++(-.8,0);
    \draw (-1-\transI,10 +\transII) -- ++(-.8,0);
    \draw[<->] (-1.4-\transI,6 +\transII) -- node[font=\small,anchor=east] {$\proc{1j}$} ++(0,4);
    \draw[<->] (-1.4-\transI,10 +\transII) -- node[font=\small,anchor=east] {$\transit{1j}$} ++(0,\transI);
    \draw (-1-\transI,10+\transI +\transII) -- ++(-1.6,0) node[font=\small,anchor=east] {$\compl{j}$};
  \end{scope}


\end{tikzpicture}\\[-1.5cm]
 \end{center}
 \caption{Representation of a schedule on two segments ($i=1,2$) and four jobs as a path-time-diagram. In this example, all jobs are processed immediately at their release date. Job~$j$ is released at time $r_j$ at the right end of segment~2 and needs to reach the left end of segment~1. Since it never has to wait, its completion time is smallest possible: $\smash{C_j = r_j + p_{2j} + \tau_{2j} + p_{1j} + \tau_{1j}}$.~\\[-.8cm]}
 \label{fig:path-time}
\end{figure}
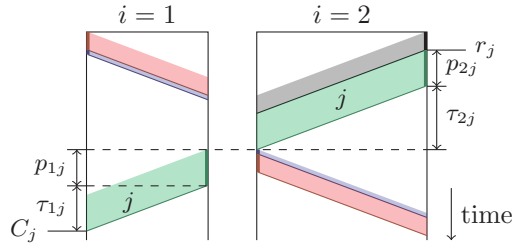

We also study a generalization of the model to situations where some of the jobs are allowed to pass each other when traveling in different directions (cf.~Fig.~\ref{fig:canal}b). This is a natural assumption, e.g., when scheduling the ship traffic on a canal, where smaller ships are allowed to pass each other while larger ships are not (cf.~\citet{LuebbLM2014}). 
In practice, the rules that decide which ships are allowed to pass each other are quite complex and depend on multiple parameters of the ships such as length, width, and draught (e.g., cf.~\cite{SeeSchStrO2013}). We model these complex rules in the most general way by a bipartite compatibility graph for each segment, where vertices correspond to jobs and two jobs running in different directions are connected by an edge if they can cross the segment concurrently. 

\subsubsection*{Our results.}
Table~\ref{tab:overview} gives a summary of our results.
We first show that scheduling bidirectional traffic is hard, even without processing times and with identical transit times~(Section~\ref{sec:unbounded}). 
The proof is via a non-standard reduction from \noun{MaxCut}. The key challenge is to use the local interaction of the jobs on the path to model global interaction between the vertices in the \noun{MaxCut}. We overcome this issue by introducing polynomially many vertex gadgets encoding the partition of each vertex and synchronizing these copies along the instance.
We complement this result with a polynomial time approximation scheme (PTAS) for a single segment and arbitrary processing times~(Section~\ref{sec:PTAS}) using the $(1+\epsilon)$-rounding technique of~\citet{Afrati1999}.

We then show that bidirectional scheduling with arbitrary compatibility graphs is $\APX$-hard already on a single segment and with identical processing times~(Section~\ref{sec:arbitrary-conflicts}). 
The proof is via a reduction from a variant of \noun{Max-3-Sat} which is $\NP$-hard to approximate within a factor smaller than 1016/1015, as shown by Berman et al.~\cite{BermanKS2003}. As a byproduct, we obtain that also minimizing the makespan is $\APX$-hard in this setting.
We again complement our hardness result by polynomial algorithms for identical jobs on constant numbers of segments and with a constant number of compatibility types~(Section~\ref{sec:constant_segments}).

\newcolumntype{P}[1]{>{\centering\arraybackslash}p{#1}}
\newcommand{\scs}{\scriptsize}

\subsubsection*{Significance.}

With this paper we initiate the mathematical study of optimized dispatching of traffic in networks with bidirectional edges, e.g.\ train networks, ship canals, communication channels, etc.
In all of these settings, traffic in one direction limits the possible throughput in the other direction.
While in the past decades a wealth of results has been established for the unidirectional case (i.e., classical scheduling, and, in particular, flow shop models), surprisingly, and despite their practical importance, bidirectional infrastructures have not received a similar attention so far.

The bidirectional scheduling model that we propose captures the essence of bidirectional traffic by distinguishing processing and transit times. 
This simple framework already allows to exhibit the computational key challenges of this setting. 
In particular, we show that bidirectional scheduling is already hard for identical jobs on a path, which is in contrast to the unidirectional case.
We observe another increase in complexity when allowing specific types of traffic to use an edge concurrently in both directions.
In practice, this is reasonable e.g.\ for ship traffic in a canal, where small vessels may pass each other.
In that sense, we show that scheduling ship traffic is already hard on a single edge and, thus, considerably harder than scheduling train traffic.

While bidirectional scheduling is hard in general, we show that certain features of real-world scenarios can make the problem tractable, e.g., a small number of turnouts along a single path and/or a small number of different vessels. 
In this work we restrict ourselves to simple paths, but we hope that our results are a first step towards understanding traffic in general bidirectional networks.

\begin{table}[tb]
\caption{Overview of our results for bidirectional scheduling.\newline
$^1$~even if $p=0$, $\transit{i} = 1$, $^2$~only if $\proc{}=1, \transit{i} \leq \text{const}$, $^3$~even if $\transit{i} = \proc{} = 1$.}
\label{tab:overview}

\begin{minipage}{\textwidth}
\begin{center}
\scriptsize
\begin{tabular*}{\linewidth}{
@{}l
@{\extracolsep{\fill}} P{2.cm}
@{\extracolsep{1ex}}   P{2.cm}
@{\extracolsep{1ex}}   P{2.cm}
}
\toprule
& \multicolumn{3}{c}{\myline{2.5cm} Number $m$ of segments \myline{2.5cm}}\\
compatibilities & $m=1$ & $m$ const. & $m$ arbitrary\\
\midrule

\multicolumn{3}{l}{\bf Different jobs $\proc{ij} = \proc{j}$, $\tau_{ij}=\tau_i$}\\
%
 & \multicolumn{1}{c}{\cellcolor{yellow!30}\scriptsize{PTAS }\scs{[Thm.~\ref{thm:ptas-oneseg-nocompatibilities}]}} & \multicolumn{1}{c}{\multirow{2}{*}{\cellcolor{red!30}}} & \multicolumn{1}{c}{\cellcolor{red!30}} \\
\multirow{-2}{*}{none/all compatible} & \multicolumn{2}{c}{\cellcolor{red!30}\cellcolor{red!30}$\NP$-hard \scs{\citep{LenstraKB1977}}} & \multicolumn{1}{c}{\multirow{-2}{*}{\cellcolor{red!30}$\NP$-hard$^1$~\scs{[Thm.~\ref{thm:hard_m}]}}} \\

 \\[-2ex]
 \midrule
 
\multicolumn{3}{l}{\bf Identical jobs $\proc{ij} = \proc{}$, $\tau_{ij}=\tau_i$}\\
none compatible & \multicolumn{1}{c}{\cellcolor{green!30}} & \multicolumn{1}{c}{\cellcolor{yellow!30}} & \multicolumn{1}{c}{\cellcolor{red!30}}\\[1ex]
const.\ \# types & \multicolumn{1}{c}{\multirow{-2}{*}{\cellcolor{green!30} polynomial \scs{[Thm.~\ref{thm:cct-singleseg-poly}]}}} &
\multicolumn{1}{c}{\multirow{-2}{*}{\cellcolor{yellow!30} polynomial$^2$~\scs{[Thm.~\ref{thm:identical-constantseg-makespan}]}}} &
\multicolumn{1}{c}{\multirow{-2}{*}{\cellcolor{red!30} $\NP$-hard$^1$~\scs{[Thm.~\ref{thm:hard_m}]}}}\\
\\[-2ex]
 arbitrary & \multicolumn{3}{c}{\cellcolor{red!30}$\APX$-hard$^3$~\scs{[Thm.~\ref{thm:graph-apx-hardness-sum}]}} \\
 \\[-2ex]
\bottomrule
\end{tabular*}
\end{center}
\end{minipage}
\end{table}

\subsubsection*{Related work.}
Scheduling problems are a fundamental class of optimization problems with a multitude of known hardness and approximation results (cf. \citet{LawlerLRS1993} for a survey). To the best of our knowledge, the bidirectional scheduling model that we propose and study in this paper has not been considered in the past nor is it contained as a special case in any other scheduling model. We give an overview of known results for related models.

For a single segment and jobs traveling from left to right, bidirectional scheduling reduces to the classical single machine scheduling problem, which \citet{LenstraKB1977} showed to be hard when minimizing total completion time. \citet{Afrati1999} gave a PTAS with generalizations to multiple identical or a constant number of unrelated machines. \citet{ChekuriKhanna2001} further generalized the result to related machines.
We give a different generalization for bidirectional scheduling.
For unrelated machines \citet{HoogeveenSW1998} showed that the completion time cannot be approximated efficiently within arbitrary precision, unless $\mathsf{P} = \NP$.

Bidirectional scheduling also has similarities to scheduling of two job families with a setup time that is required between jobs of different families. 
The general comments in \citet{PottsKov2000} on dynamic programs for such kinds of problems apply in part to our technique for Theorem~\ref{thm:cct-singleseg-poly}.

When all jobs need to be processed on all segments in the same order and all transit times are zero, bidirectional scheduling reduces to flow shop scheduling. \citet{GareyJS1976} showed that it is $\NP$-hard to minimize the sum of completion times in flow shop scheduling, even when there are only two machines and no release dates. They showed the same result for minimizing the makespan on three machines. 
\citet{HoogeveenSW1998} showed that there is no PTAS for flow shop scheduling without release dates, unless $\mathsf{P} = \NP$. 
In contrast, \citet{BrucknerKW2005} showed that flow shop problems with unit processing times can be solved efficiently, even when all jobs require a setup on the machines that can be performed by a single server only.

Job shop scheduling is a generalization of flow shop scheduling that allows jobs to require processing by the machines in any (not necessarily linear) order, cf.~\citet[Section~14]{LawlerLRS1993} for a survey. In this setting, the minimization of the sum of completion times was proven to even be $\mathsf{MAX}$-$\mathsf{SNP}$-hard by~\citet{HoogeveenSW1998}. \citet{QueyranneS00} gave a $\O((\log(m\mu)/\log\log(m\mu))^2)$-approximation for the weighted case with release dates, where~$\mu$ denotes the maximum number of operations per job. 
\citet{FishkinJM2003} gave a PTAS for a constant number of machines and operations per job.
It is worth noting that job shop scheduling does not contain bidirectional scheduling as a special case, since it does not incorporate the distinction between processing and transit times for jobs passing a machine in different directions. 

Job shop scheduling problems with unit jobs are strongly related to packet routing problems where general graphs are considered, see the discussion in seminal paper by \citet{LeightonMR1994}. They proved that the makespan of any packet routing problem is linear in two trivial lower bounds, called the congestion and the dilation. For more recent progress in this direction, see, e.g., \citet{Scheideler1998} and \citet{PeisWiese2011}. All these works, however, consider minimizing the makespan and assume that the orientation of the graph is fixed. \citet{AntonBCFMNP2014} also consider average flow time on a directed line. They give lower bounds for competitive ratios in the online setting and~$\O(1)$ competitive algorithms with resource augmentation for the maximum flow time.


\section{Preliminaries}
In the bidirectional scheduling problem, we are given a set $M = \{1,\dots,m\}$ of segments which we imagine to be ordered from left to right. Further, we are given two disjoint sets of $\rightbjs$ and $\leftbjs$ of \emph{rightbound} and \emph{leftbound} jobs, respectively, with $\jobs = \rightbjs \cup \leftbjs$ and $n = |\jobs|$.
Each job is associated with a \emph{release date} $\rel{j} \in \N$, a \emph{start segment} $s_j$ and a \emph{target segment} $t_j$, where $s_j \leq t_j$ for rightbound jobs and $s_j \geq t_j$ for leftbound jobs. A rightbound job $j$ needs to cross the segments $s_j, s_j+1,\dots,t_j-1,t_j$, and a leftbound job needs to cross the segments $s_j,s_j-1,\dots,t_j+1,t_j$. We denote by~$M_j$ the set of segments that job~$j$ needs to cross. Each job~$j$ is associated with a processing time $\proc{j} \in \N$ and each segment~$i$ is associated with a transit time $\transit{i} \in \N$. Note that we restrict ourselves to identical processing times for a single job and identical transit times for a single segment. We call $\proc{j} + \transit{i}$ the \emph{running time} of job~$j$ on segment~$i$. 

A \emph{schedule} is defined by fixing the start times $\start{ij}$ for each job $j$ on each segment $i \in M_j$. The \emph{completion time} of job $j$ on segment $i$ is then defined as $\compl{ij} = \start{ij} + \proc{j} + \transit{i}$. The overall completion time of job $j$ is $\compl{j} = \compl{t_jj}$. A schedule is feasible if it has the following properties.
\begin{compactenum}
\item Release dates are respected, i.e., $r_j \leq \start{s_jj}$ for each~$j \in \jobs$.

\item Jobs travel towards their destination, i.e., 
$\compl{ij} \leq \start{i+1,j}$ (resp. $\compl{ij} \leq \start{i-1,j}$) for rightbound (resp. leftbound) jobs $j$ and $i \in M_j \setminus \{t_j\}$.

\item Jobs $j, j'$ traveling in the same direction are not processed on segment~$i \in M_j \cap M_{j'}$ concurrently, 
i.e.,
$[\start{ij},\start{ij}+\proc{j})\cap[\start{ij'},\start{ij'}+\proc{j'}) = \emptyset$.

\item \label{it:opposing_jobs} Jobs~$j, j'$ traveling in different directions are neither processed nor in transit on segment~$i \in M_j \cap M_{j'}$ concurrently, i.e., $[\start{ij}, \compl{ij})\cap[\start{ij'}, \compl{ij'})=\emptyset$.
\end{compactenum}

 
%
%


Our objective is to minimize the \emph{total 
completion time}~$\sum C_j=\sum_{j\in \jobs} C_j$.

Other natural objectives are the minimization of the \emph{makespan}~$\cmax=\max\{C_j \mid j\in 
\jobs\}$ or the \emph{total waiting time}~$\sum W_j=\sum_{j\in \jobs} W_j$ where the 
individual waiting time of a job~$j$ is~$W_j = \compl{j} - \sum_{i \in M_j} (\proc{j} + \transit{i}) - r_j$. Note that minimizing the total waiting time is equivalent to minimizing the total completion time.



We also consider a generalization of the model, where some of the jobs traveling in different directions are allowed to pass each other. Formally, for each segment $i$, we are given a bipartite \emph{compatibility graph} $G_i = (\rightbjs \cupdot \leftbjs,\edges{i})$ with $\edges{i} \subseteq \rightbjs \times \leftbjs$. Two jobs $j,j'$ that are connected by an edge in~$G_i$ are allowed to run on segment~$i$ concurrently, i.e., condition  \ref{it:opposing_jobs} above need not be satisfied. Specifically, jobs~$j,j'$ may be processed or be in transit simultaneously.

All proofs omitted in the following sections can be found in the appendix.

\section{Hardness of bidirectional scheduling}\label{sec:unbounded}

First, we show that scheduling bidirectional traffic is hard, even when all processing times are zero and all transit times coincide. 
In other words, we eliminate all interaction between jobs in the same direction and show that hardness is merely due to the decision when to switch between left- and rightbound operation of each segment. This is in contrast to one-directional (flow shop) scheduling with identical processing times, which is trivial.
Formally, we show the following result.

\begin{restatable}{theorem}{unbounded}
\label{thm:hard_m}
The bidirectional scheduling problem is $\NP$-hard even if~$p_j=0$ and~$\tau_i=1$ for each~$j\in\jobs$ and $i\in M$.\label{thm:unbounded_hardness}
\end{restatable}

We reduce from the \noun{MaxCut} problem which is contained in Karp's list of~21 $\NP$-complete problems~\cite{Karp1972}. Given an undirected graph $G=(V,E)$ and some $k\in\mathbb{N}$ we ask for a partition $V=V_{1} \cupdot V_{2}$ with $|E\cap(V_{1}\times V_{2})|\geq k$.


For a considered instance \I of \noun{MaxCut} we construct an instance of the bidirectional scheduling problem which can be scheduled without exceeding some specific waiting time if and only if \I admits a solution. The translation to sum of completion times is then straightforward. 


\begin{figure}[t]
\centering
\includegraphics[width=.3\textwidth]{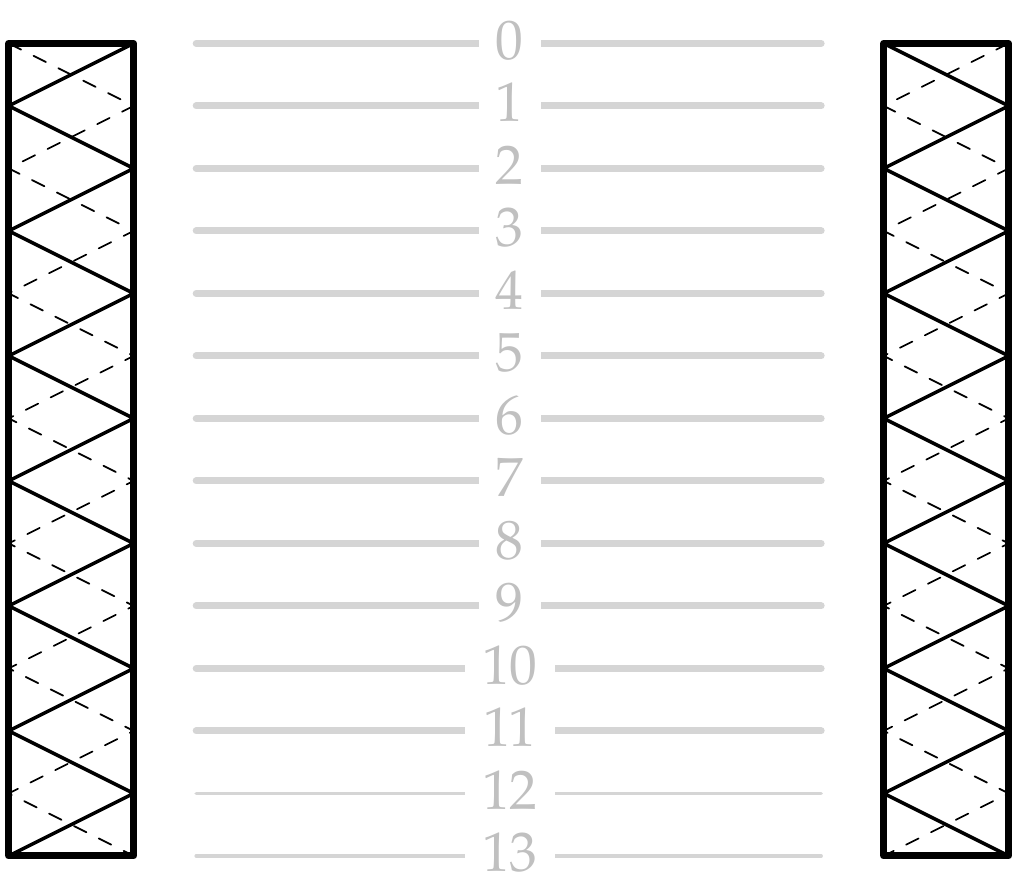}
\protect \caption { 
Illustration of the vertex gadget in the leftbound (left) and the rightbound (right) state. 
At each time~$t=0,\dots,11$ multiple right- and leftbound jobs are released.
Since all jobs have processing time 0, jobs in the same direction can be processed simultaneously.
The only two sensible schedules differ in whether leftbound jobs are processed at even or odd times.
\label{fig:vertex_gadget}\\[-0.8cm] }
\end{figure}

A cornerstone of our construction is the \emph{vertex gadget} that occupies a fixed time interval on a single segment and can only be (sensibly) scheduled in two ways~(cf.~Fig.~\ref{fig:unbounded_sketch}), which we interpret as the choice whether to put the corresponding vertex in the first or second part of the partition, respectively. We introduce multiple \emph{vertex segments} that each have exactly one vertex gadget for each vertex in \I and add further gadgets that ensure that the state of all vertex gadgets for the same vertex is the same across all segments. These gadgets allow us to synchronize vertex gadgets on consecutive vertex segments in two ways. We can either simply synchronize vertex gadgets that occupy the same time interval on the two vertex segments (\emph{copy gadget}), or we can synchronize pairs of vertex gadgets occupying the same consecutive time intervals on the two vertex segments by linking the first gadget on the first segment with the second one on the second segment and vice-versa, i.e., we can transpose the order 
of two consecutive gadgets from one vertex segment to the next (\emph{transposition gadget}).

We construct an edge gadget for each edge in \I that incurs a small waiting time if two vertex gadgets in consecutive time intervals and segments are in different states and a slightly higher waiting time if they are in the same state. By tuning the multiplicity of each job, we can ensure that only schedules make sense where vertex gadgets are scheduled consistently. Minimizing the waiting time then corresponds to maximizing the number of edge gadgets that link vertex gadgets in different states, i.e., maximizing the size of a cut.

In order to fully encode the given \noun{MaxCut} instance \I, we need to introduce an edge gadget for each edge in \I. However, edge gadgets can only link vertex gadgets in consecutive time intervals. We can overcome this limitation by adding a sequence of vertex segments and transposing the order of two vertex gadgets from one segment to the next as described before. With a linear number of vertex segments we can reach an order  where the two vertex gadgets we would like to connect with an edge gadget are adjacent. At that point, we can add the edge gadget, and then repeat the process for all other edges in \I (cf.~Fig.~\ref{fig:unbounded_sketch}).

We can reformulate Theorem~\ref{thm:unbounded_hardness} for nonzero processing times, simply by making the transit time large enough that the processing time does not matter.


\begin{restatable}{corollary}{corUnboundedWithProcessing}
  The bidirectional scheduling problem is $\NP$-hard even if~$p_j=1$ and~$\transit{i}=\transit{}$ for each~$j\in\jobs$ and $i\in M$.
\end{restatable}

\begin{figure}[tb]
\centering
\includegraphics[width=.5\textwidth]{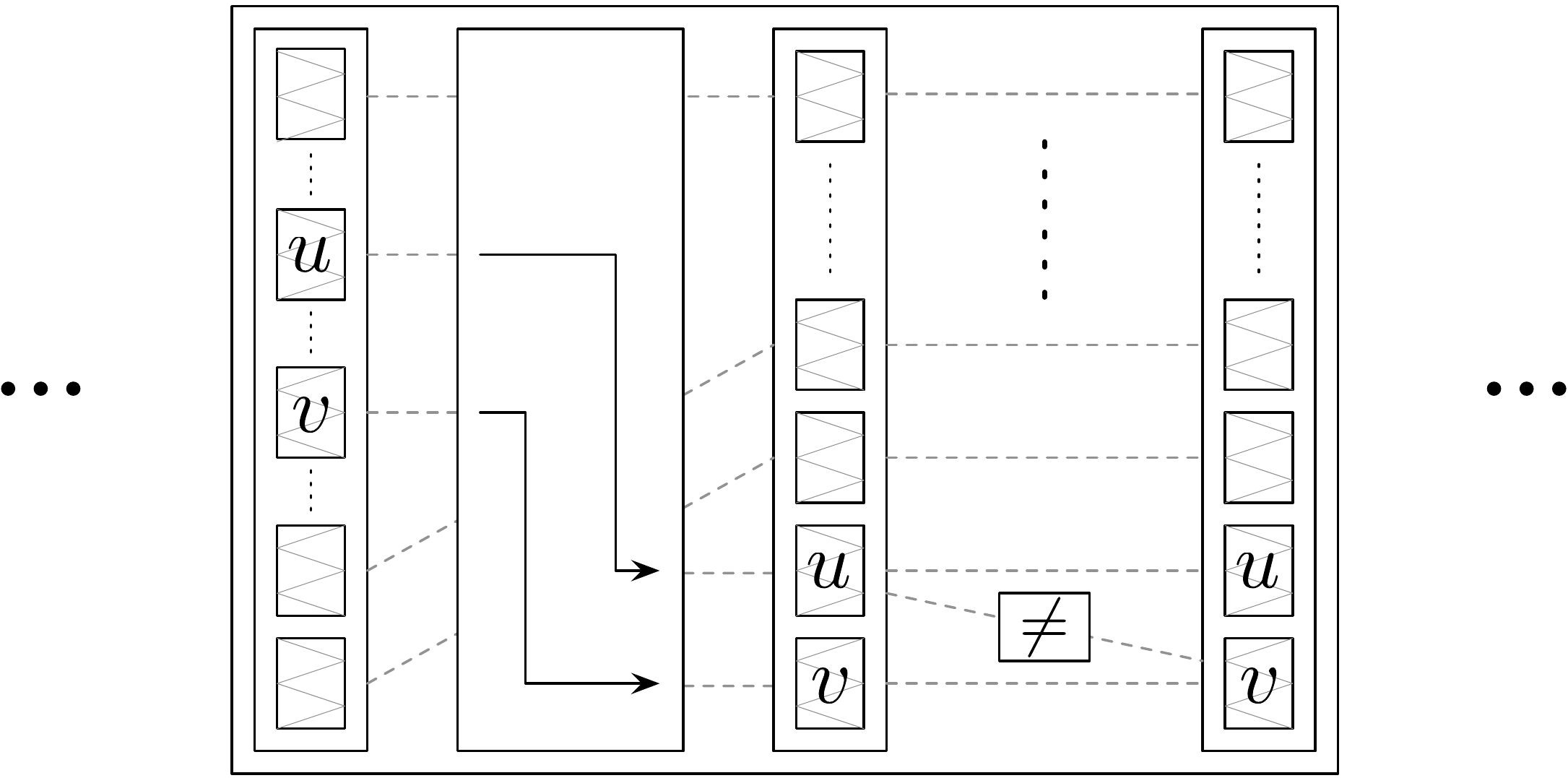}
\protect
\caption
{
  Illustration of our hardness construction for a single edge~$e=\{u,v\}$. First, a sequence of segments is used to change the order of vertex gadgets, such that the vertex gadgets corresponding to~$u$ and~$v$ occupy consecutive time intervals. Then, an edge gadget is added that incurs an increased waiting time if the vertex gadgets for~$u$ and~$v$ are in the same state.\label{fig:unbounded_sketch}\\[-0.8cm]
}
\end{figure}

\section{A PTAS for bidirectional scheduling}
\label{sec:PTAS}

We give a polynomial time approximation scheme (PTAS), i.e., a polynomial $(1+\eps)$-approximation algorithm for each~$\eps>0$, for bidirectional scheduling on a single segment with general processing times.
This problem is hard even if all jobs have the same direction~\citep{LenstraKB1977}. We extend the machine scheduling PTAS of \citet{Afrati1999} to the bidirectional case, provided that the jobs are either all pairwise in conflict or pairwise compatible.
The main issue when trying to adopt the technique of~\cite{Afrati1999} is to account for the different roles of processing and transit times for the interaction of jobs in the same and different directions.

\begin{restatable}{theorem}{thmptasOnesegNoCompatibilities}
\label{thm:ptas-oneseg-nocompatibilities}
The bidirectional scheduling problem on a single segment and with compatibility graph $G_1\in\{K_{\rightb{n},\leftb{n}},\emptyset\}$ admits a PTAS.
\end{restatable}

The first part of the proof in~\cite{Afrati1999} is to restrict to processing times and release dates of the form~$(1+\eps)^x$ for some~$x\in\N$ and~$r_j\geq\eps(p_j+\tau_1)$.
Allowing fractional processing and release times we can show that any instance can be adapted to have these properties, without making the resulting schedule worse by a factor of more than~$(1+\eps)$.
We may thus partition the time horizon into intervals~$I_x = [(1+\eps)^x,(1+\eps)^{x+1}]$, such that every job is released at the beginning of an interval.
Since jobs are not released too early, we may conclude that the maximum number of intervals~$\cs$ covered by the running time of a single job is constant.
This allows us to group intervals together in blocks~$B_t = \{I_{t\cs}, I_{t\cs+1},\dots, I_{(t+1)\cs-1}\}$ of~$\cs$ intervals each, such that every job scheduled to start in block~$B_t$ will terminate before the end of the next block~$B_{t+1}$.

To use the fact that each block only interacts with the next block in our dynamic program,
we need to specify an interface for this interaction.
For that purpose we introduce the notion of a \emph{frontier}.
A block \emph{respects an incoming frontier}~$F=(f_{\lb},f_{\rb})$ if no leftbound (rightbound) job scheduled to start in the block starts earlier than~$f_{\lb}$ ($f_{\rb}$).
Similarly, a block \emph{respects an outgoing frontier}~$F=(f_{\lb},f_{\rb})$ if no leftbound or rightbound job scheduled to start in the block would interfere with a leftbound (rightbound) job starting at time~$f_{\lb}$ ($f_{\rb}$).
The symmetrical structure of the compatibility graph ($K_{\rightb{n},\leftb{n}}$ or $\emptyset$) allows us to use this simple interface.
We introduce a dynamic programming table with entries~$T[t,F,U]$ that are designed to hold the minimum total completion time of scheduling all jobs in~$U\subseteq \jobs$ to start in block~$B_t$ or earlier, such that~$B_t$ respects the outgoing frontier~$F$.
We define~$C(t,F_1,F_2,V)$ to be the minimum total completion time of scheduling all jobs in~$V$ to start in $B_t$ with $B_t$ respecting the incoming frontier~$F_1$ and the outgoing frontier~$F_2$ (and~$\infty$ if this is impossible).
We have the following recursive formula for the dynamic programming table:
\[
  T[t,F,U] = \min\nolimits_{F',V \subseteq U}\{T[t-1,F',U \setminus V]+C(t,F',F,V)\}.
\]

To turn this into an efficient dynamic program, we need to limit the dependencies of each entry and show that~$C(\cdot)$ can be computed efficiently. The number of blocks to be considered can be polynomially bounded by~$\log D$, where~$D=\max_j r_j + n\cdot(\max_j p_j + \tau_1)$ is an upper bound on the makespan.
The following lemma shows that we only need to consider polynomially many other entries to compute~$T[t,F,U]$ and we only need to evaluate~$C(\cdot)$ for job sets of constant size, which we can do in polynomial time by simple enumeration. 

\begin{restatable}{lemma}{lemConstantInterface}
\label{lem:constant_interface}
  There is a schedule with a sum of completion times within a factor of~$(1+\eps)$ of the optimum and with the following properties:
  \begin{compactenum}
    \item The number of jobs scheduled in each block is bounded by a constant.
    \item Every two consecutive blocks respect one of constantly many frontiers.
  \end{compactenum}
\end{restatable}

\begin{proof}[sketch]
  Partitioning the released jobs of each interval direction-wise by processing time into \emph{small} and \emph{large} jobs and bundling small jobs into packages of roughly the same size allows us to bound the number of released jobs per interval by a constant, similarly as in~\cite{Afrati1999}. Furthermore, we establish that we may assume jobs to remain unscheduled only for constantly many blocks.
  
  For the second property, we stretch all time intervals by a factor of~$(1+\eps)$, which gives enough room to decrease the start times of those jobs interfering with two blocks such that an~$1/\eps^2$-fraction of an interval separates jobs starting in two consecutive blocks. Thus, we only need to consider~$\frac{\cs}{\eps^2}$ possible frontier values per direction, or a total of~$\smash{\bigl(\frac{\cs}{\eps^2}\bigr)^2}$ possible frontiers.
\end{proof}

\ifthenelse{\boolean{ptas-more}}{
We now generalize our dynamic program to a constant number of segments assuming that the transit times of any two segments differ by at most a constant factor. To this end, we split our jobs into parts, one for each segment the job needs to be processed on, with the additional constraint that no part may be scheduled before any part of the same job on earlier segments.
We are able to generalize Lemma~\ref{lem:constant_interface} to this setting, using that each part of a job runs in at most two blocks and partitioning jobs into small and large for each direction and combination of start and target segments.
The interface between consecutive time blocks needs to be extended to a frontier on each segment.
In addition, a part running in block~$B_t$ imposes a lower bound on the start time of the next part of the same job running in block~$B_{t+1}$.
Since the number of parts running in block~$B_t$ is bounded by a constant~$b$, the interface still has constant size.
We assume that jobs are ordered and write~$\vec F = (F_1,\dots,F_m)$, $\vec{\theta} = (\theta_1,\dots,\theta_b)$. 
We can define our table with entries~$T[t,\vec{F},U,V,\vec{\theta}]$ containing the minimum sum of  completion times (of completely finished jobs) when scheduling the parts in $U$ to start in block~$B_t$ or earlier, such that: $B_t$ respects the outgoing frontier~$F_i$ on segment~$i$, the parts in $V \subseteq U$ are scheduled to start in block~$B_t$, and the~$l$-th part in~$V$ stops running by time~$\theta_l$.
Similarly,~$C(t,\vec{F}',\vec{F},V,\vec\theta',\vec\theta)$ is the minimum sum of completion times for scheduling the parts in~$V$ in block~$B_t$, respecting frontiers~$\vec{F}',\vec{F}$ on the segments, such that the~$l$-th part in $V$ does not start running before time~$\theta_l'$ and stops running by time~$\theta_l$ (if possible, and~$\infty$ otherwise).
The recursive formula restricted to subsets that respect the order in which parts need to be processed becomes
\[
  T[t,\vec{F},U,V,\vec{\theta}] = \min_{\substack{\vec F',V' \subseteq U\setminus V,\vec\theta' \\ |V'|\textrm{ is consistent}}}\{T[t-1,\vec{F}',U \setminus V, V',\vec\theta'] + C(t,\vec F',\vec F,V,\vec\theta',\vec\theta)\}.
\]

We obtain the following result.\enlargethispage{1ex}

\begin{theorem}
\label{thm:ptas}
The bidirectional scheduling problem on a constant number of segments and with compatibility graphs~$G_i\in\{K_{\rightb{n},\leftb{n}},\emptyset\}$ for each $i\in M$ admits a PTAS assuming that the transit times of any two segments differ by at most a constant factor.
\end{theorem}
}{}

\section{Hardness of custom compatibilities}
\label{sec:arbitrary-conflicts}
\newcommand{\tpart}[1]{\ensuremath{P_{#1}}}

In Section~\ref{sec:unbounded}, we showed that bidirectional scheduling is hard  on an unbounded number of machines, even for identical jobs.
As the main result of this section, we show that for arbitrary compatibility graphs the problem is $\APX$-hard already on a single segment and with unit processing and transit times. 
For ease of exposition, we first show that the minimization of the makespan is $\NP$-hard. Later we extend this result towards minimum completion time and $\APX$-hardness.

\begin{restatable}{theorem}{corGraphHardnessSum}
\label{thm:graph-hardness-sum}
The bidirectional scheduling problem on a single segment and with an arbitrary compatibility graph is $\NP$-hard even if~$\proc{j}=\transit{1}=1$ for each~$j\in\jobs$.
\end{restatable}

We give a reduction from an $\NP$-hard variant of \sat (cf.~\cite{GareyJohnson1979}); \rsat{\leq\!3}{3} considers a formula with a set of clauses~$\set{C}$ of size three over a set of variables $\set{X}$, where each variable appears in at most three clauses and asks if there is a truth assignment of~$\set{X}$ satisfying~$\set{C}$. Note the difference to the polynomially solvable \rsat{3}{3}, where each variable appears in \emph{exactly} three clauses~\cite{Tovey1984}.


For a given~\rsat{\leq\!3}{3} formula we construct a bidirectional scheduling instance that can be scheduled within some specific makespan~$T$ if and only if the given formula is satisfiable. Our construction is best explained by partitioning the time horizon~$[0,T]$ into four parts (cf.~Fig.~\ref{fig:graph-hardness} along with the following).

\begin{figure}[t]
  \centering
    \subfigure[\scriptsize $\tpart{1}$ and $\tpart{2}$: variable assignment]{
\begin{tikzpicture}
\footnotesize
   \pgftransformxscale{.6*\procheight}
   \pgftransformyscale{-.4*\timeunit}


  \begin{scope}[xshift=-3.5cm]

    \draw  (-1,0) -- +(0,14) (1,0) -- +(0,14);
    \draw[draw=black!30](0,0) -- +(0,14);
    

    \begin{scope}
      \clip (-4.5,0) rectangle (2.7,14.);
      \upjob{black}{0}{-1}{1}{1}{}
      \downjob{black} {0}{-1}{1}{1}{}
    \end{scope}

    \upjob{green} {-1.5}{0}{1}{1}{}
    \upjob{green} {-1.5}{1}{1}{1}{}
    \upjob{black} {0}{2}{1}{1}{}
    \upjob{red}   {0}{3}{1}{1}{}
    \upjob{red}   {0}{4}{1}{1}{}
    \upjob{black} {0}{5}{1}{1}{}


    \downjob{black} {0}{0}{1}{1}{}
    \downjob{black} {0}{2}{1}{1}{}
    \downjob{red}   {0}{1}{1}{1}{}
    \downjob{black} {0}{3}{1}{1}{}
    \downjob{black} {0}{5}{1}{1}{}
    \downjob{green} {1.5}{4}{1}{1}

    \node[anchor=east] (xi) at (-2.2, 3){$x_i$};

    \draw[dashed] (-3,6) -- (3,6);


    \upjob{green} {0}{6+0}{1}{1}{}
    \upjob{green} {0}{6+1}{1}{1}{}
    \upjob{black} {0}{6+2}{1}{1}{}
    \upjob{red}   {-1.5}{6+3}{1}{1}{}
    \upjob{red}   {-1.5}{6+4}{1}{1}{}
    \upjob{black} {0}{6+5}{1}{1}{}

    \downjob{black} {0}{6+0}{1}{1}{}
    \downjob{black} {0}{6+2}{1}{1}{}
    \downjob{red}   {1.5}{6+1}{1}{1}{}
    \downjob{black} {0}{6+3}{1}{1}{}
    \downjob{black} {0}{6+5}{1}{1}{}
    \downjob{green} {0}{6+4}{1}{1}{}

    \node[anchor=east] (xip1) at (-2.2, 8.5){$x_{i+1}$};
    \draw[dashed] (-3,12) -- (3,12);

  \end{scope}
  

  \begin{scope}[xshift=3.5cm]\label{pic:all-part2}


    \draw  (-1,0) -- +(0,14) (1,0) -- +(0,14);
    \draw[draw=black!30](0,0) -- +(0,14);

    \upjob{black}   {0}{0}{1}{1}{}
    \upjob{black}   {0}{1}{1}{1}{}
    \downjob{black} {0}{1}{1}{1}{}
    \downjob{black} {0}{0}{1}{1}{}

    \upjob{black}   {0}{2}{1}{1}{}
    \upjob{black}   {0}{3}{1}{1}{}
    \downjob{black} {0}{3}{1}{1}{}
    \downjob{green} {0}{2}{1}{1}{}

    \node[anchor=west] (xi) at (1.5, 2){$x_{i}$};

    \draw[dashed] (-1.5,4) -- (1.5,4);

    \upjob{black}   {0}{4+0}{1}{1}{}
    \upjob{black}   {0}{4+1}{1}{1}{}
    \downjob{black} {0}{4+1}{1}{1}{}
    \downjob{black} {0}{4+0}{1}{1}{}

    \upjob{black}   {0}{4+2}{1}{1}{}
    \upjob{black}   {0}{4+3}{1}{1}{}
    \downjob{black} {0}{4+3}{1}{1}{}
    \downjob{red}   {0}{4+2}{1}{1}{}

    \node[anchor=west] (xip1) at (1.5, 6){$x_{i+1}$};
    \draw[dashed] (-1.5,8) -- (1.5,8);

    \upjob{black}   {0}{8+0}{1}{1}{}
    \upjob{black}   {0}{8+1}{1}{1}{}
    \downjob{black} {0}{8+1}{1}{1}{}
    \downjob{black} {0}{8+0}{1}{1}{}

    \upjob{black}   {0}{8+2}{1}{1}{}
    \upjob{black}   {0}{8+3}{1}{1}{}
    \downjob{black} {0}{8+3}{1}{1}{}

    \node[anchor=west] (dots) at (1.5, 9.8){$\,\,\vdots$};
    \draw[dashed] (-1.5,12) -- (1.5,12);
  \end{scope}
\end{tikzpicture}
}
\hspace*{.2cm}
\subfigure[\scriptsize $\tpart{3}$: clauses]{
\begin{tikzpicture}\label{pic:all-part3}
   \pgftransformxscale{.6*\procheight}
   \pgftransformyscale{-.4*\timeunit}

  \draw[white] (-3.5,0) -- (3.5,0);

   \draw  (-1,0) -- +(0,14) (1,0) -- +(0,14);
   \draw[draw=black!30](0,0) -- +(0,14);

  \upjob{green} {0}{0}{1}{1}{}
  \upjob{black} {0}{1}{1}{1}{}
  \downjob{black} {0}{0}{1}{1}{}
  \downjob{black} {0}{1}{1}{1}{}

  \node[anchor=west] (ck) at (1.5, 1){$c_{k}$};
  \draw[dashed] (-1.5,2) -- (1.5,2);

  \upjob{red}   {0}{2}{1}{1}{}
  \upjob{black} {0}{3}{1}{1}{}
  \downjob{black} {0}{2}{1}{1}{}
  \downjob{black} {0}{3}{1}{1}{}

  \node[anchor=west] (ckp1) at (1.5, 3){$c_{k+1}$};
  \draw[dashed] (-1.5,4) -- (1.5,4);

  \upjob{black}   {0}{4+1}{1}{1}{}
  \downjob{black} {0}{4+0}{1}{1}{}
  \downjob{black} {0}{4+1}{1}{1}{}

  \node[anchor=west] (dots) at (1.5, 5){$\,\vdots$};
  \draw[dashed] (-1.5,6) -- (1.5,6);

  \upjob{black}   {0}{4+3}{1}{1}{}
  \downjob{black} {0}{4+2}{1}{1}{}
  \downjob{black} {0}{4+3}{1}{1}{}

  \draw[dashed] (-1.5,8) -- (1.5,8);

  \fill[white] (-1.5,9.5) rectangle (1.5,14);
  \node[anchor=north] (dots2) at (0, 9.5){$\,\vdots$};

\end{tikzpicture}
}
\hspace*{.2cm}
\subfigure[\scriptsize $\tpart{4}$: leftover jobs]{
\begin{tikzpicture}\label{pic:all-part4}
   \pgftransformxscale{.6*\procheight}
   \pgftransformyscale{-.4*\timeunit}


  \draw[white] (-3.5,0) -- (3.5,0);
  \draw  (-1,0) -- +(0,14) (1,0) -- +(0,14);
  \draw[draw=black!30](0,0) -- +(0,14);

  \upjob{red}   {0}{0}{1}{1}{}
  \upjob{green} {0}{1}{1}{1}{}
  \upjob{green} {0}{2}{1}{1}{}

  \downjob{black} {0}{0}{1}{1}{}
  \downjob{black} {0}{1}{1}{1}{}
  \downjob{black} {0}{2}{1}{1}{}

  \downjob{black} {0}{3}{1}{1}{}
  \downjob{black} {0}{4}{1}{1}{}
  \downjob{black} {0}{5}{1}{1}{}

  \fill[white] (-1.5,7.5) rectangle (1.5,14);
  \node[anchor=north] (dots2) at (0, 7.5){$\,\vdots$};

\end{tikzpicture}
}
  \caption{Illustration (colored) of the four parts of our construction. Time is directed downwards, rightbound (leftbound) jobs are depicted on the left (right) of each figure.\\[-.6cm]}
  \label{fig:graph-hardness}
\end{figure}
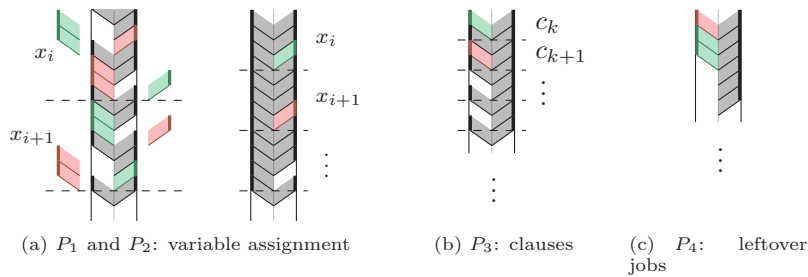

We use a frame of blocking jobs that need to be scheduled at their release date. We can enforce this by making sure that at least one blocking job is released at (almost) each unit time step and that blocking jobs that are not supposed to run concurrently are incompatible.
We release variable jobs that have to be scheduled into gaps between the blocking jobs.
More precisely, in the first part of the construction we release 6~jobs within a separate time interval for each variable. 
Two of these jobs are leftbound and need to be scheduled within the first two parts of the construction, which implies that one of the two remaining pairs of rightbound jobs must be scheduled after the second part.
If the first pair is delayed we interpret this as an assignment of \emph{true} to the variable and otherwise as \emph{false}.

The third part of the construction has a gap for each clause, with compatibilities ensuring that only variable jobs can be scheduled into the gap which satisfy the clause.
Since each literal can only appear in at most two clauses, there are enough variable jobs to satisfy all clauses if the formula is satisfied.
Finally, the last part has~$2|X|-|C|$ gaps that fit any variable job.
In order to schedule all variable jobs before the end of the last part, we thus need to schedule a variable job into each gap of a clause.
This is possible if and only if the given~\rsat{\leq\!3}{3} formula is satisfiable.
We can easily extend our result to completion or waiting times by adding many blocking jobs after the last part, such that violating the makespan also ruins the the total completion time.

With a slight adaption of the construction and more involved arguments, we can even show $\APX$-hardness of the problem. We reduce from a specific variant of \noun{Max-3-Sat}, where each literal occurs exactly twice, and which is $\NP$-hard to approximate within a factor of $1016/1015$, see Berman et al.~\cite{BermanKS2003}. 

\begin{restatable}{theorem}{GraphAPXHardnessSum}
\label{thm:graph-apx-hardness-sum}
The bidirectional scheduling problem on a single segment and with an arbitrary compatibility graph is $\APX$-hard even if~$\proc{j}=\transit{1}=1$ for each~$j\in\jobs$.
\end{restatable}

\section{Dynamic programs for restricted compatibilities}
\label{sec:constant_segments}
After establishing the hardness of bidirectional scheduling with a general compatibility graph in the last section, in this section we turn to the case of a constant number of different compatibility types. 
Due to the identical processing times, the jobs in each direction can be scheduled in the order of their release dates. The only decision left is when to switch between left- and rightbound operation of the segments. This decision is hard in the general case~(Theorem~\ref{thm:unbounded_hardness}), but we are able to formulate a dynamic program for any constant number of segments. 

Our result generalizes to the case when some jobs of different directions are compatible
as long as the number of \emph{compatibility types} is constant, where two jobs~$j_1, j_2$ in the same direction are defined to have the same compatibility type if the set of jobs compatible with~$j_1$ is equal to the set of jobs  compatible with~$j_2$ on each segment. Formally,~$j_1$ and~$j_2$ have the same compatibility type if $\bigl\{j : \{j_1,j\} \in E_i\bigr\} = \bigl\{j : \{j_2,j\} \in E_i\bigr\}$ for the compatibility graphs $G_i = (\leftbjs \cupdot \rightbjs, E_i)$ of each segment~$i$.

For a single segment we partition~$\jobs$ into~$\ctnum$ subsets of jobs~$\jobs^1, \dots, \jobs^{\ctnum}$ where all jobs of~$J^c$, $c\in 1,\dots, \ctnum$, have the same compatibility type~$c$, and let $n_c = |J^c|$. Since the jobs of each subset only differ in their release dates, they can again be scheduled in the order of their release dates. This observation allows us to define a dynamic program that decides how to merge the job sets~$\jobs^1, \dots, \jobs^{\ctnum}$ such that the resulting schedule has minimum total completion time.

\begin{restatable}{theorem}{singleSegPoly}
\label{thm:cct-singleseg-poly}
The bidirectional scheduling problem can be solved in polynomial time if~$m=1$,~$\ctnum$ is constant and~$\proc{j}=\proc{}$ for each~$j\in\jobs$.
\end{restatable}

We now consider a constant number of segments~$m>1$. The main complication in this setting is that decisions on one segment can influence decisions on other segments, and, in general, every job can influence every other job in this way. In particular, we need to keep track of how many jobs of each type are in transit at each segment, and we can thus not easily adapt the dynamic program for a single segment. We propose a different dynamic program that relies on all transit times being bounded by a constant and can be adapted for assumptions  complementary to Theorem~\ref{thm:unbounded_hardness}.

\begin{restatable}{theorem}{thmIdenticalConstantSegMakespan}
\label{thm:identical-constantseg-makespan}
The bidirectional scheduling problem can be solved in polynomial time if~$m$ and~$\ctnum$ are constant and either~$\proc{j}=1$ for each~$j\in\jobs$ and $\transit{i}$ is constant for each~$i\in M$ or~$\proc{j}=0$ for each~$j\in\jobs$ and $\transit{i}=1$ for each~$i\in M$.
\end{restatable}



\bibliographystyle{plainnat}
\bibliography{mrabbrev2012,literature}

\newpage

\appendix

\section{Proofs of Section~\ref{sec:unbounded}:\newline Hardness of bidirectional scheduling}\label{appendix:unbounded}

In this section, we give a detailed proof of the hardness of the bidirectional scheduling problem for a constant number of segments and identical processing and transit times.
We describe our reduction from \noun{MaxCut}. Let an instance $\I=(G_\I,k)$ of \noun{MaxCut} be given, with $G=(V_\I,E_\I)$, $|V_\I|=n_\I$, and $|E_\I|=m_\I$. We introduce a set of jobs on polynomially many segments that can be scheduled with a total waiting time of $W$ if and only if $\I$ admits a solution. Our construction is comprised of various gadgets which we describe in the following. We make use of suitably large parameters $x \gg y \gg z \gg 1$ that we will specify later. For example, $x$ is chosen in such a way that if ever $x$ jobs are located at the same segment, these jobs need to be processed immediately in order to achieve a waiting time of $W$. Note that because jobs take no time in being processed (i.e.,~$p_j=0$), we can schedule any number of jobs sharing direction simultaneously on a single segment. Also, since~$\tau=1$, it makes no sense for a segment to stay idle if jobs are available. 
This allows us to restrict our analysis to schedules that are \emph{sensible} in the sense that for each segment and at every time step all jobs in one direction available at the segment get scheduled.
On the other hand, the non-zero transit time induces a cost of switching the direction of jobs that are processed at a segment.

\subsubsection*{Vertex gadget.}

Each of the segments $1,10,19,28,\dots$ hosts one vertex gadget for each of the vertices in $V_\I$ (cf.~Figure~\ref{fig:vertex_gadget} with the following). Each vertex gadget~$g_t$ on segment $9\ell+1$ occupies a distinct time interval $[13t,13(t+1))$, $t<n_\I$, on the segment and is associated with one of the vertices $v\in V_\I$. The gadget comes with $24y$ \emph{vertex jobs} that only need to be processed at segment $9\ell+1$, half of them being leftbound, half being rightbound. Exactly $y$ jobs of each direction are released at times $13t,13t+1,\dots,13t+11$. We say that $g_t$ is scheduled \emph{consistently} if either all leftbound vertex jobs are processed immediately when they are released and all rightbound jobs wait for one time unit, or vice-versa. We say the gadget is in the \emph{leftbound} (\emph{rightbound}) \emph{state} and interpret this as vertex $v$ being part of set $V_1$ ($V_2$) of the partition of $V_\I = V_{1} \cupdot V_{2}$ we are implicitly constructing. A schedule is \emph{consistent} if all vertex gadgets are scheduled consistently. The following lemma allows us to distinguish consistent schedules.

\begin{restatable}{lemma}{lemVertexGadget}
The vertex jobs of a single vertex gadget can be scheduled consistently with a waiting time of $12y$, while every inconsistent schedule has waiting time at least $13y$.
\label{lem:vertex_gadget}
\end{restatable}
\begin{proof}
Since $p=0$, we can schedule all available jobs with the same direction simultaneously. It follows that both consistent schedules are valid, and, since in both exactly half of the vertex jobs wait for one unit of time, the total waiting time of such a schedule is~$12y$. Any inconsistent (sensible) schedule would have to send jobs in the same direction in two consecutive unit time intervals, which means that in addition to the minimum waiting time of~$12y$, at least $y$ jobs have to wait an extra unit of time.
\end{proof}

\subsubsection*{Synchronizing vertex gadgets.}
Since every vertex $v\in V_\I$ is represented by multiple vertex gadgets on different segments, we need a way to ensure that all vertex gadgets for $v$ are in agreement regarding which part of the partition $v$ is assigned to. We introduce two different gadgets that handle synchronization. The \emph{copy gadget} synchronizes the vertex gadgets $g_t$ occupying the same time interval on segments~$9\ell+1$ and $9\ell+10$, while the \emph{transposition gadget} synchronizes gadgets $g_t,g_{t+1}$ on segment~$9\ell+1$ with gadgets $g_{t+1},g_t$ on segment~$9\ell+10$. Using a combination of copy and transposition gadgets, we can transition between any two orders of vertex gadgets on distant segments.

\begin{figure}[t]
\begin{centering}
\includegraphics[width=.5\textwidth]{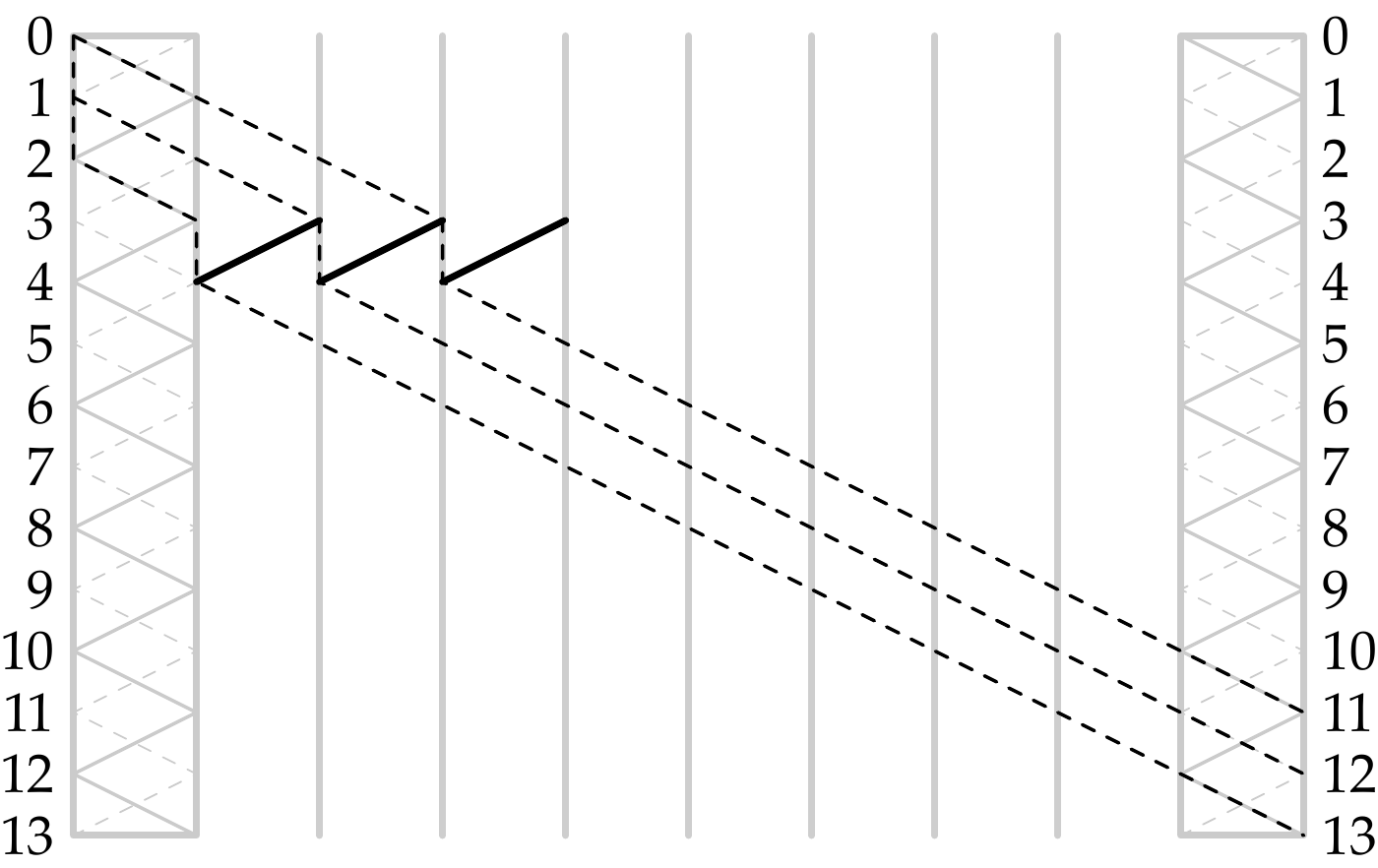}
\par\end{centering}
\protect \caption { Illustration of the copy gadget between two vertex gadgets. The dashed lines depict all sensible trajectories of the synchronizing jobs, assuming that the vertex gadgets are in the same state. \label{fig:copy_gadget} }
\end{figure}

We first specify the copy gadget that synchronizes the vertex gadgets $g_t$ on two segments~$9\ell+1$ and $9\ell+10$ (cf.~Figure~\ref{fig:copy_gadget} with the following). The gadget consists of $2z$ rightbound \emph{synchronization jobs}, half of which are released at time~$13t$ and half at time~$13t+1$. The jobs need to be processed on all segments~$9\ell+1,\dots,9\ell+10$ in this order. In addition, we introduce~$3x$ \emph{blocking jobs} that are used to enforce that specific time intervals on a segment are reserved for leftbound/rightbound operation. Essentially, releasing~$x$ blocking jobs at time~$t$ on a single segment prevents any jobs to be processed in opposite direction during the time interval~$[t,t+1)$ (and even earlier). In this manner, we block the interval starting at time $13t+3$ on segments $9\ell+2,9\ell+3,9\ell+4$.

\begin{restatable}{lemma}{lemCopyGadget}
In any consistent schedule, the synchronization jobs of a single copy gadget can be scheduled with a waiting time of $3z$ if the two corresponding vertex gadgets are in the same state, otherwise their waiting time is at least $5z$.
\label{lem:copy_gadget}
\end{restatable}
\begin{proof}
  Since $x \gg z$, we need to schedule all blocking jobs as soon as they are released.
  If both vertex gadgets $g_t$ linked by the copy gadget are in the rightbound state, the synchronization jobs released at time $13t$ only have to wait for one time unit at segment $9\ell+4$, while the other jobs have to wait at segments $9\ell+1$ and $9\ell+2$. Similarly, if the vertex gadgets are in the leftbound state, the first half of the jobs have to wait at segments $9\ell+1$ and $9\ell+3$, while the other half only has to wait at segment $9\ell+3$. The waiting time in either case is $3z$. If the vertex gadgets are in opposite states, all jobs have to additionally wait at segment $9\ell+10$, which results in a total waiting time of at least $5z$.
\end{proof}

\begin{figure}[t]
\begin{centering}
\includegraphics[width=.5\textwidth]{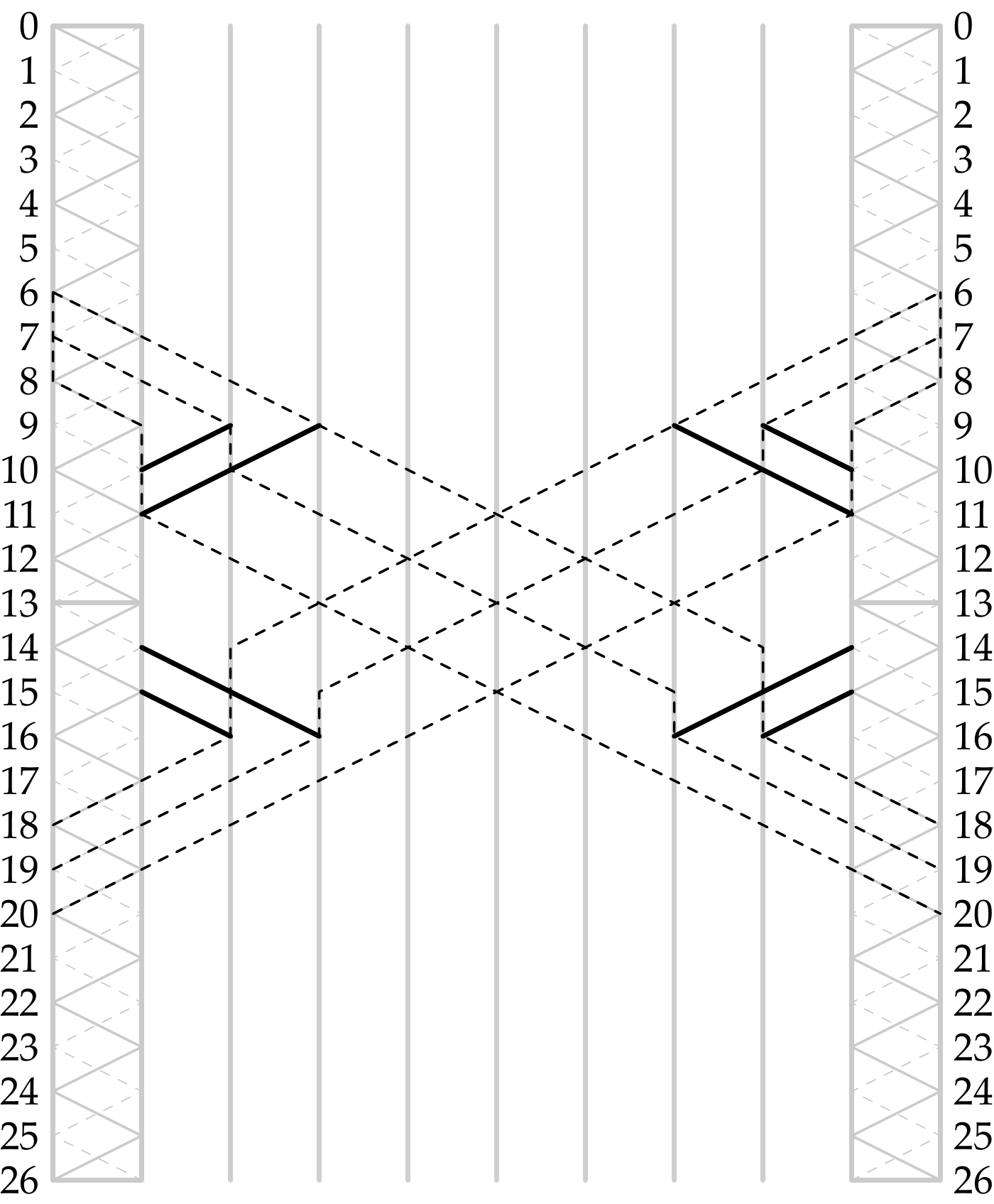}
\par\end{centering}
\protect \caption { Illustration of the transposition gadget. The dashed lines depict all sensible trajectories of the synchronizing jobs, assuming that the vertex gadgets are pairwise in the same states. Note that jobs in different directions never meet while in transit through the same segment. \label{fig:transposition_gadget} }
\end{figure}

We now describe the transposition gadget that synchronizes the vertex gadgets $g_t,g_{t+1}$ on segment~$9\ell+1$ with the vertex gadgets $g_{t+1},g_t$ on segment~$9\ell+10$ (cf.~Figure~\ref{fig:transposition_gadget} with the following). The challenge here is that jobs synchronizing the different pairs of vertex gadgets need to pass each other without interfering. We achieve this by making sure that the jobs never meet while being in transit at the same segment. The gadget consists of $4z$ synchronization jobs, half being rightbound and half being leftbound. Half of each are released at times~$13t+6$ and $13t+7$, and all need to be processed at segments~$9\ell+1,\dots,9\ell+10$ (in different directions). In addition, we introduce~$12x$ blocking jobs to block the intervals starting at the following times: at times $13t+9$, $13t+10$ for rightbound jobs and at times $13t+14$, $13t+15$ for leftbound jobs on segment $9\ell+2$, at times $13t+9$ for rightbound and at $13t+15$ for leftbound on segment~$9\ell+3$, and the corresponding (symmetrical) intervals in opposite direction on segments~$9\ell+8$ and $9\ell+9$ (cf.~Figure~\ref{fig:transposition_gadget}).

\begin{restatable}{lemma}{lemTranspositionGadget}
In any consistent schedule, the synchronization jobs of a single transposition gadget can be scheduled with a waiting time of $10z$ if each of the two pairs of corresponding vertex gadgets are in the same state, otherwise their waiting time is at least $12z$.
\label{lem:transposition_gadget}
\end{restatable}
\begin{proof}
  Since $x \gg z$, we need to schedule all blocking jobs as soon as they are released.
  It is easy to verify that all synchronization jobs wait at exactly 2 segments due to blocking jobs. In addition, half of the jobs wait for one unit of time at the segment where they are released -- for a total of $10z$ time units. If the pair of vertex gadgets is in opposite states, all connecting synchronization jobs need to wait at least one additional unit of time at their last segment. Observe that synchronization jobs in opposite directions are never in transit on the same segment at the same time.
\end{proof}

\subsubsection*{Edge gadget.}

\begin{figure}[t]
\begin{centering}
\includegraphics[width=.5\textwidth]{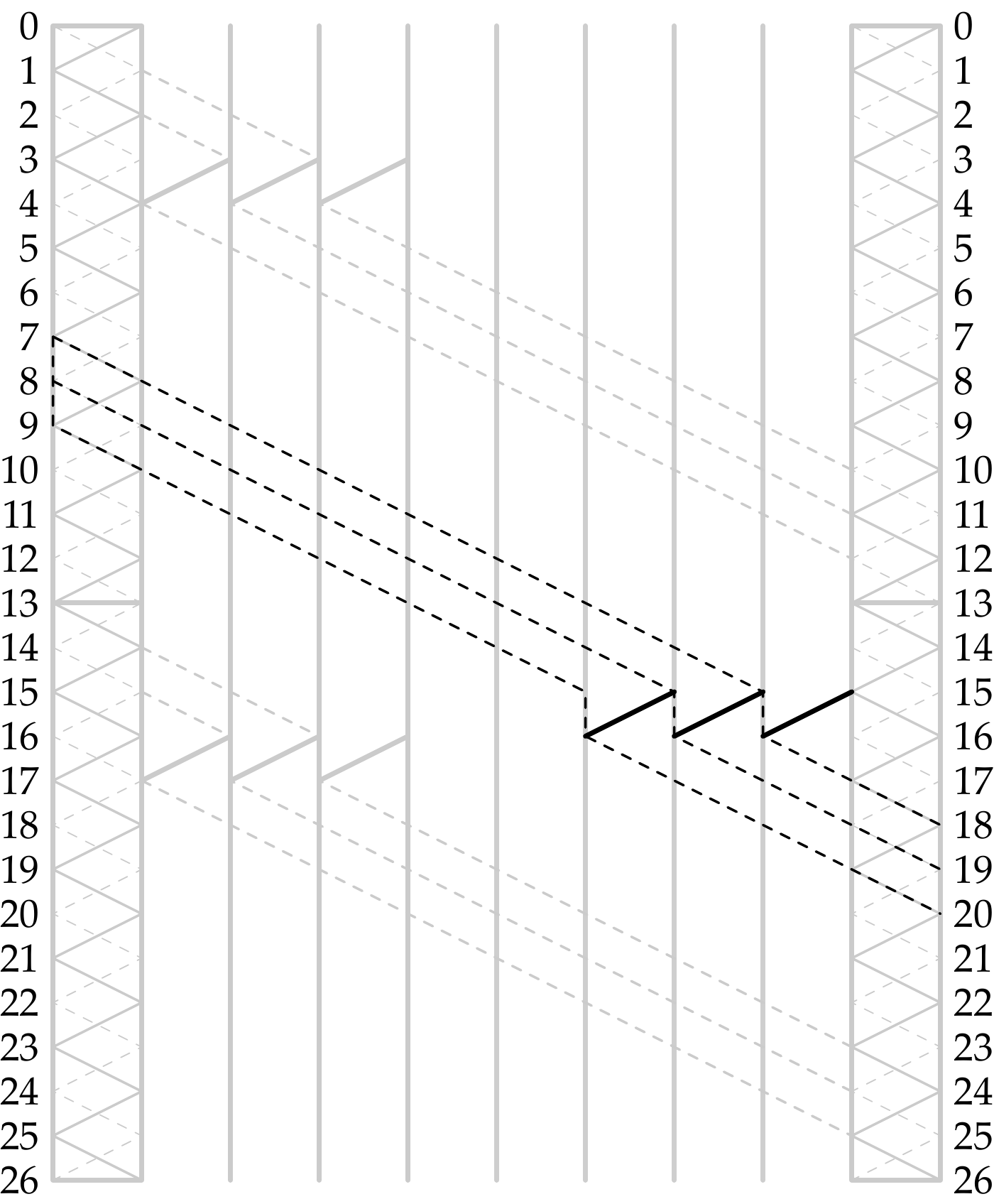}
\par\end{centering}
\protect \caption { Illustration of the edge gadget. The dashed lines depict all sensible trajectories of the synchronizing jobs, assuming that the vertex gadgets are in opposite states. Note that edge jobs do not interact with synchronization jobs of copy gadgets for both vertices. \label{fig:edge_gadget} }
\end{figure}

The purpose of an edge gadget between vertex gadget $g_t$ on segment~$9\ell+1$ and $g_{t+1}$ on segment $9\ell+10$ is to produce a small additional waiting time if the two vertex gadgets are in the same state (cf.~Figure~\ref{fig:edge_gadget} with the following). We will introduce edge gadgets between vertex gadgets representing two vertices $u,v$ that share an edge in $G$. This way, every edge that connects vertices in different parts of the partition is beneficial for the resulting waiting time. The edge gadget itself consists of 2 rightbound \emph{edge jobs}, one being released at time $13t+7$ and the other at time $13t+8$. Both jobs need to be processed on segments $9\ell+1,\dots,9\ell+10$. We add $3x$ blocking jobs to block the unit time interval starting at time $13t+15$ on segments $9\ell+7,9\ell+8,9\ell+9$.

\begin{restatable}{lemma}{lemEdgeGadget}
In any consistent schedule, the edge jobs of a single edge gadget can be scheduled with a waiting time of $3$ if the two connected vertex gadgets are in opposite states, otherwise their waiting time is at least $5$.
\label{lem:edge_gadget}
\end{restatable}
\begin{proof}
  One job always has to wait for a time unit at the first segment. Both jobs have to wait for the blocking jobs (since~$x \gg 1$). If the vertex gadgets are in the same state, both jobs have to wait an additional unit of time at the last segment.
\end{proof}

\subsubsection*{Construction.}

We are now ready to combine our gadgets and explain the final construction.

\unbounded*
\begin{proof}
We start by introducing a vertex gadget~$g_t$ on segment~$1$ for each vertex~$v_t\in V_\I$ of the given \noun{MaxCut}-instance.
For each edge~$\{u,v\}$ we extend the construction by appending more segments as follows.
We add a sequence of blocks of $9$ segments, the last of which contains again a vertex gadget for each vertex.
In between we add copy and transposition gadgets in such a way that on the last segment $i$ the vertex gadgets $g_0$ and $g_1$ represent the vertices $u$ and~$v$.
We can achieve this by adding less than $n_\I$ segments.
We add an additional block of $9$ segments, and add copy gadgets for each of the variables.
Finally, we add an edge gadget connecting vertex gadget $g_0$ on segment $i$ with $g_1$ on the last segment.
Observe that the edge jobs do not interfere with any of the synchronization jobs for the copy gadgets for the first two vertices~(cf.~Figure~\ref{fig:edge_gadget}).
We repeat the process once for each edge.
The total number of segments is $\O(n_\I m_\I)$,
and the total number of jobs is $\O(n^2_\I m_\I(x+y+z))$. 
The number of vertex gadgets is $n_{v}<n^2_\I m_\I$, and the number of transposition and copy gadgets is $n_t<n_c<n_{v}$.

We claim that if the \noun{MaxCut} instance admits a solution~$\mathcal{S}$, we
can schedule all jobs with waiting time at most 
$W=12 n_v y + 3 n_c z + 10 n_t z + 5 m_\I - 2k$.
We do this by scheduling all vertex gadgets consistently in the state corresponding to the part of the partition the corresponding vertex belongs to in $\mathcal{S}$.
Lemmas~\ref{lem:vertex_gadget}
through~\ref{lem:transposition_gadget} guarantee that we can schedule
everything but the edge jobs without incurring a waiting time greater
than $12 n_v y + 3 n_c z + 10 n_t z$. Finally, since at least $k$ edges in the
\noun{MaxCut} solution are between vertices in different sets of the
partition, and the vertex gadgets are set accordingly, by 
Lemma~\ref{lem:edge_gadget}, we obtain an additional waiting time of at most $5m_\I-2k$ as claimed.

It remains to establish that the waiting time exceeds $W$ in case
the \noun{MaxCut} instance does not admit a solution. We set $x=W+1$, such that all blocking jobs have to be scheduled as soon as they are released. 
By Lemma~\ref{lem:vertex_gadget},
scheduling at least one vertex gadget inconsistently produces a total waiting
time of at least $12 n_v y + y$. We now set $y = 18 n_\I^2 m_\I z > 3 n_c z + 10 n_t z + 5m_\I$ for the vertex jobs,
such that a single inconsistent vertex gadget results in a waiting time
greater than $W$. 
Hence, each vertex gadget needs to be scheduled consistently. By Lemmas~\ref{lem:copy_gadget} and~\ref{lem:transposition_gadget},
we have that if not all vertex gadgets corresponding to the same vertex are in the same state, the waiting time for vertex and synchronization jobs is at least $12 n_v y + 3 n_c z + 2n_t z + z$. 
We set $z = 5m_\I$, which
allows us to conclude that all vertex gadgets are in agreement regarding
the partition of the vertices. Finally, Lemma~\ref{lem:edge_gadget}
enforces that there are at least $k$ edge gadgets between vertices
in different states. This however is impossible as our \noun{MaxCut}
instance does not admit a solution.
\end{proof}

\corUnboundedWithProcessing*
\begin{proof}
  We adapt our construction by setting~$p=1$ and~$\tau=n^2 m$ and scaling all release times by $n^2 m$, where $n,m$ are the number of jobs and segments, respectively. 
	We claim that the original instance admits a solution of some waiting time $W$ if and only if it now admits a solution with waiting time in $[W\tau, (W+1)\tau)$.
	This proves the Corollary, as the intervals are pairwise disjoint for different (integer) values of $W$. 
	
	If the original construction (with~$p=0$ and~$\tau=1$) does not admit a solution with waiting time at most $W$,
then a scaled version with~$p=0$ and~$\tau = n^2 m$ does not admit a solution with waiting time at most~$W\tau$. But the lowest possible waiting is monotonically increasing with increasing processing times, hence the adapted instance with~$p=1$ does not admit a solution of waiting time at most~$W\tau$.

Conversely, assume we have a solution of the original instance with waiting time~$W$. We fix the order in which jobs are processed along each segment and construct a schedule for the setting~$p=1$, $\tau=n^2 m$ by introducing additional waiting periods for each job. Clearly, each job has to wait at most one time unit for each other job to be processed at each segment. Hence, the additional waiting time overall is smaller than~$n^2 m = \tau$. 
\end{proof}

\section{Proofs of Section~\ref{sec:PTAS}: \newline A PTAS for bidirectional scheduling}
\label{appendix:ptas}

\ifthenelse{\boolean{ptas-more}}{
In this Section we restate the Lemmas with detailed proofs that are necessary to show the existence of a PTAS if the processing times of the jobs are not restricted to be equal.

\subsection{Single Segment}

We consider first the case of a single segment, or, more precisely, the bidirectional scheduling problem on a single segment and with compatibility graph~$G_1\in\{K_{\rightb{n},\leftb{n}},\emptyset\}$.
}{
In this Section we state the Lemmas with detailed proofs that are necessary to show the existence of a PTAS if the processing times of the jobs are not restricted to be equal in the case of a single segment. More precisely, we consider the bidirectional scheduling problem on a single segment with compatibility graph~$G_1\in\{K_{\rightb{n},\leftb{n}},\emptyset\}$.
}
Following the proof scheme of~~\cite{Afrati1999}, we introduce several lemmas that allow us to make assumptions at ``$\O(1+\eps)$-loss'', meaning that we can modify any input instance and optimum schedule to adhere to these assumptions, such that the resulting schedule is within a factor polynomial in~$(1+\eps)$ of the optimum schedule for the original instance.
To not complicate matters unnecessarily, in the following we allow fractional release dates and processing times.

\begin{restatable}{lemma}{lemPtasGeometricRounding}
\label{lem:ptas-geometric_rounding}
  With~$\O(1+\eps)$-loss we can assume that $\rel{j},\proc{j}\in\{(1+\eps)^x\mid x\in\N\}\cup\{0\}$, $\rel{j} \geq \eps(\proc{j}+\transit{1})$, and~$r_j\geq 1$ for each job~$j\in\jobs$.
\end{restatable}
\begin{proof}
  Increasing any value~$v\in\mathbb{R}$ to the smallest power of~$(1+\eps)$ not smaller than~$v$ yields a value~$v' = (1+\eps)^x=(1+\eps)(1+\eps)^{x-1}\leq(1+\eps)v$. 
  Hence, multiplying all start times of a schedule by~$(1+\eps)$ gives a feasible schedule even when rounding up all nonzero processing times 
  to the next power of~$(1+\eps)$. The total completion time does not increase by more than a factor of~$(1+\eps)$.

  By shifting the completion times of a schedule with adapted processing times by a factor of~$(1+\eps)$, we obtain increased start times~$\start{j}'$ for each job~$j$:
  \[
  \start{j}' = (1+\eps)\compl{j} - (\proc{j}+\transit{1}) 
	= (1+\eps)(\start{j} + \proc{j} + \transit{1}) - (\proc{j} + \transit{1}) 
  \geq \eps(\proc{j}+\transit{1})\text{.}  
  \]
  Hence, by losing not more than a~$(1+\eps)$-factor we may assume that all jobs have release dates of at least an~$\eps$ fraction of their running time. Now, we can scale the instance by some power of~$(1+\eps)$, such that the earliest release date is at least one (since jobs with $\rel{j}=\proc{j}=\transit{1}=0$ can be ignored). 

  Finally, multiplying again all start times of a schedule with adapted processing times and release dates by~$(1+\eps)$ yields a feasible schedule even when rounding up all nonzero release dates to the next power of~$(1+\eps)$.
\end{proof}

We define~$R_x=(1+\eps)^x$ and consider time intervals~$I_x=[R_x,R_{x+1}]$ of length~$\eps R_x$.

\enlargethispage*{1ex}
\begin{restatable}{lemma}{lemPtasBoundedCrossing}
  \label{lem:ptas-bounded_crossing}
    Each job runs for at most~$\cs:=\lceil\log_{1+\eps}\frac{1+\eps}{\eps} \rceil$ intervals, i.e., a job starting in interval~$I_x$ is completed before the end of~$I_{x+\cs}$.
\end{restatable}
\begin{proof}
    Consider some job~$j$ and assume that~$j$ starts in~$I_x$ in some schedule. By Lemma~\ref{lem:ptas-geometric_rounding} we get 
    \[
    |I_x| = \eps R_x \geq \eps\rel{j} \geq \eps^2(\proc{j}+\transit{1})\text{.}
    \]
  Thus, the running time of~$j$ is bounded by~$|I_x|/\eps^2$. The constant upper bound of~$1/\eps^2$ for the number of used intervals can still be improved since the length of the next~$\cs$ succeeding intervals with increasing size is sufficient to cover a length of~$|I_x|/\eps^2$. Using the fact that~$\sum_{k=0}^n z^k = \frac{1-z^{n+1}}{1-z}$ we get
  \begin{align*}
    \sum_{i=0}^{\cs} |I_{x+i}| & = \sum_{i=0}^{\cs} (R_{x+i+1}- R_{x+i}) = 
  |I_x|\sum_{i=0}^{\cs} (1+\eps)^i \\
    & = |I_x|\frac{1-(1+\eps)^{{\cs}+1}}{1-(1+\eps)}\\
    & \geq |I_x|\frac{1-\frac{1+\eps}{\eps}}{-\eps} = 
  |I_x|\frac{1+\eps-\eps}{\eps^2} =\frac{|I_x|}{\eps^2},
  \end{align*}
  which concludes the proof.
\end{proof}

We use the common technique of \emph{time-stretching}. We shift each start time (or completion time) to the next interval while maintaining the same offset to the beginning of the interval.  This way, the schedule remains feasible and the objective is increased by a factor of at most~$(1+\eps)$. Intuitively, this process can be interpreted as stretching the length of each time interval~$I_x$ by a factor of~$(1+\eps)$, i.e., its length is increased by~$\eps|I_x|$. When applying (multiple) time-stretches we use the following observation to assess the additional empty space created between jobs:

\begin{restatable}{lemma}{lemPtasTimeStretch}
  \label{lem:ptas-time-stretch}
  Consider two distinct times~$T_1 < T_2$ with~$T_1\in I_{x(1)}$ and~$T_2\in I_{x(2)}$. Applying~$\ell$ time-stretches yields shifted times~$T'_1 < T'_2$ with
  \begin{equation}\label{eq:bidir-time-shift}
      (T'_2 - T'_1) \geq (T_2 - T_1) + \idle{x(1)}{x(2)},
  \end{equation}
	where $\idle{x(1)}{x(2)} := \sum_{x(1)\leq x < x(2)}\ell\eps|I_x|$.
\end{restatable}
\begin{proof}
  We calculate
  \begin{align*}
    (T'_2 - T'_1) & = R_{x(2)+\ell} + (T_2-R_{x(2)}) - [R_{x(1)+\ell} + (T_1-R_{x(1)})]\\
     & = ((1+\eps)^{\ell} - 1) R_{x(2)} + T_2 - ((1+\eps)^{\ell} - 1) R_{x(1)} - T_1\\
     & \geq (T_2 - T_1) + (1+\ell\eps -1)(R_{x(2)} - R_{x(1)})\\
     & = (T_2 - T_1) + \ell\eps\sum_{x(1)\leq x < x(2)}|I_x|.
  \end{align*}
\end{proof}

We can now apply time-stretches to the start or completion times of all jobs and use the above observation to quantify the additional space created in the schedule. Consider two jobs $j,k\in\jobs$ with starting times~$\start{j}<\start{k}$, and let $s(j),s(k)$ (resp.\ $c(j),c(k)$) denote the intervals in which their start (completion) times fall, i.e., $\start{j}\in I_{s(j)}$ (and $\compl{j}\in I_{c(j)}$). E.g., if we apply~$\ell$ time-stretches to starting times, we obtain an additional gap of $\idle{s(j)}{s(k)}$ between the new starting and completion times. Table~\ref{tab:bidir-time-shift} summarizes the resulting gaps depending on whether start or completion times are stretched and whether~$j,k$ travel in the same or opposite directions.
%
%
\begin{table}[ht]
\caption{Summary of the increased differences between start and completion times of jobs~$j,k\in\jobs$, $\start{j}<\start{k}$, when stretching start times (denoted by~$\smash{'}$) or completion times (denoted by~$\smash{''}$). We use Lemma~\ref{lem:ptas-time-stretch} together with the fact that~$j$ and~$k$ did not overlap before the time-stretch.}\label{tab:bidir-time-shift}
  \centering
  \scriptsize
  \begin{tabulary}{\linewidth}{RRCC}
    \toprule
    time-stretch on & &  same direction &  opposite direction \\
    \midrule
    \multirow{1}{*}{start times} &%
    \small\eqref{eq:bidir-time-shift}&%
    \small$\start{k}'\geq \start{j}'+\proc{j} + \idle{s(j)}{s(k)}$ &%
    \small$\start{k}'\geq \start{j}'+\proc{j} +\transit{} + \idle{s(j)}{s(k)}$\\
    & \small $\Rightarrow$&%
    \small$\compl{k}'\geq \compl{j}'+\proc{k} + \idle{s(j)}{s(k)}$ &%
    \small$\compl{k}'\geq \compl{j}'+\proc{k} + \transit{}+\idle{s(j)}{s(k)}$\\
    \midrule
    \multirow{1}{*}{compl.\ times} &%
    \small\eqref{eq:bidir-time-shift}&%
    \small$\compl{k}''\geq \compl{j}''+\proc{k} + \idle{c(j)}{c(k)}$ &%
    \small$\compl{k}''\geq \compl{j}''+\proc{k} + \transit{}+\idle{c(j)}{c(k)}$\\
    & \small $\Rightarrow$&%
    \small$\start{k}''\geq \start{j}''+\proc{j}+\idle{c(j)}{c(k)}$ &%
    \small$\start{k}''\geq \start{j}''+\proc{j}+\transit{}+\idle{c(j)}{c(k)}$\\
    \bottomrule
  \end{tabulary}
\end{table}

To analyze the set of jobs released within each interval we partition them as follows. A job~$j$ released  at~$R_x$ is called \emph{small} if~$\proc{j}\leq\frac{\eps^2}{4}|I_x|$ and \emph{large} otherwise. With this, we partition for each direction~$\dir\in\{\rb, \lb\}$ the jobs~$\dirbset{J}{\dir}_{x}:=\{j\in \dirbset{J}{\dir} \mid  \rel{j}=R_x\}$ released at~$R_x$ into the subsets~$\dirbset{S}{\dir}_x=\{j\in \dirbset{J}{\dir}_{x} \mid j \text{ is small}\}$ and~$\dirbset{L}{\dir}_x=\{j \in \dirbset{J}{\dir}_{x} \mid j \text{ is large}\}$.
We will see that the arrangement of jobs of each~$\dirbset{S}{\dir}_x$ does not influence the remaining jobs too much such that we can assume a fixed order for each of these sets.  To do so, we say that a subset~$\set{J'}\subseteq\jobs$ of jobs is scheduled in \emph{SPT order} (shortest processing time first)  if~$\start{j_1}\leq\start{j_2}$ for any pair of jobs~$j_1, j_2\in \set{J'}$ with~$\proc{j_1}<\proc{j_2}$. Furthermore, we denote the sum of processing times of~$\set{J'}$ as~$p(\set{J'})$ and the union of small jobs  released up to some point~$R_{x}$ with direction~$\dir\in\{\rb, \lb\}$ by~$\dirbset{S}{\dir}_{\leq x} = \bigcup_{x'\leq x}\dirbset{S}{\dir}_{x}$.

\begin{restatable}{lemma}{lemPtasSPTForSmallJobs}
  \label{lem:ptas-SPT_for_small_jobs}
    With~$\O(1+\eps)$-loss we can restrict to schedules such that for each~$x\geq 0$ and each~$\dir\in\{\rb, \lb\}$:
    \begin{compactenum}
      \item\label{ptas-it:small-without-rel} the processing of no small job contains a release date,
      \item\label{ptas-it:small-spt} jobs contained in~$\dirbset{S}{\dir}_{\leq x}$ are scheduled in SPT order within~$I_x$, and
      \item\label{ptas-it:small-bounded}~$\proc{}(\dirbset{S}{\dir}_x) \leq |I_x|$.
    \end{compactenum}
\end{restatable}
\begin{proof}
To prove claim~\ref{ptas-it:small-without-rel} we consider some schedule and apply a time-stretch via start times. Observe that no further crossing of a processing over a release date is produced for small jobs. If there was a release date~$R_{s(j)+1}$ contained in the processing interval of a small job of~$I_{s(j)}$ it is moved behind the processing since we get by Lemma~\ref{lem:ptas-time-stretch} that~$R_{s(j)+2}-\start{j}' \geq R_{s(j)+1}-\start{j} + \eps|I_{s(j)}|$ which gives an increase larger than the processing time of this job.

For a proof of claim~\ref{ptas-it:small-spt} consider a schedule~$S$ where no processing of a small job contains a release date and apply one time-stretch via start times. This increases the objective value by at most a~$1+\eps$ factor. Denote the resulting schedule as~$S'$. To achieve the demanded properties, apply the following procedure for each direction~$\dir\in\{\rb, \lb\}$. 
First, remove all small jobs from schedule~$S'$. Now consider each interval~$I_{x}, x=0, 1, \dots$. Denote by~$A_x$ the set of removed jobs from~$I_x$. If jobs have been removed in~$I_x$ there are idle intervals where jobs in direction~$\dir$ can be scheduled. Denote the subset of~$\dirbset{S}{\dir}_{\leq x}$ already scheduled in earlier intervals by~$B_{<x}$ and order the subset~$C_{x}:=\dirbset{S}{\dir}_{\leq x}\setminus B_{<x}$ of unscheduled jobs in SPT order.  Define for~$t\in I_x$ by~$p_t(A_x):=p(\{j\in A_x \mid \start{j}'\leq t\})$ the amount of processing time of jobs started before time~$t$ in~$S'$. Now let~$C_{x}(t)$ be the smallest SPT-subset of~$C_x$ such that~$p(C_{x}(t))\geq p_t(A_x)$ or~$C_{x}(t)=C_x$. Iterate from the earliest created maximal empty interval to the latest and fill each interval~$[t_1,t_2]$ in SPT order such that the jobs of~$p(C_{x}(t_2))$ start before~$t_2$.  Note that~$p(C_{x}(t_2))\leq p_{t_2}(A_x)+\frac{\eps^2}{4}|I_x|$ since we consider only small jobs. To maintain feasibility we increase the start of the following jobs from~$J\setminus C_x$, if necessary. (This decreases eventually the size of the following empty interval which is no problem). Nevertheless, the start time of no job from~$J\setminus C_x$ is increased by more than~$\frac{\eps^2}{4}|I_x|$. Hence, their completion time is increased by less than a~$1+\eps$ factor and the jobs starting after~$R_{x+1}$ are not affected. Note that no processing of the assigned small jobs~$B_x:=C_{x}(R_{x+1})$ contains~$R_{x+1}$.

Since we used in each interval an assignment via SPT order we know that at each point in time the number of already started small jobs has not been decreased. Therefore, the total completion time of small jobs overall has not been increased.

To prove claim~\ref{ptas-it:small-bounded} consider for each~$x = 0, 1, \dots$ the largest SPT-subset~$\set{J}'_x$ of~$\dirbset{S}{\dir}_x$, such that~$p(\set{J}'_x)\leq|I_x|$. By assumptions~\ref{ptas-it:small-spt} and~\ref{ptas-it:small-without-rel} we can be sure that all jobs of~$\dirbset{S}{\dir}_x\setminus\set{J}'_x$ are not scheduled within~$I_x$ and thus, we can move their release dates to~$R_{x+1}$.
\end{proof}

Once, we have a fixed order to schedule small jobs with the same release date we are able to glue them to job packs of a certain minimum size. For this purpose we apply a further time-stretch to join the processing of jobs assigned to the same pack.  This increases for each interval~$I_y$ the amount of processing per direction and each earlier interval~$I_x$ by at most the size of one job being small at time~$R_x$.
The following lemma yields that the extra space of one interval created by one time-stretch is sufficient to cover this amount for all earlier Intervals.

\begin{restatable}{lemma}{lemAPtasIntervalVolume}\label{lem:A-ptas-interval_volume}
  We have $\sum_{x<y}\eps^2 |I_x| \leq \eps|I_y|$.
\end{restatable}
\begin{proof}
  To prove the claim we again use that~$\sum_{k=0}^n z^k = \frac{1-z^{n+1}}{1-z}$:
  \begin{equation*}
    \eps^3\sum_{x<y}(1+\eps)^x = \eps^3\frac{1-(1+\eps)^y}{1-(1+\eps)}
     = \eps^2 ((1+\eps)^y-1) \leq \eps|I_y|.
  \end{equation*}
\end{proof}

\begin{restatable}{lemma}{lemPtasSmallPackages}
	\label{lem:ptas-small-packages}
  With~$\O(1+\eps)$-loss we can restrict to schedules such that for each~$x\geq 0$ and each~$\dir\in\{\rb, \lb\}$ the jobs of~$\dirbset{S}{\dir}_x$ in SPT order are joined to unsplittable job packs with size of at most~$\frac{\eps^2}{4}|I_x|$ and at least~$\frac{\eps^2}{8}|I_x|$ each. 
\end{restatable}
\begin{proof}
Consider a schedule satisfying at~$\O(1+\eps)$-loss the properties of Lemma~\ref{lem:ptas-SPT_for_small_jobs} and apply one time-stretch via start times. We now apply the following procedure for each direction~$\dir\in\{\rb, \lb\}$ and each~$x=0, 1, \dots$. Recall that the jobs of~$\dirbset{S}{\dir}_x$ are scheduled in SPT order. Let~$\dirbset{T}{\dir}_x = \{j\in \dirbset{S}{\dir}_x \mid \proc{j}< \frac{\eps^2}{8}|I_x|\}$ be the subset of jobs being too small. Remove the jobs of~$\dirbset{T}{\dir}_x$ from the current schedule and join the jobs of~$\dirbset{T}{\dir}_x$ successively in SPT order to minimal job packs such that the processing times of each job pack sum up to at least~$\frac{\eps^2}{8}|I_x|$. (The processing time of the last pack is artificially increased if necessary.) We now reassign complete job packs to the empty intervals similarly to the procedure in the proof of Lemma~\ref{lem:ptas-SPT_for_small_jobs}. Hence, no start time of~$\dirbset{T}{\dir}_x$ has been increased and the start time of no job in~$J\setminus \dirbset{T}{\dir}_x$ has been increased by more than~$\frac{\eps^2}{4}|I_x|$.
  
In total, the start time of no job starting in interval~$I_{y+1}$ has been increased by more than~$2\cdot\sum_{x<y} \frac{\eps^2}{4}|I_x| + 2\cdot\frac{\eps^2}{4}|I_y|\leq \eps|I_y|$ due to Lemma~\ref{lem:A-ptas-interval_volume}. By Lemma~\ref{lem:ptas-time-stretch} (or Table~\ref{tab:bidir-time-shift}) we can conclude that no job has been delayed to a later interval by the rearrangement. Note that properties~\ref{ptas-it:small-without-rel} and~\ref{ptas-it:small-bounded} of Lemma~\ref{lem:ptas-SPT_for_small_jobs} still hold whereas property~\ref{ptas-it:small-spt} (SPT order) remains true only within each~$\dirbset{S}{\dir}_x$.
\end{proof}

Therefore, we can consider each job pack simply as one small job. Nevertheless, the original jobs must be used for the evaluation of the completion times.
Besides the scheduling restrictions for small jobs we can also bound the amount of large jobs released at the beginning of each interval.

\begin{restatable}{lemma}{lemPtasBoundedReleasePerInterval}
\label{lem:ptas-bounded_release_per_interval}
  With~$\O(1+\eps)$-loss we can assume for each~$x\geq 0$ and each~$\dir\in\{\rb, \lb\}$ that:
\begin{compactenum}
  \item the number of possible processing times in~$\dirbset{L}{\dir}_x$ is bounded by~$5\log_{(1+\eps)}\frac{1}{\eps}$, and
  \item the number of jobs per processing time in~$\dirbset{L}{\dir}_x$ is bounded by~$\frac{4}{\eps^2}$.
\end{compactenum}
\end{restatable}
\begin{proof}
  Consider some scheduling instance, some~$\dir\in\{\rb, \lb\}$ and some~$x\geq 0$.   The processing time of the jobs in~$\dirbset{L}{\dir}_x$ are, by definition, at least~$\frac{\eps^3}{4}(1+\eps)^x$. On the other hand, by Lemma~\ref{lem:ptas-geometric_rounding}, the processing times are at most~$\frac{1}{\eps}(1+\eps)^x$. Let~$x_j$ be such that~$\proc{j}=(1+\eps)^{x_j}$. We get
  \[
  \begin{array}{rcl}
  \frac{\eps^3}{4} \leq & \frac{(1+\eps)^{x_j}}{(1+\eps)^{x}} & \leq \frac{1}{\eps}\\
  \Longrightarrow\quad \log_{(1+\eps)}\frac{\eps^3}{4} \leq & x_j-x & \leq \log_{(1+\eps)}\frac{1}{\eps}\\
  \end{array}
  \]
  The difference of these bounds is~$4\log_{(1+\eps)}\frac{1}{\eps} + \log_{(1+\eps)}4$ which gives a constant number of possible integer values for~$x_j$ and, hence, a constant number of possible processing times for each job in~$\dirbset{L}{\dir}_x$. Finally, since each large job in~$I_x$ has a processing time of at least~$\frac{\eps^2}{4} |I_x|$, we can schedule at most~$4/\eps^2$ jobs per direction within~$I_x$, and the remaining jobs need to start after~$R_{x+1}$.
\end{proof}

\begin{restatable}{lemma}{lemPTASSafetyNet}
\label{lem:ptas-safety_net}
  With~$\O(1+\eps)$-loss we can assume, that each job is finished within a constant number of intervals after its release.
\end{restatable}
\begin{proof}
  Consider the set of jobs~$\jobs_x$ released at time~$R_x$. By Lemma~\ref{lem:ptas-geometric_rounding} the running time of each such job is at most~$R_x/\eps$. 
  Therefore, applying Lemmas~\ref{lem:ptas-SPT_for_small_jobs} and~\ref{lem:ptas-bounded_release_per_interval} we can bound the time needed to first schedule all jobs of one direction and afterward all jobs of the other direction:
  \begin{align*} 
  \sum_{\dir\in\{\rb, \lb\}}\left[ p(\dirbset{S}{\dir}_x) + p(\dirbset{L}{\dir}_x) + 
  \transit{1}\right] 
  & \leq 2\bigg[\eps(1+\eps)^x \\
  & \phantom{=}  \quad + \frac{4}{\eps^2}\cdot\frac{1}{\eps}(1+\eps)^x 
  \cdot5\log_{(1+\eps)}\frac{1}{\eps }\bigg] \\
  & =  \eps^2(1+\eps)^x\cdot 2\left[\frac{1}{\eps} + 
  \frac{20}{\eps^5}\log_{(1+\eps)}\frac{1}{\eps}
  \right]\\
  & \leq \eps^2(1+\eps)^x(1+\eps)^{\cs'-1} = \eps|I_{x+\cs'-1}|,
  \end{align*}
  where~$\cs'$ is the smallest possible integer such that~$2\left[\frac{1}{\eps} + \frac{20}{\eps^5}\log_{(1+\eps)}\frac{1}{\eps}\right] \leq (1+\eps)^{\cs'-1}$. Note, that~$\cs'$ is constant. 

  Applying one time-stretch on the start times creates idle time for each interval~$I_x$ somewhere after~$\cs'$ intervals that is sufficient to host all unfinished jobs of~$J_x$, cf.~Lemma~\ref{lem:ptas-time-stretch} and Table~\ref{tab:bidir-time-shift}.
  If no job was running at time~$R_{x+\cs'}$ before the time-stretch this created idle time is now part of interval~$I_{x+\cs'}$.
  Otherwise let~$j$ be the latest of these jobs with start time~$S_j\in I_{s(j)}$ and completion time~$\compl{j}\in I_{c(j)}$ before the time-stretch. Note that~$s(j)\leq x+\cs'-1$ which induces~$c(j)\leq x + \cs' + \cs -1$ due to Lemma~\ref{lem:ptas-bounded_crossing}. By Lemma~\ref{lem:ptas-time-stretch} we can be sure that after the time-stretch there is idle time of~$\sum_{x=s(j)}^{c(j)-1}\eps|I_{k}|$ before 1.\ the start of the next job after~$j$ and 2.\ the end of interval~$I_{x_c+1}$. 
  By definition of~$\cs'$, this time is sufficient to first schedule all jobs of~$J_x$ in heading of~$j$ and then all remaining.
  This way, all jobs of~$J_x$ are scheduled before the end of interval~$I_{x+\cs'+\cs}$.

  Note that this argument assumes that there are no compatibilities. An analog reasoning concerning only the processing times works if all opposed jobs are compatible.
\end{proof}

We can now limit the interface of our dynamic program by showing Lemma~\ref{lem:constant_interface} of Section~\ref{sec:PTAS}.

\lemConstantInterface*

\begin{proof}
  By Lemma~\ref{lem:ptas-small-packages} we may assume that small jobs in~$\dirbset{S}{\dir}_{x}$ have processing time at least~$\eps^2|I_x|/8$. By Lemma~\ref{lem:ptas-SPT_for_small_jobs}, the total processing time of these jobs is at most~$|I_x|$, and hence the number of jobs in~$\dirbset{S}{\dir}_{x}$ is bounded by a constant. The same is true for large jobs, by Lemma~\ref{lem:ptas-bounded_release_per_interval}. Finally, together with Lemma~\ref{lem:ptas-safety_net}, this implies that the number of jobs running during each interval is bounded by a constant.

  For the second property, we apply one time-stretch on the completion times. Consider now the latest job~$j$ of each direction that starts within block~$B_{t}$ and is completed in interval~$I_{c(j)}$ of the following block. By Lemma~\ref{lem:ptas-time-stretch} (and Table~\ref{tab:bidir-time-shift}) we know that there is idle time of at least~$\eps|I_{c(j)-2}|$ before the start of job~$j$ (or before the start of the earliest job aligned with~$j$ with completion time in~$I_{c(j)}$ and start time in~$B_{t}$. Hence, we can decrease the start time of these jobs such that the values~$\compl{j}$ and~$\start{j}+\proc{j}$ fall below the next~$\frac{1}{\eps^2}$ fraction of~$I_{c(j)}$, i.e., by an amount of at most~$\eps^2|I_{c(j)}| \leq \eps|I_{c(j)-2}|$. Hence, the first job starting in~$B_{t+1}$ (of each direction in case of compatibilities) can be scheduled at an~$\frac{1}{\eps^2}$ fraction of~$I_{c(j)}$ without any further loss. Thus, we only need to consider~$\frac{\cs}{\eps^2}$ possible frontier values per direction, or a total of~$\left(\frac{\cs}{\eps^2}\right)^2$ possible frontiers.
%
%
\end{proof}

\ifthenelse{\boolean{ptas-more}}{
\subsection{Multiple Segments}

We now consider the
bidirectional scheduling problem with a constant number~$m$ of segments and give detailed proofs for the required extensions to the single segment case where the argumentation is more complex.
We need to generalize or reformulate all of the above lemmas. 
We denote by~$\dirbset{S}{\dir}_{x,s,t}$ the set of all small jobs with respective source segment~$s$ and target segment~$t$ and direction~$\dir$ that are released at time~$R_x$ and we denote the corresponding large jobs by~$\dirbset{L}{\dir}_{x,s,t}$.

The statement of Lemma~\ref{lem:ptas-geometric_rounding} must be extended as in Lemma~\ref{lem:ptas-geometric_rounding-multiple}. Nevertheless, the proof works almost similar.

\begin{lemma}
\label{lem:ptas-geometric_rounding-multiple}
  With~$\O(1+\eps)$-loss we can assume that for each~$j\in\jobs$ and each~$i\in \{s_j, \dots, t_j\}$:
  \begin{compactitem}
    \item $\start{ij} \geq \eps(\proc{j}+\transit{i})$ defining a segment-wise release date~$\rel{ij}\geq r_j$,
    \item $\rel{ij}\geq 1$, and
    \item $\proc{j}, \rel{ij}\in\{(1+\eps)^x\mid x\in\N\}\cup\{0\}$.
  \end{compactitem}
We furthermore assume for each job that~$\rel{j}=\rel{s_jj}$ and that its segment-wise release dates do not decrease in heading of~$j$.
\end{lemma}
\begin{proof}
  Multiplying all segment-wise start times of a schedule by~$(1+\eps)$ increases the distance between any two distinct start times by a factor of~$(1+\eps)$. Therefore, we get a feasible schedule even when rounded up all nonzero processing times 
  to the next power of~$(1+\eps)$. This increases the total completion time less than a factor of~$(1+\eps)$.

  Now, we define for each job~$j$ and each segment~$i$ new completion times~$\compl{ij}'=(1+\eps)\compl{ij}$ and obtain increased start times~$\start{ij}' = (1+\eps)\compl{ij} - (\proc{j}+\transit{i}) \geq (1+\eps)\start{ij} + \eps(\proc{j} + \transit{i})$
  Since each space between two pairwise distinct completion times is increased by a factor of~$1+\eps$ the schedule remains feasible.  In particular, it holds that~$\start{ij}'\geq \compl{hj}'$ if~$j$ has to pass segment~$h$ before segment~$j$.
  Hence, we define a release date~$\rel{ij}:=\max\{\rel{j},\rel{hj},\eps(\proc{j}+\transit{i})\}$ for segment~$i$.

  Without any loss, we can scale all times by some power of~$(1+\eps)$ if necessary such that the earliest release date is at least one (since job parts with $\rel{j}=\proc{j}=\transit{s_j}=0$ can be omitted). At a further loss of~$(1+\eps)$ we finally can round the release dates again to the next power of~$(1+\eps)$.  
\end{proof}

By this, an identical proof as for Lemma~\ref{lem:ptas-bounded_crossing} yields the following.

\begin{lemma}
  \label{lem:ptas-bounded_crossing-multiple}
    Each part runs for at most~$\cs:=\lceil\log_{1+\eps}\frac{1+\eps}{\eps} \rceil$ intervals, i.e.\ a part starting in interval~$I_x$ is completed before the end of~$I_{x+\cs}$.
\end{lemma}

The proof of Lemmas~\ref{lem:ptas-SPT_for_small_jobs} becomes significantly more involved.

\begin{lemma}\label{lem:ptas-SPT_for_small_jobs_more_segments}
  With~$\O(1+\eps)$-loss we can restrict to schedules such that for each direction~$\dir\in\{\rb,\lb\}$, each source and target pair~$s,t\in\{1, \dots, m\}$, and each~$x\geq 0$:
  \begin{compactitem} 
    \item the jobs contained in~$\dirbset{S}{\dir}_{x,s,t}$ are scheduled on each segment~$i\in\{s, \dots, t\}$ in SPT order, i.e.,~$\start{i,j_1}\leq\start{i,j_2}$ for any pair of jobs~$j_1, j_2\in\dirbset{S}{\dir}_{x,s,t}$ with~$\proc{j_1}<\proc{j_2}$, and 
    \item the jobs of~$\dirbset{S}{\dir}_{x,s,t}$ in SPT order are joined in the complete schedule to unsplittable packages with size of at most~$\eps^2|I_x|$ for all packages and at least~$\eps^2|I_x|/2$ for all but the last packages.
  \end{compactitem}
\end{lemma}
\begin{proof}
  The proof works in principle as the proof of Lemma~\ref{lem:ptas-SPT_for_small_jobs}. During the rearrangement procedure we have to ensure that each interval on the target segment is filled with at least the same volume of small jobs as before, as long as jobs are unscheduled. For this, the jobs in the demanded order must have arrived at the respective segment in time. To ensure this property we have to deal with the following two difficulties. A convoy of very small jobs can be replaced by one (larger) small job. This job can only continue its processing on the next segment after its completion on the previous segment, while the first very small job could already start on the next segment before the last very small job is completed on the previous segment, cf.\ Figure~\ref{fig:rearrangement}. The other difficulty arises from the fact that some jobs scheduled within interval~$I_x$ arrive during  interval~$I_{x'}$ at the next segment, and the remaining jobs arrive one interval later. Since all but the last gaps within interval~$I_x$ have been decreased a bit within the rearrangement and only the last gap covers the lost volume there might be not enough volume available for the next segment in~$I_{x'}$.

  \begin{figure}
    \begin{center}
      \subfigure[A convoy of very small jobs.]{
\begin{tikzpicture}[xscale=.65, yscale=-.5]
  
  \edef\pxstart{0}
  \edef\pxend{8}
  \edef\pystart{0}
  \edef\pyend{10}

  \edef\ni{0}
  \edef\li{2}
  \edef\nii{2}
  
  \edef\niip{3}
  \edef\lii{6}
  \edef\niii{9}

  \draw[thick] (\ni,\pyend) -- (\ni,\pystart) -- node[font=\small,anchor=south]{$i_1$} (\nii,\pystart) -- (\nii,\pyend);

  \draw[thick] (\niip,\pyend) -- (\niip,\pystart) -- node[font=\small,anchor=south]{$i_1$} (\niii,\pystart) -- (\niii,\pyend);



  \foreach \length/\starti/\startii in {
      .2/.3/2.5,
      .1/.5/2.7,
      .2/.6/2.8,
      .15/.8/3.}{
    \updstripjob{red}{\nii}{\starti}{\length}{\li}{}{\li}
    \updstripjob{red}{\niii}{\startii}{\length}{\lii}{}{\lii}
  }
\end{tikzpicture}
}
\hspace*{1cm}
\subfigure[Replacement by a larger small job or a package.]{
\begin{tikzpicture}[xscale=.65, yscale=-.5]
  
  \edef\pxstart{0}
  \edef\pxend{8}
  \edef\pystart{0}
  \edef\pyend{10}

  \edef\ni{0}
  \edef\li{2}
  \edef\nii{2}
  
  \edef\niip{3}
  \edef\lii{6}
  \edef\niii{9}

  \draw[thick] (\ni,\pyend) -- (\ni,\pystart) -- node[font=\small,anchor=south]{$i_1$} (\nii,\pystart) -- (\nii,\pyend);

  \draw[thick] (\niip,\pyend) -- (\niip,\pystart) -- node[font=\small,anchor=south]{$i_1$} (\niii,\pystart) -- (\niii,\pyend);



  \foreach \length/\starti/\startii in {
      .65/.3/2.95}{
    \updstripjob{red}{\nii}{\starti}{\length}{\li}{}{\li}
    \updstripjob{red}{\niii}{\startii}{\length}{\lii}{}{\lii}
  }
  \end{tikzpicture}
}
      \caption{Illustration of the start times postponement when rearranging small jobs on multiple segments.}
      \label{fig:rearrangement}
    \end{center}
  \end{figure}

  To deal with the second difficulty, we allow for the moment a small job to start already~$p_j$ time units earlier than the completion on the previous segment. 
  We employ the following procedure for each direction~$\dir\in\{\rb, \lb\}$. First, we apply~$m^2$ time-stretches. Now, consider~$\dirbset{S}{\dir}_{x,s,t}$ for each source target pair~$s,t=1, \dots, m$  compatible with~$\dir$ and each~$x\geq 0$. If~$s=t$, we can simply apply the same procedure as for Lemma~$\ref{lem:ptas-SPT_for_small_jobs}$. Otherwise, proceed as follows (cf.\ Figure~\ref{fig:early-enough}). Define~$p(i,\tilde x)$ to be the amount of processing time from~$\dirbset{S}{\dir}_{x,s,t}$ scheduled within~$I_{\tilde x}$ on segment~$i$. For each reasonable combination of~$i_1$ and a succeeding~$i_2$ and~$x_1\leq x_2$ define~$p(i_1, x_1, i_2, x_2)$ to be the amount of processing time of jobs in~$\dirbset{S}{\dir}_{x,s,t}$ scheduled on segment~$i_1$ in interval~$I_{x_1}$ and on segment~$i_2$ on interval~$I_{x_2}$. On the other hand, let~$u(i_1, x_1, i_2, x_2)$ be the latest point for the end of processing of a small job within interval~$I_{x_1}$ on segment~$i_1$, such that it still can be completely processed within interval~$I_{x_2}$ on segment~$i_2$ if possible. Since the interval sizes are increasing with time, there is at most one interval~$I_{x_2}$ on segment~$i_2$ that yields an upper bound below~$R_{x_1+1}$.

  \begin{figure}
    \begin{center}
      \begin{tikzpicture}[xscale=.65, yscale=-.5]
  
  \edef\pxstart{-.4}
  \edef\pxend{11.4}
  \edef\pystart{-.4}
  \edef\pyend{17.4}

  \edef\ni{0}
  \edef\li{2}
  \edef\nii{2}
  \edef\lii{6}
  \edef\niii{8}
  \edef\liii{3}
  \edef\niv{11}

  \edef\rxi{0}
  \edef\rxii{5}
  \edef\rxiii{10.5}
  \edef\rxiv{17}

  \coordinate (start-i-i) at (\ni,\rxi);
  \path (start-i-i) ++(45:1) coordinate (dir-i-i);
   \path (start-i-i) -- (intersection of {start-i-i}--{dir-i-i} and \niv,0--\niv,1) coordinate (trans-i-i);
  \coordinate (start-ii-i) at (\ni,\rxii);
  \path (start-ii-i) ++(45:1) coordinate (dir-ii-i);
  \path (start-ii-i) -- (intersection of {start-ii-i}--{dir-ii-i} and \niv,0--\niv,1) coordinate (trans-ii-i);
  \fill[black!10] (start-i-i) -- (trans-i-i) -- (trans-ii-i) -- (start-ii-i);


  \draw[thick] (\ni,\pystart) -- (\niv,\pystart);
  \draw[thick] (\ni,\pystart) -- (\ni,\pyend);
  \draw[thick] (\nii,\pystart) -- (\nii,\pyend);
  \draw[thick] (\niii,\pystart) -- (\niii,\pyend);
  \draw[thick] (\niv,\pystart) -- (\niv,\pyend);

  \draw[dashed] (\pxstart,\rxi) -- (\pxend,\rxi);
  \draw[dashed] (\pxstart,\rxii) -- (\pxend ,\rxii);
  \draw[dashed] (\pxstart,\rxiii) -- (\pxend ,\rxiii);
  \draw[dashed] (\pxstart,\rxiv) -- (\pxend ,\rxiv);

  \path (\ni,\pystart) -- node[font=\small,anchor=south] {$i_1$} (\nii,\pystart);
  \path (\nii,\pystart) -- node[font=\small,anchor=south] {$i_1+1$} (\niii,\pystart);
  \path (\niii,\pystart) -- node[font=\small,anchor=south] {$i_2$} (\niv,\pystart);
  \node[font=\small,anchor=west] at (\pxend,\rxi) {$R_{x_1}$};
  \node[font=\small,anchor=west] at (\pxend,\rxii) {$R_{x_2}$};
  \node[font=\small,anchor=west] at (\pxend,\rxiii) {$R_{x_2+1}$};
  \node[font=\small,anchor=west] at (\pxend,\rxiv) {$R_{x_2+2}$};

  \coordinate (start-iii-iii) at (\niii,\rxiii);
  \path (start-iii-iii) ++(45:1) coordinate (dir-iii-iii);
  \draw[thick,dotted] (start-iii-iii) -- (intersection of {start-iii-iii}--{dir-iii-iii} and \ni,0--\ni,1) node[font=\small,anchor=east] {$u(x_1,i_1,x_2,i_2)$};


  \foreach \length/\starti/\startii/\startiii in {
      .2/.3/3.2/9.7,
      .1/.5/3.4/9.9,
      .4/.6/3.5/10.}{
    \updstripjob{red}{\nii}{\starti}{\length}{\li}{}{\li}
    \updstripjob{red}{\niii}{\startii}{\length}{\lii}{}{\lii}
    \updstripjob{red}{\niv}{\startiii}{\length}{\liii}{}{\liii}
  }

  \path (\ni,.3) -- node[font=\small,anchor=east]{$p(x_1,i_1,x_2,i_2)$} (\ni,1);

  \foreach \length\start in {
      .2/5.5,
      .2/5.7,
      .4/6.3,
      .1/6.7,
      .3/6.8,
      .2/8.7,
      .1/8.9,
      .1/9.0,
      .3/9.1
  }{
    \updstripjob{blue}{\niv}{\start}{\length}{\liii}{}{\liii}
  }

\end{tikzpicture}
      \caption{Example of three adjacent segments (drawn contiguously). The dotted line illustrates the construction of~$u(x_1,i_1,x_2,i_2)$. The processing times of all sketched small jobs on segment~$i_2$ sum up to~$p(i_2,x_2)$ whereas only the processing times of the red jobs contribute to~$p(x_1,i_1,x_2,i_2)$.}
      \label{fig:early-enough}
    \end{center}
  \end{figure}

  Remove all jobs of~$\dirbset{S}{\dir}_{x,s,t}$ from the schedule. To refill the gaps consider each segment~$i_1 = s, \dots, t-1$ and on segment~$i_1$ each interval~$I_{x_1}$ that contains gaps with increasing~$x_1$. For each succeeding segment~$i_2$ consider an interval~$I_{x_2}$.
  We now apply the gap filling procedure from the proof of Lemma~\ref{lem:ptas-SPT_for_small_jobs} but decrease and increase by at most one small job appropriately such that before each~$u(i_1, x_1, i_2, x_2)$ a volume of at least~$p(i_1, x_1, i_2, x_2)$ including the overage of the former intervals is ensured. More formally, let the \emph{overage}~$o(i_1, x_1, i_2, x_2)$ be defined as the difference of the scheduled volume before~$u(i_1, x_1, i_2, x_2)$ plus $o(i_1, x_1-1, i_2, x_2)$ and the demanded~$p(i_1, x_1, i_2, x_2)$ for~$x_1>x$ and zero otherwise. The refill now ensures that each~$o(i_1, x_1, i_2, x_2) \geq 0$. For~$i= s, \dots, t$ we furthermore ensure the analog for~$p(i,\tilde x)$, $o(i, \tilde x)$, and the end of~$I_{\tilde x}$.

  We now prove by induction over~$i=s, \dots, t$ that the required processing volume within each interval~$I_{\tilde x}, \tilde x\geq x$ on segment~$i$ is available. To be more precise we claim: for each segment~$i=s, \dots, t$ and each~$\tilde x\geq x$, enough jobs of~$\dirbset{S}{\dir}_{x,s,t}$ are available to ensure a volume of at least~$p(i, \tilde x, i_2, x_2)-o(i, \tilde x-1, i_2, x_2)$ until~$u(i, \tilde x, i_2, x_2)$ for each succeeding segment~$i_2$ until~$t$ and~$x_2\geq \tilde x$, and a volume of at least~$p(i,\tilde x)-o(i, \tilde x -1)$ within the complete interval, for as long as jobs are unscheduled. The claim is true for~$i=s$ since all jobs of~$\dirbset{S}{\dir}_{x,s,t}$ are available by~$R_x$. For some~$i\in\{s, \dots, t\}$ let~$i_1$ be the corresponding preceding segment. Consider a~$\tilde x\geq x$ with~$p(i, \tilde x)>0$. All the demanded volume must have been scheduled earlier on~$i_1$. Hence,~$p(i, \tilde x) = \sum_{x_1\leq \tilde x} p(i_1, x_1, i, \tilde x)$. By the induction hypothesis, this amount on the previous segment is scheduled within each interval in time (in particular within the former gaps). If~$\tilde x$ is the first one considered on~$i$, enough jobs with the required processing volume are available. Otherwise, there might be one small job missing that is scheduled in an earlier interval. This is covered by the overage. Consider now a succeeding~$i_2$ and an~$x_2$ with~$p(i, \tilde x, i_2, x_2)>0$. 
  If~$u(i, \tilde x, i_2, x_2)=r_{\tilde x+1}$ the claim ensues from~$p(i, \tilde x, i_2, x_2) \leq \sum_{x_1\leq \tilde x} p(i_1, x_1, i, \tilde x)$. Otherwise, note  that~$u(i, \tilde x, i_2, x_2)-\transit{i_1}\in I_{x_1}$ is equal to~$u(i_1, x_1, i_2, x_2)$. Then, we additionally have to use that~$p(i, \tilde x, i_2, x_2) \leq \sum_{x_1'\leq x_1} p(i_1, x_1, i_2, x_2)$ is scheduled to be processed before~$u(i_1, x_1, i_2, x_2)$. If one small missing job is scheduled again too early, we use the overage. To conclude, all jobs arrive in time on their target segment and the completion time is increased by a factor of at most~$\O(1+\eps)$. 

  The described rearrangement procedure for the jobs of~$\dirbset{S}{\dir}_{x,s,t}$ still only needs an extra time of~$\eps^2|I_x|$ on each segment for~$x\geq 0$ and~$s, t = 1, \dots, m$ compatible with~$\dir$. By Lemma~\ref{lem:A-ptas-interval_volume} we again get that~$m^2$ time-stretches are sufficient to cover this amount. To see this, consider some segment~$i$ and some interval~$y$ and bound the needed amount as follows:
  \begin{align*}
    \sum_{i_1\dirb{\prec}{\dir}i}\sum_{i\dirb{\preceq}{\dir}i_2}\sum_{x\leq y}\eps^2|I_x| & \leq \sum_{i_1\dirb{\prec}{\dir}i}\sum_{i\dirb{\preceq}{\dir}i_2}(\eps + \eps^2)|I_y|\\
    & \leq \sum_{i'=1}^m i'(\eps + \eps^2)|I_y|\leq m^2 \eps|I_y|\\
  \end{align*}

  Finally, we have to 
  resolve infeasible start times~$\start{ij}$ that fall below the corresponding completion time~$\compl{(i-1)j}$ on the previous segment.
  For this, we readjust those small jobs that have been scheduled a bit before the actual completion time on their previous segment. Since the respective error propagates from segment to segment we have to create an extra time window of~$m\eps^2|I_x|$ for each~$\dirbset{S}{\dir}_{x,s,t}$. We can do this by applying another~$2m^3$ time-stretches.
\end{proof}

Therefore, we get the following variant of Lemma~\ref{lem:ptas-bounded_release_per_interval} with a similar proof. 

\begin{lemma}
\label{lem:ptas-bounded_release_per_interval-multiple}
  With~$\O(1+\eps)$-loss we can assume for each interval~$I_x, x\geq 0$, each~$\dir\in\{\rb, \lb\}$, and each source and target pair~$s,t\in\{1, \dots, m\}$:
\begin{compactenum}
  \item $\proc{}(\dirbset{S}{\dir}_{x,s,t}) \leq (1+\eps^2)|I_x|$
  \item the number of possible processing times in~$\dirbset{L}{\dir}_{x,s,t}$ is bounded by $4\log_{(1+\eps)}\frac{1}{\eps}$, and
  \item the number of jobs per processing time in~$\dirbset{L}{\dir}_{x,s,t}$ is bounded by $\frac{1}{\eps^2}$.
\end{compactenum}
\end{lemma}

Also the construction of the safety net (Lemma~\ref{lem:ptas-safety_net}) must be adapted to the setting of multiple segments.

\begin{lemma}\label{lem:ptas-safety_net_more_segments}
  At~$\O(1+\eps)$-loss we can assume, that each part is completed on each segment within a constant number of intervals after its release date for the corresponding segment.
\end{lemma}
\begin{proof}
  We again start by applying one time-stretch. This creates extra space out of each interval for each segment. We want to assign for each segment-wise release date and each direction one interval per segment whose created extra space is able to host all corresponding parts that have not been scheduled on that segment. 
  Therefore, we first bound the transit time on segment~$i$ as follows. Let~$x_i$ be the smallest integer such that~$\eps\transit{i}\leq(1+\eps)^{x_i}$. Note that~$(1+\eps)^{x_i}\leq \rel{ij}$ for each job~$j$ with~$i\in\{s_j, \dots, t_j\}$ by Lemma~\ref{lem:ptas-geometric_rounding-multiple}. We get
  \begin{gather}\label{equ:ptas-safety_net-transit}
  \transit{i} \leq \frac{1}{\eps}(1+\eps)^{x_i} \leq \frac{\eps^2}{4}(1+\eps)^{x_i}(1+\eps)^{\cs'-1} \leq \frac{\eps}{4}|I_{x_i+\cs'-1}|
  \end{gather}
  for some integer~$\cs'$ with~$\frac{4}{\eps^3} \leq (1+\eps)^{\cs'-1}$.
  Furthermore, we want to bound the processing time of all jobs released at time~$R_x$ traveling in direction~$\dir\in\{\rb,\lb\}$ by Lemma~\ref{lem:ptas-bounded_release_per_interval-multiple}:
  \begin{gather}\label{equ:ptas-safety_net-proc_rel}
  \begin{split}
  \sum_{s=1}^m\sum_{t=1}^m\left[ p(\dirbset{S}{\dir}_{x,s,t}) + p(\dirbset{L}{\dir}_{x,s,t})\right] & \leq  m^2\bigg[(1+\eps^2)\eps(1+\eps)^x + \frac{1}{\eps}(1+\eps)^x\cdot\frac{4}{\eps^2}\log_{(1+\eps)}\frac{1}{\eps}\bigg] \\
  & = \quad  \eps^3(1+\eps)^x\cdot m^2\left[\frac{(1+\eps^2)}{\eps^2} + 
  \frac{4}{\eps^6}\log_{(1+\eps)}\frac{1}{\eps} \right]\\
  & \leq \quad \frac{\eps^3}{4}(1+\eps)^x(1+\eps)^{\cs'-1} = \frac{\eps^2}{4}|I_{x+\cs'-1}|,
  \end{split}
  \end{gather}
where we define~$\cs'$ to be such that~$\frac{4}{\eps^3} \leq (1+\eps)^{\cs'-1}$ and
\[
4m^2\left[\frac{(1+\eps^2)}{\eps^2} + \frac{4}{\eps^6}\log_{(1+\eps)}\frac{1}{\eps}\right] \leq (1+\eps)^{\cs'-1}\text{.}  
\]
Note that~$\cs'$ is a constant depending only on~$\eps$ and~$m$.

  We now apply the following strategy. For each direction~$\dir\in\{\rb,\lb\}$ iterate over the segments in the corresponding order~$\{i_1,\dots, i_m\}$ and assign for each segment~$i$ the jobs~$J_{ix}:=\{j \in \dirbset{J}{\dir} \mid \rel{ij}=R_{x}\}$ with release date~$R_{x}$ on~$i$ to one concrete created extra space.
  On segment~$i_1$, we can assign the unfinished jobs of~$J_{i_1x_1}$ for each~$x_1$ to the created space of~$\frac{\eps}{2}|I_{x_1+\cs'-1}|$ due to Equations~\eqref{equ:ptas-safety_net-transit} and~\eqref{equ:ptas-safety_net-proc_rel}. By Lemmas~\ref{lem:ptas-time-stretch} and~\ref{lem:ptas-bounded_crossing-multiple} we know that this space is available before the end of~$I_{x_1+\cs'+\cs}$. 
  Consider know the jobs of some~$J_{i_2x_2}$. Since by Lemma~\ref{lem:ptas-geometric_rounding-multiple}~$\rel{i_1j}\leq R_{x_2}$ for each~$j\in J_{i_2x_2}$ it is ensured that all preceding parts are finished before~$R_{x_2+\cs'+\cs+1}$. We now assign the created extra space of~$\frac{\eps}{2}|I_{x_2+\cs'+\cs}|$ to these parts which is available between~$R_{x_2+\cs'+\cs+1}$ and the end of~$I_{x_2+\cs'+2\cs+1}$ with the same argument. This space is sufficient to host one transit by Equation~\eqref{equ:ptas-safety_net-transit} and the processing of all jobs of~$J_{i_2x_2}$ with global release date~$r_j=R_{x}\leq R_{x_2}$ due to Equation~\eqref{equ:ptas-safety_net-proc_rel} and Lemma~\ref{lem:A-ptas-interval_volume}.
  Continuing analogously we get for each~$1 \leq k \leq m$ and each~$x_k$ that we can complete all unfinished parts of~$J_{i_kx_k}$ in the created space~$\frac{\eps}{2}|I_{x_k+\cs'+ (k-1)(\cs+1)-1}|$ before the beginning of interval~$I_{x_k+\cs'+ k(\cs+1)}$.
  

\end{proof}

To conclude that the number of parts running in each block~$B_t$ is bounded by a constant we finally have to use that the ratio of any two transit times of segments is bounded by a constant.

\begin{lemma}
 Assume that~$\transit{k}/\transit{i}$ is bounded by a constant~$(1+\eps)^T$ for any two segments~$k, i\in\{1, \dots, m\}$. Then there is only a constant number of intervals between each two segment-wise release dates of each job~$j$:
  \[
    \frac{\rel{kj}}{\rel{ij}} \leq (1+\eps)^{T+1}\text{.}
  \]
\end{lemma}
\begin{proof}
Let~$k\in\{s_j,\dots, t_j\}$ be the segment with maximum transit time. If~$\rel{kj}=\rel{j}$ we know that all segment-wise release dates are equal. Otherwise we can conclude by construction in the proof of Lemma~\ref{lem:ptas-geometric_rounding-multiple} that~$\eps(\proc{j}+\transit{k}) \leq \rel{kj} \leq (1+\eps)\eps(\proc{j}+\transit{k})$. Therefore, we get
  \[
   \rel{kj} \leq (1+\eps)\eps(\proc{j}+\transit{k}) \leq (1+\eps)\eps(\proc{j}+(1+\eps)^T\transit{s_j})\leq (1+\eps)^{T+1}\eps(\proc{j}+\transit{s_j})\leq (1+\eps)^{T+1}\rel{j}.
  \]
Applying that~$\rel{j} \leq \rel{ij} \leq \rel{kj}$ for each~$i\in\{s_j,\dots, t_j\}$ yields finally the claim.
\end{proof}
}{}

\section{Proofs of Section~\ref{sec:arbitrary-conflicts}:\newline Hardness of custom compatibilities}

In this section we give a detailed hardness proof for bidirectional scheduling on a single segment where jobs can be compatible. Our proof holds even for unit processing and transit times. We first consider the makespan objective and extend the proof in a second step to waiting time and total completion time.

\subsection{$\NP$-Hardness of Makespan Minimization}

\begin{restatable}{theorem}{thmGraphHardness}
\label{thm:graph-hardness}
Minimizing the makespan for~$m=1$ with an arbitrary compatibility graph~$G_1$ is $\NP$-hard even if~$\proc{j}=\transit{1}=1$ for each~$j\in\jobs$.
\end{restatable}

In the following, we explain the construction of a bidirectional scheduling instance for a given~\rsat{\leq 3}{3} instance with variable set $X = \{ x_i \mid i = 0,\dots,|X|-1\}$ and clause set $C = \{c_k \mid k = 0,\dots,|C|-1\}$. The constructed instance yields a demanded makespan~$\cmax$ if and only if the given \rsat{\leq 3}{3} formula is satisfiable.
For the construction, we partition the time horizon into four parts~$\tpart{1}, \dots, \tpart{4}$ with start time~$\tstart{1}=0, \tstart{2}=6\varn, \tstart{3}=10\varn,$ and $\tstart{4}=10\varn+2\cln$. There is a (virtual) last part starting at time~$\tstart{5}=12\varn+\cln$. The demanded makespan~$\cmax=\tstart{5}+1$ will enforce that all jobs start before the end of the fourth part.

The rough idea is as follows: In the first four parts we release a tight frame of \emph{blocking jobs}~$\blockingjs{}$ and \emph{dummy jobs}~$\dummyjs{}$ that have to start running immediately at their release date in any schedule that achieves~$\cmax$. We use these jobs to create gaps for \emph{variable jobs} that represent the variable assignments. By defining the compatibilities for the blocking jobs we are able to control which of these variable jobs can be scheduled into each gap. In the first part of our construction, we release all variable jobs, which come in two \emph{types}: one type representing a \emph{true} assignment to the corresponding variable and the other type representing a \emph{false} assignment. Our construction will enforce the following properties in each of its parts:


\begin{restatable}{lemma}{lemWellDefiniedAssignment}
	\label{lem:well-defined-assignment}
  In every feasible schedule with makespan~$\cmax$, all jobs released before~$\tstart{3}$ are scheduled in parts~$\tpart{1}$ and $\tpart{2}$, except for two rightbound variable jobs of same type for each variable.
\end{restatable}


\begin{restatable}{lemma}{lemClauses}\label{lem:clauses}
  In every feasible schedule with makespan~$\cmax$, the only jobs released before~$\tstart{3}$ and scheduled in~$\tpart{3}$ are rightbound variable jobs each corresponding to a variable assignment satisfying a different clause.
\end{restatable}

\begin{restatable}{lemma}{lemStorage}\label{lem:storage}
  In every feasible schedule with makespan~$\cmax$, the only jobs released before~$\tstart{4}$ and scheduled in~$\tpart{4}$ are rightbound variable jobs, and there are not more than~$2|X|-|C|$ of them.
\end{restatable}


In the following we explicitly define the released jobs of each part achieving the above properties. Each part is accompanied by a figure illustrating when jobs are released, the respective compatibility graph and an example of a schedule. In all figures, time is directed downwards, and all rightbound jobs are depicted to the left and all leftbound jobs to the right of the segment. Since compatible jobs can run concurrently, the schedules of the leftbound and the rightbound jobs are drawn separately. 

It is convenient to prove Lemmas 6 to 8 in reverse order. To this end, we start by specifying the jobs released in~$\tpart{4}$.

\subsubsection*{Jobs released in $\tpart{4}$.}

In the fourth part, we release a set of $2\varn-\cln$ leftbound blocking jobs~$\blockingjs{4} = \{\blockingjob{i}{} \mid i=0, \dots, 2\varn-\cln-1\}$. Each blocking job $b_i$ is released at time~$\tstart{4} + i$. The purpose of a blocking job is to leaving space for a leftover rightbound variable job that has not been scheduled until the beginning of this part. Each blocking jobs $b_i \in B_4$ is only compatible with all rightbound variable jobs.

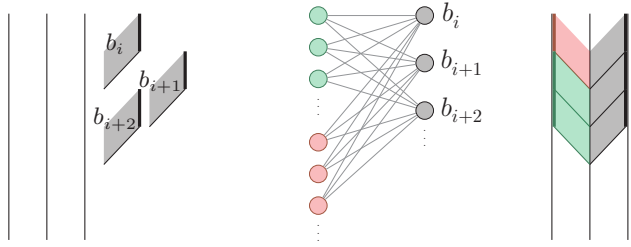
\begin{figure}
  \begin{center}
    \begin{tikzpicture}\label{pic:part4-released}
   \pgftransformxscale{\procheight}
   \pgftransformyscale{-\timeunit}

   \draw  (-1,0) -- +(0,6)  
          (0,0) -- +(0,6)
          (1,0) -- +(0,6)  ;

  \downjob{black} {1.5}{0}{1}{1}{}
  \downjob{black} {2.7}{1}{1}{1}{}
  \downjob{black} {1.5}{2}{1}{1}{}


  \begin{scope}[font = \small,anchor=center]
    \node (lbi) at (1.8, .8){$\blockingjob{i}{}$};
    \node (lbip1) at (3., 1.8){$\blockingjob{i+1}{}$};
    \node (lbip2) at (1.8, 2.8){$\blockingjob{i+2}{}$};
  \end{scope}
\end{tikzpicture}
\hspace*{1cm}
\scalebox{.7}{
\begin{tikzpicture}\label{pic:part4-graph}
  \pgftransformyscale{-.6}

  \begin{scope} 
  \clip (-1.5,-.5) rectangle (1.5,7.1);
 
  \gnode{green} {t1}{-1}{0}{}
  \gnode{green} {t2}{-1}{1}{}
  \gnode{green} {t3}{-1}{2}{}
  \node[font=\tiny] at (-1,2.7) {$\vdots$};
  
  \gnode{red}   {f1}{-1}{4}{}
  \gnode{red}   {f2}{-1}{5}{}
  \gnode{red}   {f3}{-1}{6}{}
  \node[font=\tiny] at (-1,6.7) {$\vdots$};
  
  \gnode{black} {b1}{1}{0}{}
  \gnode{black} {b2}{1}{1.5}{}
  \gnode{black} {b3}{1}{3}{}
  \node[font=\tiny] at (1,3.7) {$\vdots$};

  \foreach \i in {1,2,3}{
    \foreach \j in {1,2,3}{
      \draw[gray] (b\i) -- (t\j);
      \draw[gray] (b\i) -- (f\j);
    }
  }

  \end{scope} 

  \begin{scope}[font = \Large, anchor=center, overlay]
    \node (lbi) at (1.5, 0){$\blockingjob{i}{}$};
    \node (lbip1) at (1.7, 1.5){$\blockingjob{i+1}{}$};
    \node (lbip1) at (1.7, 3.){$\blockingjob{i+2}{}$};
  \end{scope}
\end{tikzpicture}
}
\hspace*{1cm}
\begin{tikzpicture}\label{pic:part4-sched}
   \pgftransformxscale{\procheight}
   \pgftransformyscale{-\timeunit}

   \draw  (-1,0) -- +(0,6)  
          (0,0) -- +(0,6)
          (1,0) -- +(0,6)  ;

  \upjob{red}   {0}{0}{1}{1}{}
  \upjob{green} {0}{1}{1}{1}{}
  \upjob{green} {0}{2}{1}{1}{}

  \downjob{black} {0}{0}{1}{1}{}
  \downjob{black} {0}{1}{1}{1}{}
  \downjob{black} {0}{2}{1}{1}{}

\end{tikzpicture} 
  \end{center}
  \caption{Part $\tpart{4}$ with blocking jobs reserving space for all remaining rightbound variable jobs.}
  \label{fig:part4}
\end{figure}

We are now in position to prove Lemma~\ref{lem:storage}, i.e., in a schedule with makespan $\cmax$ the only jobs released before $A_4$ that can be scheduled in $P_4$ are up to $2|X|-|C|$ rightbound variable jobs.

\begin{proof}[Proof of Lemma~\ref{lem:storage}]
  First, observe that with the required makespan of~$\tstart{5}+1 = \tstart{4}+2\varn-\cln+1$ each blocking job of~$\blockingjs{4}$ must be scheduled directly at its release date. Consequently, there is no room to delay the start of any leftbound job released before~$\tpart{4}$ to this part.  Due to the compatibilities, the rightbound blocking and dummy jobs released before~$\tpart{4}$ are also forced to run before the start of~$\tpart{4}$. Therefore, there are exactly~$2\varn-\cln$ open slots within~$\tpart{4}$ reserved for rightbound variable jobs. 
\end{proof}
 
 We proceed to explain the jobs released in the third part of our construction.

\subsubsection*{Jobs released in $\tpart{3}$.}
The third part (Figure~\ref{fig:part3}) is responsible for the assignment of satisfying literals to each clause. During that part, we release a set of blocking jobs~$\blockingjs{3} = \{\blockingjob{k}{} \mid k=0,\dots,|C|-1\}$ which contains one leftbound blocking job~$\blockingjob{k}{}$ for each clause~$c_k$. Each blocking job $\blockingjs{k}{}$ is released at time~$\tstart{3} + 2k$ and is compatible with each rightbound variable job that represents a variable assignment that satisfying the corresponding clause~$c_k$. The gaps between the release times of the blocking jobs are filled with a set dummy jobs~$\dummyjs{3} = \{\dummyjob{k}{\rb} \mid k = 0,\dots,|C|-1\} \cup \{\dummyjob{k}{\lb} \mid k = 0,\dots,|C|-1\}$ containing one rightbound job~$\dummyjob{k}{\rb}$ and one leftbound job~$\dummyjob{k}{\lb}$ with release date~$\tstart{3} + 2k + 1$ each. Each leftbound dummy job $\dummyjob{k}{\lb}$ is compatible with all rightbound variable jobs, furthermore each rightbound dummy job~$\dummyjob{k}{\rb}$ is compatible with the three leftbound jobs released during the time interval~$[r_{\dummyjob{k}{\rb}}-1, r_{\dummyjob{k}{\rb}}+1]$.

\begin{figure}
  \begin{center}
    \begin{tikzpicture}\label{pic:part3-released}
   \pgftransformxscale{\procheight}
   \pgftransformyscale{-\timeunit}

   \draw  (-1,0) -- +(0,6)  
          (0,0) -- +(0,6)
          (1,0) -- +(0,6)  ;

  \downjob{black}     {2.7}{0}{1}{1}{}
  \downjob{lightgray} {1.5}{1}{1}{1}{}
  \downjob{black}     {2.7}{2}{1}{1}{}
  \downjob{lightgray} {1.5}{3}{1}{1}{}

  \upjob{lightgray} {-1.5}{1}{1}{1}{}
  \upjob{lightgray} {-1.5}{3}{1}{1}{}

  \begin{scope}[font = \small,anchor=center]
    \node (lbi) at (3.1, .8){$\blockingjob{k}{}$};
    \node (lbip1) at (3., 2.8){$\blockingjob{k+1}{}$};

    \node (lupdi) at (-2., 1.8){$\dummyjob{i}{\rb}$};
    \node (lupdip1) at (-1.8, 3.8){$\dummyjob{i+1}{\rb}$};

    \node (ldowndi) at (1.8, 1.8){$\dummyjob{i}{\lb}$};
    \node (ldowndip1) at (2., 3.8){$\dummyjob{i+1}{\lb}$};
  \end{scope}
  \begin{scope}[font = \small,color=gray,anchor=center, overlay]
    \node at (4.5,.9) {$c_{k}$};
    \node at (4.5,2.9) {$c_{k+1}$};
  \end{scope}

\end{tikzpicture}
\hspace*{1cm}
\scalebox{.7}{
\begin{tikzpicture}\label{pic:part3-graph}
  \pgftransformyscale{-.6}
  
  \begin{scope}
  \clip (-2.5,-.5) rectangle (1.5,7.1);
 
  \gnode{green}     {c1}{-1}{0}{}
  \gnode{green}     {c2}{-1}{1}{}
  \gnode{red}       {c3}{-1}{2}{}
  \gnode{green}     {cc1}{-1.5}{0}{}
  \gnode{green}     {cc2}{-1.5}{1}{}
  \gnode{red}       {cc3}{-1.5}{2}{}
    \gnode{black}   {b1}{1}{1}{}

  \gnode{lightgray} {d1}{1}{3}{}
  \gnode{lightgray} {d2}{-1}{3}{}
  
  \gnode{red}       {c4}{-1}{4}{}
  \gnode{green}     {c5}{-1}{5}{}
  \gnode{red}       {c6}{-1}{6}{}
  \gnode{red}     {cc4}{-1.5}{4}{}
  \gnode{green}     {cc5}{-1.5}{5}{}
  \gnode{red}       {cc6}{-1.5}{6}{}
  \gnode{black}     {b2}{1}{5}{}

  \gnode{lightgray} {d3}{1}{7}{}
  \gnode{lightgray} {d4}{-1}{7}{}


  \foreach \i in {1,2,...,6}{
    \draw[gray] (c\i) -- (d1);
    \draw[gray] (cc\i) -- (d1);
  }
    
  \draw[gray] (d2) -- (b1);
  \draw[gray] (d2) -- (d1);
  \draw[gray] (d2) -- (b2);

  \foreach \i in {1,2,3}{
    \draw[gray] (c\i) -- (b1);
    \draw[gray] (cc\i) -- (b1);
  }

  \foreach \i in {4,5,6}{
    \draw[gray] (c\i) -- (b2);
    \draw[gray] (c\i) -- (d3);
    \draw[gray] (cc\i) -- (b2);
    \draw[gray] (cc\i) -- (d3);
  }

  \draw[gray] (d4) -- (b2);
  \draw[gray] (d4) -- (d3);
  \end{scope} 

  \begin{scope}[font = \Large,anchor=center, overlay]
    \node (lbi) at (1.5, 1){$\blockingjob{k}{}$};
    \node (lbip1) at (1.7, 5){$\blockingjob{k+1}{}$};

    \node (ta) at (-2.1, 0){$\truejob{a,*}{\rb}$};
    \node (tb) at (-2.1, 1){$\truejob{b,*}{\rb}$};
    \node (fc) at (-2.1, 2){$\falsejob{c,*}{\rb}$};

    \node (fd) at (-2.1, 4){$\falsejob{d,*}{\rb}$};
    \node (te) at (-2.1, 5){$\truejob{e,*}{\rb}$};
    \node (ff) at (-2.1, 6){$\falsejob{f,*}{\rb}$};
%
%
%
  \end{scope}
\end{tikzpicture}
}
\hspace*{1cm}
\begin{tikzpicture}\label{pic:part3-sched}
   \pgftransformxscale{\procheight}
   \pgftransformyscale{-\timeunit}

   \draw  (-1,0) -- +(0,6)  
          (0,0) -- +(0,6)
          (1,0) -- +(0,6)  ;

  \upjob{green}     {0}{0}{1}{1}{}
  \upjob{lightgray} {0}{1}{1}{1}{}
  \upjob{red}       {0}{2}{1}{1}{}
  \upjob{lightgray} {0}{3}{1}{1}{}

  \downjob{black}     {0}{0}{1}{1}{}
  \downjob{lightgray} {0}{1}{1}{1}{}
  \downjob{black}     {0}{2}{1}{1}{}
  \downjob{lightgray} {0}{3}{1}{1}{}

\end{tikzpicture} 
  \end{center}
  \caption{Part $\tpart{3}$ for $c_k = (x_a \vee x_b \vee \bar x_c)$ and
           $c_{k+1} = (\bar x_d \vee x_e \vee \bar x_f)$. Note that each variable job can be adjacent
with more than one clause job (although this does not occur in the example).}
  \label{fig:part3}
\end{figure}
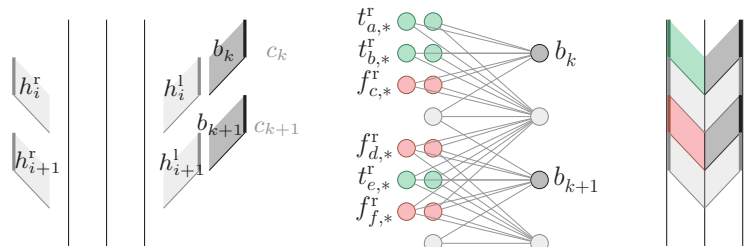

We are now in position to prove Lemma~\ref{lem:clauses}, i.e., in a schedule with makespan $\cmax$ the only jobs released before $A_3$ that can be scheduled in $P_3$ are one rightbound variable job for each clause such that the variable assignment satisfies the clause.

\begin{proof}[Proof of Lemma~\ref{lem:clauses}]
  By Lemma~\ref{lem:storage} all jobs released within~$\tpart{3}$ must start before the end of~$\tpart{3}$. Hence, each leftbound dummy and blocking job is forced to start at its release date. Therefore, due to the compatibilities, each rightbound dummy job must be scheduled directly when released. The only remaining~$\cln$ free slots can be filled with rightbound variable jobs -- exactly one free slot per clause~$c_k$ reserved for a variable job representing an assignment that satisfies~$c_k$.
\end{proof}

We proceed to explain the jobs released in parts $\tpart{1}$ and~$\tpart{2}$.

\subsubsection*{Jobs released in $\tpart{1}$.}
The first two parts are responsible for obtaining a correct assignment of the variables. In the first part, we release different types of jobs for each variable $x_i$, $i =0,\dots, |X|-1$, cf.\ Figure~\ref{fig:part1} with the following. For each variable $x_i$, $i = 0,\dots, |X|-1$, we release
\begin{itemize}
\item two rightbound true variable jobs $\truejob{i,1}{\rb}$, $\truejob{i,2}{\rb}$ at times $6i$ and $6i+1$, respectively,
\item two rightbound false variable jobs $\falsejob{i,1}{\rb}$, $\falsejob{i,2}{\rb}$ at times $6i+3$ and $6i+4$, respectively,
\item one leftbound true variable job $\smash{\truejob{i}{\lb}}$ at time~$6i+4$,
\item one leftbound false variable job $\smash{\falsejob{i}{\lb}}$ at time $6i+1$,
\item two leftbound indefinite variable jobs $\indefjob{i}{\mtrue}$, $\indefjob{i}{\mfalse}$ at times $6i+1$ and $6i+4$, respectively.
\item two leftbound blocking jobs $\blockingjob{i}{\mtrue}$, $\blockingjob{i}{\mfalse}$ at times $6i$ and $6i+3$, respectively.
\item two leftbound dummy jobs $\dummyjob{{i}}{\lb\mtrue}$, $\dummyjob{i}{\lb\mfalse}$ at times $6i+2$ and $6i+5$, respectively.
\item two rightbound dummy jobs $\dummyjob{i}{\rb\mtrue}$, $\dummyjob{i}{\rb\mfalse}$ at times $6i+2$ and $6i+5$, respectively.
\end{itemize}
In the following, we write $\truejs{\rb} = \{ \truejob{i,1}{\rb}, \truejob{i,2}{\rb} \mid x_i\in \set{X}\}$ for the set of rightbound true variable jobs, $\falsejs{\rb} = \{ \falsejob{i,1}{\rb}, \falsejob{i,2}{\rb} \mid x_i\in \set{X}\}$ for the set of rightbound false variable jobs and $\indefjs = \{ \indefjob{i}{\mtrue}, \indefjob{i}{\mfalse} \mid x_i\in \set{X}\}$ for the set of indefinite jobs.

  The compatibility graph~$G_1$ is defined such that
  \begin{itemize}
  \item	each blocking job~$\blockingjob{i}{\mtrue}$ is compatible with the corresponding true variable jobs $\truejob{i,1}{\rb}$ and~$\truejob{i,2}{\rb}$,
  \item each blocking job $\blockingjob{i}{\mfalse}$ is compatible with with the corresponding false variable jobs $\falsejob{i,1}{\rb}$ and~$\falsejob{i,2}{\rb}$,
  \item each indefinite job~$\indefjob{i}{\mtrue}$ is compatible with the corresponding rightbound true variable jobs~$\truejob{i,1}{\rb}$ and~$\truejob{i,2}{\rb}$
  \item each indefinite job $\indefjob{i}{\mfalse}$ is compatible with the corresponding rightbound false variable jobs $\falsejob{i,1}{\rb}$ and~$\falsejob{i,2}{\rb}$.
  \item each dummy job~$h$ is compatible with the opposed jobs released in~$[r_{h}-1, r_{h}+1]$,
  \item none of the remaining pairs of jobs are compatible.
  \end{itemize}
  
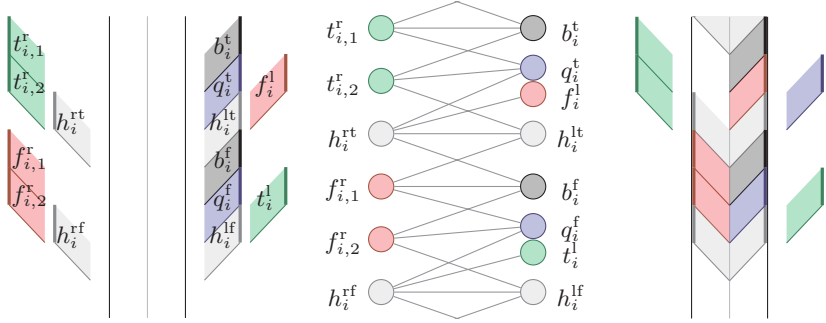
\begin{figure}
  \begin{center}
    \begin{tikzpicture}\label{pic:part1-released}
	 \pgftransformxscale{\procheight}
	 \pgftransformyscale{-\timeunit}
  

	 \draw  (-1,0) -- +(0,8) (1,0) -- +(0,8);
   \draw[draw=black!30](0,0) -- +(0,8);
  
  \upjob{green}{-2.7}{0}{1}{1}{}
  \upjob{green}{-2.7}{1}{1}{1}{}
  \upjob{lightgray}{-1.5}{2}{1}{1}{}
  \upjob{red}{-2.7}{3}{1}{1}{}
  \upjob{red}{-2.7}{4}{1}{1}{}
  \upjob{lightgray}{-1.5}{5}{1}{1}{}

  \begin{scope}[font = \small,anchor=center]
    \node (lt1) at (-3.1, .8){$\truejob{i,1}{\rb}$};
    \node (lt2) at (-3.1, 1.8){$\truejob{i,2}{\rb}$};
    \node (ld1) at (-2., 2.8){$\dummyjob{i}{\rb\mtrue}$};
    \node (lf1) at (-3.1, 3.8){$\falsejob{i,1}{\rb}$};
    \node (lf2) at (-3.1, 4.8){$\falsejob{i,2}{\rb}$};
    \node (ld2) at (-2., 5.8){$\dummyjob{i}{\rb\mfalse}$};
  \end{scope}

  \downjob{black}{1.5}{0}{1}{1}{}
  \downjob{blue}{1.5}{1}{1}{1}{}
  \downjob{red}{2.7}{1}{1}{1}{}
  \downjob{lightgray}{1.5}{2}{1}{1}{}
  \downjob{black}{1.5}{3}{1}{1}{}
  \downjob{blue}{1.5}{4}{1}{1}{}
  \downjob{green}{2.7}{4}{1}{1}{}
  \downjob{lightgray}{1.5}{5}{1}{1}{}

  \begin{scope}[font = \small,anchor=center]
    \node (lb1) at (2., .8){$\blockingjob{i}{\mtrue}$};
    \node (lb2) at (2., 3.8){$\blockingjob{i}{\mfalse}$};
    \node (ld3) at (2., 2.8){$\dummyjob{i}{\lb\mtrue}$};
    \node (ld4) at (2., 5.8){$\dummyjob{i}{\lb\mfalse}$};
    \node (lit) at (2., 1.8){$\indefjob{i}{\mtrue}$};
    \node (lf) at  (3.1, 1.8){$\falsejob{i}{\lb}$};
    \node (lif) at (2., 4.8){$\indefjob{i}{\mfalse}$};
    \node (lt) at (3.1, 4.8){$\truejob{i}{\lb}$};
  \end{scope}
\end{tikzpicture}
\hspace*{.5cm}
\begin{tikzpicture}\label{pic:part1-graph}
  \pgftransformyscale{-.7}
  
  \begin{scope}
  \clip (-1.5,-.5) rectangle (1.5,5.5);

  \gnode{lightgray} {dm2}{-1}{-1}{};
  \gnode{lightgray} {dm4}{1}{-1}{};

  \gnode{green}     {t1}{-1}{0}{}
  \gnode{green}     {t2}{-1}{1}{}
  \gnode{lightgray} {d1}{-1}{2}{};
  \gnode{black}     {b1}{1}{0}{};
  \gnode{blue}      {it}{1}{0.75}{};
  \gnode{red}       {lf}{1}{1.25}{};
  \gnode{lightgray} {d3}{1}{2}{};

  \gnode{black}     {b2}{1}{3}{};
  \gnode{blue}      {if}{1}{3.75}{};
  \gnode{green}      {lt}{1}{4.25}{};
  \gnode{lightgray} {d4}{1}{5}{};
  \gnode{red}       {f1}{-1}{3}{};
  \gnode{red}       {f2}{-1}{4}{};
  \gnode{lightgray} {d2}{-1}{5}{};

  \gnode{green}     {t3}{-1}{6}{};
  \gnode{black}     {b3}{1}{6}{};

  \draw[gray] (dm2) -- (b1)
              (t1) -- (dm4)
              (t1) -- (b1)
              (t2) -- (b1)
              (t2) -- (d3)
              (t1) -- (it)
              (t2) -- (it)
              (d1) -- (it)
              (d1) -- (lf)
              (d1) -- (d3)
              (d1) -- (b2)
              (f1) -- (d3)
              (f1) -- (b2)
              (f2) -- (b2)
              (f2) -- (d4)
              (t3) -- (d4)
              (f1) -- (if)
              (f2) -- (if)
              (d2) -- (if)
              (d2) -- (lt)
              (d2) -- (d4)
              (d2) -- (b3);
  \end{scope}

  \begin{scope}[font = \small,anchor=center, overlay]
    \node (lt1) at (-1.5, .1){$\truejob{i,1}{\rb}$};
    \node (lt2) at (-1.5, 1.1){$\truejob{i,2}{\rb}$};
    \node (lf1) at (-1.5, 3.1){$\falsejob{i,1}{\rb}$};
    \node (lf2) at (-1.5, 4.1){$\falsejob{i,2}{\rb}$};
    \node (ld1) at (-1.5, 2.1){$\dummyjob{i}{\rb\mtrue}$};
    \node (ld2) at (-1.5, 5.1){$\dummyjob{i}{\rb\mfalse}$};
    
    \node (lb1) at (1.5, .1){$\blockingjob{i}{\mtrue}$};
    \node (lb2) at (1.5, 3.1){$\blockingjob{i}{\mfalse}$};
    \node (lit) at (1.5, 0.8){$\indefjob{i}{\mtrue}$};
    \node (llf) at (1.5, 1.3){$\falsejob{i}{\lb}$};
    \node (lif) at (1.5, 3.8){$\indefjob{i}{\mfalse}$};
    \node (llt) at (1.5, 4.3){$\truejob{i}{\lb}$};
    \node (ld1) at (1.5, 2.1){$\dummyjob{i}{\lb\mtrue}$};
    \node (ld2) at (1.5, 5.1){$\dummyjob{i}{\lb\mfalse}$};
  \end{scope}
\end{tikzpicture}
\hspace*{.5cm}
\begin{tikzpicture}\label{pic:part1-sched}
   \pgftransformxscale{\procheight}
   \pgftransformyscale{-\timeunit}
  

  \clip (-2.7,0) rectangle (2.7,8.);

   \draw  (-1,0) -- +(0,8) (1,0) -- +(0,8);
   \draw[draw=black!30](0,0) -- +(0,8);
  
  \upjob{lightgray}{0}{-1}{1}{1}{}

  \upjob{green}{-1.5}{0}{1}{1}{}
  \upjob{green}{-1.5}{1}{1}{1}{}
  \upjob{lightgray}{0}{2}{1}{1}{}
  \upjob{red}{0}{3}{1}{1}{}
  \upjob{red}{0}{4}{1}{1}{}
  \upjob{lightgray}{0}{5}{1}{1}{}

  \downjob{lightgray}{0}{-1}{1}{1}{}

  \downjob{black}{0}{0}{1}{1}{}
  \downjob{lightgray}{0}{2}{1}{1}{}
  \downjob{blue}{1.5}{1}{1}{1}{}
  \downjob{red}{0}{1}{1}{1}{}
  \downjob{black}{0}{3}{1}{1}{}
  \downjob{lightgray}{0}{5}{1}{1}{}
  \downjob{blue}{0}{4}{1}{1}{}
  \downjob{green}{1.5}{4}{1}{1}{}
\end{tikzpicture} 
  \end{center}
  \caption{Released jobs per variable~$x_i$ in $\tpart{1}$, the corresponding
compatibilities given by~$G_1$ and a scheduled example for a true variable assignment.}
  \label{fig:part1}
\end{figure}

\subsubsection*{Jobs released in $\tpart{2}$.}
In the second part (Figure~\ref{fig:part2}), there is room for exactly one indefinite job and one leftbound variable job per variable. This is realized by a set of rightbound blocking jobs~$\blockingjs{2}=\{ \blockingjob{i,1}{},\blockingjob{i,2}{} \mid x_i\in \set{X}\}$ where each blocking job $\blockingjob{i,1}{}$ is released at time~$\tstart{2} + 4i$ and is compatible with the corresponding two indefinite jobs~$\indefjob{i}{\mtrue}$ and~$\indefjob{i}{\mfalse}$. Each blocking job $\blockingjob{i,2}{}$ is released at time $\tstart{2} + 4i+2$ and is compatible with the corresponding two leftbound variable jobs~$\falsejob{i}{\lb}$ and~$\truejob{i}{\lb}$.
The gaps between two subsequent released blocking jobs are closed in both directions by dummy jobs~$\dummyjs{2} = \{\dummyjob{i,1}{\rb}, \dummyjob{i,1}{\lb} \mid x_i \in \set{X}\} \cup \{\dummyjob{i,2}{\rb}, \dummyjob{i,2}{\lb} \mid x_i\in \set{X}\}$ released at times $\tstart{2} + 4i + 1$ and~$\tstart{2} + 4i + 3$, respectively. Each dummy job is compatible with all jobs of~$\indefjs, \truejs{\lb}, \falsejs{\lb}$, or~$\blockingjs{2}$ and the corresponding opposed dummy job released concurrently.

We are now in position to prove Lemma~\ref{lem:well-defined-assignment}.

\label{prf:well-defined-assignment}
\begin{proof}[Proof of Lemma~\ref{lem:well-defined-assignment}]
  By Lemmas \ref{lem:clauses} and \ref{lem:storage}, each rightbound dummy and blocking job of~$\dummyjs{2}$ and~$\blockingjs{2}$ must be scheduled before the end of~$\tpart{2}$ and hence, directly at its release. By the given compatibilities this is also true for the leftbound dummy jobs of~$\dummyjs{2}$. Therefore, there are exactly two open slots per variable~$x_i$, one reserved for the two corresponding indefinite jobs~$\indefjob{i}{\mtrue}, \indefjob{}{\mfalse}$ and one for the two corresponding leftbound variable jobs~$\falsejob{i}{\lb}, \truejob{i}{\lb}$. Since no further space is left, for both pairs exactly one can be scheduled within~$\tpart{2}$. The remaining one must be completed already by the end of~$\tpart{1}$.

  Also, for the first part, we can conclude that no blocking and no dummy job released in~$\tpart{1}$ can start after the end of~$\tpart{1}$.
  Consider now one variable~$x_i$ and assume that no job corresponding to~$x_i$ can start within part~$\tpart{1}$ after~$6i+5$. This assumption holds obviously for~$x_n$. Then,~$\dummyjob{i}{\rb\mfalse}$ and~$\dummyjob{i}{\lb\mfalse}$, the latest released jobs corresponding to~$x_i$, must both start at their release.

  If the leftbound job~$\truejob{i}{\lb}$ is scheduled within part~$\tpart{1}$ it must be scheduled at its release and hence~$\falsejob{i,1}{\rb}$ and~$\falsejob{i,2}{\rb}$ must be postponed to the next parts. In this case, also the second blocking job~$\blockingjob{i}{\mfalse}$ as well as the first two dummy jobs~$\dummyjob{i}{\rb\mtrue}$  and~$\dummyjob{i}{\lb\mtrue}$ are forced to start at their release, consequently also~$\blockingjob{i}{\mtrue}$. In this case it is not possible anymore to schedule~$\indefjob{i}{\mfalse}$ within part~$\tpart{1}$. For this reason, the counter part~$\indefjob{i}{\mtrue}$ must be scheduled at its release time and the leftbound~$\falsejob{i}{\lb}$ must be postponed. With this, there is exactly one free slot for~$\truejob{i,2}{\rb}$ and one for~$\truejob{i,1}{\rb}$.

  If, on the other hand, the leftbound job~$\truejob{i}{\lb}$ is scheduled after part~$\tpart{1}$, we have to schedule~$\falsejob{i}{\lb}$ within part~$\tpart{1}$. Due to the conflicts with~$\dummyjob{i}{\rb\mfalse}$, the start time of~$\falsejob{i}{\lb}$ and the blocking and dummy jobs in between must in particular be scheduled at their release. For that reason~$\indefjob{i}{\mtrue}$ must be postponed and~$\indefjob{i}{\mfalse}$ must be scheduled at its release. Hence, also the rightbound true jobs~$\truejob{i}{\rb}$ and~$\truejob{i}{\rb}$ must be postponed and there are exactly two slots for the two false jobs.

In both cases, the scheduled leftbound jobs ensure that no earlier released variable job can start after~$6(i-1)+5$. Hence, it can be concluded by induction that, for each variable, either all corresponding false jobs or all corresponding true jobs must be scheduled after part~$\tpart{1}$. And since, by Lemmas \ref{lem:clauses} and \ref{lem:storage}, at least~$2n$ rightbound variable jobs must be scheduled within~$\tpart{1}$ the free spots ensure that exactly the two counter parts are scheduled within~$\tpart{1}$.
\end{proof}

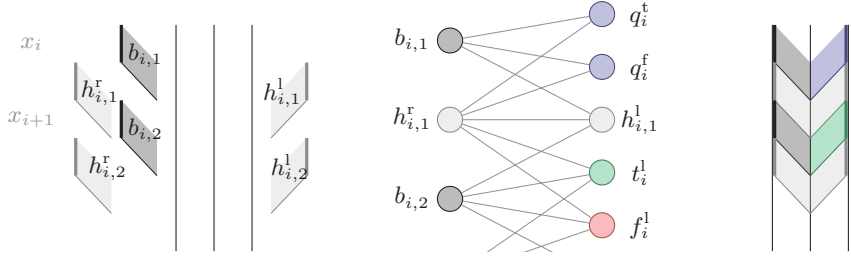
\begin{figure}[t]
  \begin{center}
    \begin{tikzpicture}\label{pic:part2-released}
   \pgftransformxscale{\procheight}
   \pgftransformyscale{-\timeunit}

   \draw  (-1,0) -- +(0,6)  
          (0,0) -- +(0,6)
          (1,0) -- +(0,6)  ;

  \upjob{black}       {-1.5}{0}{1}{1}{}
  \upjob{lightgray}   {-2.7}{1}{1}{1}{}
  \downjob{lightgray} {1.5} {1}{1}{1}{}

  \upjob{black}       {-1.5}{2}{1}{1}{}
  \upjob{lightgray}   {-2.7}{3}{1}{1}{}
  \downjob{lightgray} {1.5} {3}{1}{1}{}

  \begin{scope}[font = \small,anchor=center]
    \node (lbi) at (-1.8, .8){$\blockingjob{i,1}{}$};
    \node (lbip1) at (-1.8, 2.8){$\blockingjob{i,2}{}$};

     \node (lupdi) at (-3., 1.8){$\dummyjob{i,1}{\rb}$};
    \node (lupdip1) at (-2.8, 3.8){$\dummyjob{i,2}{\rb}$};

    \node (ldowndi) at (1.8, 1.8){$\dummyjob{i,1}{\lb}$};
    \node (ldowndip1) at (2., 3.8){$\dummyjob{i,2}{\lb}$};
  \end{scope}
  \begin{scope}[font = \small,color=gray,anchor=center, overlay]
    \node at (-4.8,.5) {$x_i$};
    \node at (-4.8,2.5) {$x_{i+1}$};
  \end{scope}

\end{tikzpicture}
\hspace*{1cm}
\begin{tikzpicture}\label{pic:part2-graph}
  \pgftransformyscale{-.7}

  \begin{scope}
  \clip (-1.5,-.5) rectangle (2,4.5);

  
  \gnode{black}     {b1}{-1}{0.5}{}
  \gnode{blue}      {it}{1}{0}{}
  \gnode{blue}      {if}{1}{1}{}
  \gnode{lightgray} {d1}{-1}{2}{}
  \gnode{lightgray} {d2}{1}{2}{}

  \gnode{black}     {b2}{-1}{3.5}{}
  \gnode{green}      {lt}{1}{3}{}
  \gnode{red}      {lf}{1}{4}{}

  \gnode{lightgray} {d3}{-1}{5}{}
  \gnode{lightgray} {d4}{1}{5}{}

  \draw[gray] (b1) -- (it)
              (b1) -- (if)
              (d1) -- (it)
              (d1) -- (if)
              (b1) -- (d2)
              (d1) -- (d2)
              (b2) -- (d2)
              (b2) -- (lt)
              (b2) -- (lf)
              (d1) -- (lt)
              (d1) -- (lf)
              (d3) -- (lt)
              (d3) -- (lf)
              (b2) -- (d4);
  \end{scope}

  \begin{scope}[font = \small,anchor=center, overlay]
    \node (lbi) at (-1.5, .5){$\blockingjob{i,1}{}$};
    \node (lbip1) at (-1.5, 3.5){$\blockingjob{i,2}{}$};

    \node (lupdi) at (-1.5, 2){$\dummyjob{i,1}{\rb}$};

    \node (ldowndi) at (1.5, 2){$\dummyjob{i,1}{\lb}$};

    \node (li1) at (1.5, 0){$\indefjob{i}{\mtrue}$};
    \node (li2) at (1.5, 1){$\indefjob{i}{\mfalse}$};

    \node (llt) at (1.5, 3){$\truejob{i}{\lb}$};
    \node (llf) at (1.5, 4){$\falsejob{i}{\lb}$};
  \end{scope}
\end{tikzpicture}
\hspace*{1cm}
\begin{tikzpicture}\label{pic:part2-sched}
   \pgftransformxscale{\procheight}
   \pgftransformyscale{-\timeunit}

   \draw  (-1,0) -- +(0,6)  
          (0,0) -- +(0,6)
          (1,0) -- +(0,6)  ;

  \upjob{black}       {0}{0}{1}{1}{}
  \upjob{lightgray}   {0}{1}{1}{1}{}
  \downjob{lightgray} {0}{1}{1}{1}{}
  \downjob{blue}      {0}{0}{1}{1}{}

  \upjob{black}       {0}{2}{1}{1}{}
  \upjob{lightgray}   {0}{3}{1}{1}{}
  \downjob{lightgray} {0}{3}{1}{1}{}
  \downjob{green}      {0}{2}{1}{1}{}

\end{tikzpicture} 
  \end{center}
  \caption{Part $\tpart{2}$ creates a structure of blocking and dummy jobs with respective compatibilities that create space for exactly one indefinite job per variable~$x_i$.}
  \label{fig:part2}
\end{figure}

We can conclude the following claim and hence, Theorem~\ref{thm:graph-hardness}.

\textit{Claim.} There is a satisfying assignment for the given~$\rsat{\leq\!3}{3}$ instance if and only if there is a feasible schedule for the constructed scheduling instance with makespan~$\cmax = \tstart{5}+1$.

\label{prf:graph-hardness}
\begin{proof}[Proof of Theorem~\ref{thm:graph-hardness}]
If there is a schedule with makespan~$\cmax$ we can apply Lemmas~\ref{lem:well-defined-assignment} to \ref{lem:storage}. Within the resulting schedule we can therefore be sure that~$\cln$ rightbound variable jobs are scheduled within the clause part. Since by Lemma~\ref{lem:well-defined-assignment} the assignment of each variable is well defined we get by Lemma~\ref{lem:clauses} a satisfying truth assignment for the clauses.

If on the other hand a satisfying truth assignment is given, the described schedule with demanded makespan can be created in straight-forward manner, by postponing the assignment jobs corresponding to the truth assignment and scheduling all other jobs within the part they are released in (or in part~$\tpart{2}$ in the case of leftbound variable jobs or indefinite jobs).
\end{proof}

\subsection{$\NP$-Hardness of Total Completion Time Minimization}

\corGraphHardnessSum*

We give an analogous reduction as for Theorem~\ref{thm:graph-hardness}. Note, that solutions optimal for the total completion time and those optimal for the total waiting time are equivalent. Hence, it is sufficient to prove the hardness for the latter. The goal is to enforce the same structure as for makespan minimization when minimizing the total waiting time. To do so, we start by calculating an upper bound of  the resulting waiting time.


We can trivially bound the total waiting time of a schedule that achieves a makespan of $\cmax$ by~$W=|\jobs|\cdot\cmax = |\jobs| \cdot (A_5+1)$, where~$\jobs$ is the set of all jobs in our construction.
With this polynomial bound we can extend the construction of a scheduling instance for a given~\rsat{\leq\!3}{3} instance by part~$\tpart{5}$ with~$W+1$ further leftbound blocking jobs $\blockingjs{5}=\{\blockingjob{i}{} \mid i=0, \dots, W-1\}$ with release date~$\tstart{5}+i+1$ for each~$\blockingjob{i}{5}\in\blockingjs{5}$ that are not compatible to any of the previous jobs.

  \textit{Claim.} There is a satisfying truth assignment for the given~\rsat{\leq\!3}{3} instance if and only if there is a feasible schedule for the constructed scheduling instance with total waiting time of at most~$W$.

\label{prf:graph-hardness-sum}
\begin{proof}[Proof of Theorem~\ref{thm:graph-hardness-sum}]
Assume first that there is a satisfying assignment for the \rsat{\leq\!3}{3} instance. In this case, there is a schedule where no job released in the first four parts starts processing after~$\tstart{5}$ and hence the resulting total waiting time does not exceed~$W$. 

Assume on the other hand, that there is a solution for the constructed scheduling instance whose objective does not exceed~$W$. For such a solution, either all jobs released in the first four parts start before~$\tstart{5}$ or their is at least one starting later. In the first case, we get, by Lemmas~\ref{lem:storage} to~\ref{lem:well-defined-assignment}, a schedule together with a satisfying truth assignment with waiting time bounded by~$W$. 

In the second case each postponed job~$j$ with starting time~$\start{j}'$ increases the already existing waiting time by at least an amount of~$(\start{j}'-\tstart{5}) + W+1-(\start{j}'-\tstart{5})=W+1$. Hence, the first case applies.
\end{proof}

\subsection{$\APX$-Hardness}
\label{appendix:apx_hardness}

In this section, we show the $\APX$-hardness of bidirectional scheduling. As for the $\NP$-hardness proof, it is convenient to first prove the $\APX$-hardness for minimizing the makespan before turning to the minimization of the total completion time.
 
\begin{restatable}{theorem}{thmGraphApxHardness}
\label{thm:graph-apx-hardness}
Minimizing the makespan for~$m=1$ with an arbitrary compatibility graph~$G_1$ is $\APX$-hard even if~$\proc{j}=\transit{1}=1$ for each~$j\in\jobs$.
\end{restatable}

\begin{proof}
We reduce from a specific variant of \noun{Max-3-Sat} which is $\NP$-hard to approximate to within a factor of $1016/1015$, see Berman et al.~\cite{BermanKS2003}. An instance of \noun{Symm-4-Occ-Max-3-Sat} is given by a Boolean formula with a set $C$ of clauses of size three over a set of variables $X$, where both the positive and the negative literal of each variable $x_i \in X$ appears in exactly two clauses. Berman et al.~\cite{BermanKS2003} construct a family of instances of \noun{Symm-4-Occ-Max-3-Sat} with $1016n$ clauses, where $n \in \N$. They show that for any $\delta \in (0,1/2)$, it is $\NP$-hard to distinguish between the ``bad'' instances where at most $(1015+\delta)n$ clauses can be satisfied and the ``good'' instances where at least $(1016-\delta)n$ instances can be satisfied.

Let $\phi$ be a formula of the above family. Based on $\phi$, we use the same construction as in Theorem~\ref{thm:graph-hardness} with one small adaption: In the first part, for each variable $x_i$, $i = 0,\dots, |X|-1$, we release additionally two virtual jobs $v_{i,1}$ and $v_{i,2}$ at times $6i$ and $6i+1$, respectively. Both jobs are compatible with all leftbound blocking, dummy and variable jobs of the same variable. We claim that for this bidirectional scheduling instance the optimal makespan is $12|X|+|C| +1 + \tilde{c}$ if and only if the minimum number of unsatisfied clauses of $\phi$ is $\tilde{c}$. Assuming the correctness of the claim, we derive that for a good instance with $1016n$ clauses, the makespan is at most $12|X| + 1016n + 1 + \delta n$. Using the identity $|X| = 3|C|/4 = 3\cdot 1016 n / 4$, we can bound the makespan from above by $(10160+ \delta)n +1$. For bad instances, on the other hand, the makespan is at least $(10161 - \delta)n$, i.e., the optimal makespan cannot be approximated by a factor of $10161/10160 \approx 1.000098$. 

It is left to prove the correctness of the claim. It is easy to see that the optimal makespan is bounded from above by $12|X|+|C| + 1 + \tilde{c}$ by a small adaption of the arguments of the proof of Theorem~\ref{thm:graph-hardness}. To see this, fix a variable assignment satisfying all but $\tilde{c}$ clauses. In parts one and two (where the variable assignments are fixed) we schedule all jobs as in the proof of Theorem~\ref{thm:graph-hardness} with respect to the variable assignment. Additionally, the leftbound variable jobs not scheduled in the first part, leave a gap in the schedule that is a perfect fit for the additional virtual jobs, see also the right illustration in Figure~\ref{fig:part1}. In the third part, we schedule one satisfying variable for each clause that is satisfied. In the forth part, we schedule any $2|X|-|C|$ variable jobs left over from previous parts. By construction, at the end of the forth part, we are left with $\tilde{c}$ variable jobs (that could not be matched to any clause job in the third part). Scheduling them one after another, we obtain the claimed makespan of $12|X| + |C| + 1 + \tilde{c}$.

To see that $12|X|+|C| + 1 + \tilde{c}$ is lower bound on the optimal makespan, we argue using the concept of matched jobs. First, note that there is always an optimal schedule in which all jobs are processed at an integral point in time. Otherwise, we could move the first job scheduled at a non-integral point in time to the previous integral point in time without violating any constraints. Iterating this process, we obtain a schedule in which all jobs are processed at integral times, as claimed. Given such an integral schedule, we call a job processed at time $t$ \emph{matched}, if it is leftbound and there is another rightbound job processed at time $t$, or vice versa. Otherwise the job is called unmatched.

For the following arguments, fix an integral schedule. We proceed to argue that there are at least $\tilde{c}$ unmatched jobs that are mutually incompatible.

First consider the (clause) blocking jobs released in part three. For $k,l \in \{0,1,2\}$, let $X_{k,l}$ be the set of variables $x_i$ such that $k$ rightbound true variable jobs $\truejob{i,1}{\rb}$, $\truejob{i,2}{\rb}$ are matched to a (clause) blocking job and $l$ rightbound false variable jobs $\falsejob{i,1}{\rb}$, $\falsejob{i,2}{\rb}$ are matched to a (clause) blocking job. 
Intuitively, the sets $X_{1,1}, X_{2,1}, X_{1,2}, X_{2,2}$ contain the variables that are not set consistently according to a well-defined truth assignment.
Using that at most $|C|-\tilde{c}$ clauses of $\phi$ can be satisfied, we derive that at least
\begin{align}\label{eq:unmatched_jobs}
\tilde{c} - |X_{1,1}| - |X_{2,1}| - |X_{1,2}| - 2|X_{2,2}|
\end{align}
(clause) blocking jobs (or rightbound dummy jobs) are unmatched.

For any variable $x_i \in X_{2,2}$, the leftbound blocking jobs $\blockingjob{i}{\mtrue}$, $\blockingjob{i}{\mfalse}$, dummy jobs $\dummyjob{i}{\lb\mtrue}$, $\dummyjob{i}{\lb\mfalse}$, indefinite jobs $\indefjob{i}{\mtrue}$, $\indefjob{i}{\mfalse}$, and variable jobs $\falsejob{i}{\lb}$, $\truejob{i}{\rb}$ are matched by at most the two rightbound dummy jobs $\dummyjob{i}{\rb\mtrue}$, $\dummyjob{i}{\rb\mtrue}$ and the two virtual jobs $v_{i,1}$, $v_{i,2}$ released in part 1 as well as the two blocking jobs $\blockingjob{i,1}{}$, $\blockingjob{i,2}{}$ released in part 2, so that in the end, at least two rightbound jobs are left unmatched. Equivalently, for any variable $x_i \in X_{1,2} \cup X_{2,1}$ at least one of the rightbound jobs above is left unmatched.

For any variable $x_i \in X_{1,1}$, consider the leftbound variable jobs $\falsejob{i}{\lb}$ and $\truejob{i}{\lb}$ as well as the leftbound indefinite jobs $\indefjob{i}{\mtrue}$ and $\indefjob{i}{\mfalse}$. At most one indefinite job and one variable job most can be matched with the blocking jobs $\blockingjob{i,1}{}$ and $\blockingjob{i,2}{}$ released in part 2. The other two jobs, say the true variable job $\truejob{i}{\lb}$ and the indefinite job $\indefjob{i}{\mtrue}$, are only compatible with the rightbound variable jobs, the virtual jobs and the dummy jobs jobs, leaving at least one job unmatched. Using \eqref{eq:unmatched_jobs}, we may conclude that the total number of unmatched jobs is at least $\tilde{c}$.

As argued above, the unmatched jobs are either (clause) blocking jobs released in part three or remainders of the different types of leftbound jobs associated with variables and released in the first part. As none of them are compatible, the makespan is at least $12|X| + |C| + 1 + \tilde{c}$, as claimed.
\end{proof}

We are now ready to prove the $\APX$-hardness of the minimization of the total completion time.

\GraphAPXHardnessSum*

\begin{proof}[Sketch]
Let $\phi$ be a formula with $1016n$ clauses for some $n \in \N$ with $\tilde{c}$ unsatisfiable clauses, as in Berman et al.~\cite{BermanKS2003} (cf.\ proof of Theorem~\ref{thm:graph-apx-hardness}). We use a similar idea as in the proof of Theorem~\ref{thm:graph-hardness-sum}, i.e., we use the same construction as in the reduction for the makespan but add an additional set of $M$ leftbound blocking jobs $B_5 = \{\blockingjob{i}{} | i=0,\dots,M-1\}$ with release date $M = 12|X|+|C|+1 = 10160n+1$.
With similar arguments as before, we can show that there is an optimum schedule in which exactly $\tilde{c}$ (clause) jobs are unmatched before time $M$, with only exactly $\tilde{c}$ incompatible variable jobs remaining unscheduled after time~$2M$.
The sum of completion times of this schedule is $a n^2 + bn\tilde{c} + \tilde{c}(\tilde{c}+1)/2 + \mathcal{O}(n)$ for some constants $a,b\in\mathbb{N}$.

Now consider a ``good'' instance with at most $\delta n$ unsatisfiable clauses. The optimum schedule has a sum of completion times of at most
$$ an^2 + b\delta n^2 + n^2\delta^2/2 + \mathcal{O}(n). $$
On the other hand, a ``bad'' instance with at least $(1-\delta)n$ unsatisfiable clauses leads to a sum of completion times of at least
$$ an^2 + b(1-\delta) n^2 + n^2(1-\delta)^2/2 + \mathcal{O}(n).$$
Since, for $n \to \infty$, good and bad instance cannot be distinguished in polynomial time unless $\mathsf{P}=\NP$~(cf.~\cite{BermanKS2003}), no algorithm can approximate the sum of completion times by a factor better than
$$ \frac{a + b(1-\delta) + (1-\delta)^2}{a + b\delta + \delta^2} \underset{\delta \to 0}{\rightarrow} \frac{a+b+1}{a}, $$
which is constant.
\end{proof}

\section{Proofs of Section~\ref{sec:constant_segments}:\newline Dynamic programs for restricted compatibilities}

In this section we present the dynamic programs for a constant number of compatibility types and a constant number of segments.

\singleSegPoly*
\begin{proof}
  Let~$\jobs^1, \dots, \jobs^{\ctnum}$ be a partition into subsets of invariant compatibility type.
  We consider each subset~$\jobs^c$ ordered non-increasingly by release dates and denote by~$J_i^c$ the%
  ~$i$-th job of~$J^c$ in this order,
  i.e., the $(n_c-i)$-th job to be released.
  Each entry~$T[i_1,t_1, \dots, i_{\ctnum}, t_{\ctnum}; c]$ of our dynamic programming table is designed to hold the minimum sum of completion times that can be achieved when scheduling only the $i_{c'}$ jobs of largest release date of each compatibility type~$c'$, such that $J_{i_{c'}}^{c'}$ is not scheduled before time~$t_{c'}$ and~$J_{i_c}^c$ is the first job that is scheduled.
  We start by setting~$T[0,t_1, \dots, 0, t_{\ctnum}; c]=0$ and define the dependencies between table entries in the following.
  
Let~$C(j,t)=\max\{t,r_j\}+p+\tau_1$ denote the smallest possible completion time of job~$j$ when scheduling it not before~$t$.
  Depending on the types of jobs~$j_1,j_2$ (and in particular of their directions), we can compute in constant time the earliest time~$\theta(j_1,t_1,j_2,t_2)$ not before~$t_1$ that job~$j_1$ can be scheduled at, assuming that~$j_2$ is scheduled earlier at time~$\max\{t_2,r_{j_2}\}$.
  We let~$\delta_{cc'}=1$ if~$c=c'$ and~$\delta_{cc'}=0$ otherwise, abbreviate~$\theta_{c'}=\theta(J_{i_{c'}}^{c'},t_{c'},J_{i_c}^c,t_c)$, and get the following recursive formula for~$i_c>0$:
  \[
  T[i_1,t_1,\dots, i_{\ctnum},t_{\ctnum};c] = 
  \min_{c':i_{c'} \neq 0} \lbrace 
  T[i_1-\delta_{1c}, \theta_1,
  \dots,       
  i_\ctnum-\delta_{\ctnum c}, \theta_\ctnum; c'] + C(J_{i_c}^c,t_c) \rbrace.
  \]
  
  We can fill out our table in order of increasing sums~$\sum i_c$ and finally obtain the desired minimum completion time as~$\min_c T[n_1,0,\dots,n_\ctnum,0;c]$. We can reconstruct the schedule from the dynamic programming table in straightforward manner. It remains to argue that we only need to consider polynomially many times~$t_c$. This is true, since all relevant times are contained in the set~$\{\rel{j} + k\transit{} + \ell\proc{} \mid j, k, \ell \leq n\}$ of cardinality~$\O(n^3)$.
\end{proof}

\begin{restatable}{theorem}{thmIdenticalConstantSegMakespan}\label{thm:identical-constantseg-makespan}
The bidirectional scheduling problem can be solved in polynomial time if~$m$,~$\ctnum$, and $\transit{i}$ are constant for each~$i\in M$, and~$\proc{j}=1$ for each~$j\in\jobs$.
\end{restatable}
\begin{proof}
  Again, we consider subsets of identical jobs. In addition to their conflict type~$c$, we further distinguish jobs by their start and target segments~$s,t$ and form subsets~$J_{s,t}^c$ correspondingly. The number of subsets is bounded by~$\ctnum m^2$.
  Since all release times are integer and since~$p_j=1$, we only need to consider integer points in time. Hence, only~$\tau_i+1$ possible positions need to be considered for a job running on segment~$i$, and no two jobs of the same direction can occupy the same position.
  The state of the system can be fully described by (i) the number of available jobs per segment and~$J_{s,t}^c$, and (ii)~for each position on each segment and each~$J_{s,t}^c$, the fact whether a job of~$J_{s,t}^c$ is occupying this position.
  The number of states is bounded by $\prod_{i=1}^m n^{\ctnum{}m^{2}}\cdot\prod_{i=1}^m 2^{\ctnum{}m^2(\transit{i}+1)}=\mathrm{poly}(n)$.
  
  We define the successors of each state to be all states that can be reached in one time step where not all jobs wait, or by waiting for the next release date.
  This way, the state representation changes from one state to the next.
  The system always makes progress towards the final state where each job has arrived at its target.
  The state graph can thus not have a cycle, and we may consider states in a topological order.
  We formulate a dynamic program that computes for each state the smallest partial completion time to reach the state, where the partial completion time is defined as the sum of completion times of all completed jobs plus 
the current time for each uncompleted job.
  The dynamic program is well-defined as each value only depends on predecessor states.
\end{proof}

\begin{restatable}{corollary}{corIdenticalConstantSegTotalC}\label{cor:identical-constantseg-totalc}
The bidirectional scheduling problem can be solved in polynomial time if $m$ and $\ctnum$ are constant, $\transit{i}=1$ for each~$i\in M$, and~$\proc{j}=0$ for each~$j\in\jobs$.
\end{restatable}
\begin{proof}
  Since all release dates are integer, at each integer point in time no jobs are running on any segment. We can thus use a simpler version of the dynamic program we introduced in the proof of Theorem~\ref{thm:identical-constantseg-makespan}.
\end{proof}

\end{document}